\documentclass[11pt]{article}
\usepackage[utf8]{inputenc}
\usepackage[T1]{fontenc}
\usepackage{amsthm, amsmath}
\usepackage{amsfonts}
\usepackage[margin=1in]{geometry}
\usepackage[algo2e,ruled,noend,resetcount,linesnumbered]{algorithm2e}
\usepackage{graphicx}
\usepackage{comment}
\usepackage{enumitem}
\usepackage{thmtools}
\usepackage{thm-restate}
\usepackage{algorithm, algorithmicx, algpseudocode} 
\usepackage[draft,margin=false,inline=true]{fixme}
\fxsetup{mode=multiuser,theme=color, layout=inline}
\FXRegisterAuthor{bs}{abs}{\color{blue} {\bf Barna}}

\setlist[enumerate]{nosep, topsep=1ex}
\setlist[itemize]{nosep, topsep=1ex}
\setlist[description]{nosep}
\allowdisplaybreaks
\usepackage[colorlinks]{hyperref}
\usepackage{cleveref}
\usepackage{xspace}

\usepackage{caption}
\usepackage{subcaption}

\newcommand{\alg}[1]{\textsc{\bfseries \footnotesize #1}}

\newcommand{\bigO}[1]{O \left( #1 \right)}
\newcommand{\tO}[1]{\tilde{O} (#1)}
\newcommand{\bigtO}[1]{\tilde{O} \left( #1 \right)}

\newcommand{\E}{\mathbb{E}}

\newtheorem{theorem}{Theorem}[section]
\newtheorem{corollary}[theorem]{Corollary}
\newtheorem{definition}[theorem]{Definition}
\newtheorem{lemma}[theorem]{Lemma}
\newtheorem{proposition}[theorem]{Proposition}

\newtheorem{assumption}[theorem]{Assumption}

\newcommand{\Z}{\mathbb{Z}}
\newcommand{\R}{\mathbb{R}}

\newcommand{\family}{\mathcal{F}}

\newcommand{\eps}{\varepsilon}

\newcommand{\given}{\textrm{\xspace s.t. \xspace}}
\newcommand{\andT}{\textrm{\xspace and \xspace}}
\newcommand{\otherwise}{\textrm{\xspace o/w \xspace}}
\newcommand{\ceil}[1]{\left\lceil #1 \right\rceil}
\newcommand{\floor}[1]{\left\lfloor #1 \right\rfloor}
\newcommand{\set}[1]{\{#1\}}

\newcommand{\wt}{\mathsf{wt}}

\newcommand{\rank}{\mathrm{rank}}

\newcommand{\twoApproxAPASP}{\alg{2ApproxAPSP}}
\newcommand{\randomApproxAPASP}{\alg{Random2ApproxAPSP}}
\newcommand{\additiveAPASP}{\alg{AdditiveAPASP}}
\newcommand{\dominatingAPASP}{\alg{DominatingSetAPASP}}
\newcommand{\twoedgeAPASP}{\alg{Distance2APASP}}
\newcommand{\alphabetaAPASP}{\alg{MultAddAPASP}}
\newcommand{\boundedAdditiveAPASP}{\alg{BoundedAdditiveAPASP}}
\newcommand{\shortAdditiveAPASP}{\alg{ShortAdditiveAPASP}}

\newcommand{\longMultiplicativeAPASP}{\alg{LongMultiplicativeAPASP}}
\newcommand{\longMultiplicativeAPASPFMM}{\alg{FMMLongMultiplicativeAPASP}}

\newcommand{\APSP}{\alg{APSP}}
\newcommand{\sparseAPASP}{\alg{SparseAPASP}}
\newcommand{\denseAPASP}{\alg{DenseAPASP}}

\newcommand{\baswanaAPASP}{\alg{(2,1)-APASP}}
\newcommand{\bkAPASP}{\alg{BK2APASP}}

\newcommand{\hittingSet}{\alg{HittingSet}}
\newcommand{\rHittingSet}{\alg{rHittingSet}}
\newcommand{\dominate}{\alg{Dominate}}
\newcommand{\rDominate}{\alg{rDominate}}
\newcommand{\decompose}{\alg{Decompose}}
\newcommand{\rDecompose}{\alg{rDecompose}}
\newcommand{\wDecompose}{\alg{WeightedDecompose}}

\newcommand{\bfs}{\alg{BFS}}

\newcommand{\dijkstra}{\alg{Dijkstra}}

\newcommand{\boundedMinPlus}{\alg{BoundedMinPlus}}
\newcommand{\monotoneMinPlus}{\alg{MonotoneMinPlus}}
\newcommand{\approximateMinPlus}{\alg{ApproximateMinPlus}}
\newcommand{\bdToMonotone}{\alg{BDtoMonotone}}

\newcommand{\minplus}{$(\min, +)$\xspace}

\newcommand{\edgeset}{\mathcal{E}}

\newcommand{\level}{\ell}

\newcommand{\heavy}{h}
\newcommand{\kHeavy}{\overline{h}}

\newcommand{\multapproxlimit}{k}

\title{Faster Approximate All Pairs Shortest Paths}
\author{Barna Saha\thanks{University of California, San Diego. The authors are partially supported by NSF grants 1652303, 1909046, 2112533, and HDR TRIPODS Phase II grant 2217058.} \and Christopher Ye\footnotemark[1]}
\date{}

\begin{document}

\vspace{-20pt}
\maketitle
\vspace{-30pt}

\setcounter{page}{0}
\thispagestyle{empty}

\begin{abstract}
    
The all pairs shortest path problem (APSP) is one of the foundational problems in computer science. 
For weighted dense graphs on $n$ vertices, no truly sub-cubic algorithms exist to compute APSP exactly even for undirected graphs. This is popularly known as the APSP conjecture and has played a prominent role in developing the field of fine-grained complexity. 
The seminal results of Seidel and Zwick show that using fast matrix multiplication (FMM) it is possible to compute APSP on unweighted undirected graphs exactly in $\tO{n^{\omega}}$ time, and can be approximated within $(1+\epsilon)$ factor in weighted undirected graphs in time $\tO{n^{\omega}}$ respectively. 
Here $\omega$ is the exponent of FMM, which currently stands at $\omega=2.37188$. 
Moreover even for unweighted undirected graphs, it is not possible to obtain a $(2-\epsilon)$-multiplicative approximation of APSP for any $\epsilon >0$ in $o(n^\omega)$ time. 
Since 2000, a result by Dor, Halperin, and Zwick gave the best $2$ approximation algorithm for APSP in unweighted undirected graphs in time $\tO{n^{7/3}}$. 
This result was recently improved by Deng, Kirkpatrick, Rong, Williams and Zhong to $\tO{n^{2.2593}}$ using fast min-plus product for bounded-difference matrices which uses FMM as a subroutine (the stated bound here uses new results for computing such min-plus products by Durr). 
In fact both these results obtain a $+2$-additive approximation. 
Recently, Roditty (STOC, 2023) improved the previous bounds for multiplicative $2$-approximation of APSP in unweighted undirected graphs giving the best known bound of $\tO{n^{2.25}}$. All these algorithms are deterministic. 
Roditty also considers estimating shortest paths for all paths of length $\geq k$ for $k \geq 4$, and gives improved bounds when the underlying graph is sparse using randomization. 
Though for dense graphs, the best known bounds still remained at those provided by Dor et al. more than two decades back.

In this paper, we provide a multitude of new results for multiplicative and additive approximations of APSP in undirected graphs for both unweighted and weighted cases. 
We provide new algorithms for multiplicative 2-approximation of unweighted graphs: a deterministic one that runs in $\tO{n^{2.072}}$ time and a randomized one that runs in $\tO{n^{2.0318}}$ on expectation improving upon the best known bound of $\tO{n^{2.25}}$. 
The algorithm uses FMM as well as new combinatorial insights. 
For $2$-approximating paths of length $\geq k$, $k \geq 4$, we provide the first improvement after Dor et al. for dense graphs even just using combinatorial methods, and then improve it further using FMM.
We next consider additive approximations, and provide improved bounds for all additive $\beta$-approximations, $\beta \geq 4$. 
For example, we achieve a running time of $\tO{n^{2.155}}$ for $+4$ additive approximation improving over the previously known bound of $\tO{n^{2.2}}$, and for a $+6$ additive approximation, our algorithm has a running time of $\tO{n^{2.103}}$ as opposed to the $\tO{n^{2.125}}$ time that was previously known. 
For weighted graphs, we show that by allowing small additive errors along with an $(1+\epsilon)$-multiplicative approximation, it is possible to improve upon Zwick's $\tO{n^\omega}$ algorithm. 
For example, it is possible to obtain a bi-criteria $(1+\epsilon, 2w_{u,v})$ approximation in $\tO{n^{2.152}}$ time for the shortest path distance between all vertex pairs $u,v$ where $w_{u,v}$ is the highest weight edge on the $u$-$v$ shortest path. 
Additionally, we provide a landscape of such bi-criteria approximations for weighted and unweighted graphs. 
Our results point out the crucial role that FMM can play even on approximating APSP on unweighted undirected graphs, and reveal new bottlenecks towards achieving a quadratic running time to approximate APSP.

\newpage

\end{abstract}

\newpage
\tableofcontents

\newpage
\setcounter{page}{1}
\section{Introduction} 
\label{sec:intro}

Computing All Pairs Shortest Path (APSP) on graphs is a landmark problem in computer science. 
It is both one of the foundational problems of fine grained complexity, as well as one that directly or indirectly aids in the computation of many important graph and matrix problems. 
A large variety of graph and matrix problems can be fine-grained reduced to either unweighted or weighted APSP showing that a better algorithm for APSP will lead to a better algorithm for all those problems \cite{williams2010subcubic}. 
The classic approaches like Floyd-Warshall computes APSP on weighted dense graphs, $G=(V,E), |V|=n, |E|=m$, in $O(n^3)$ time; whereas the best result known by Williams has a running time of $O\left( \frac{n^3}{2^{\sqrt{\log n}} }\right)$ \cite{williams2014faster} which is still not sub-cubic. 
Indeed the weighted APSP conjecture states that there does not exist any truly subcubic algorithm, that is one running in $O(n^{3 - \varepsilon})$ time for some constant $\varepsilon > 0$. 
However, for unweighted undirected graphs, a seminal result of Seidel showed that APSP can be computed in $O(n^{\omega})$ time \cite{seidel1995all} where $\omega$ is the exponent of fast matrix multiplication, and currently stands at $\omega=2.37188$ \cite{duan2022matrixmult}. More generally, when weights are bounded integers in the range $[-M,M]$, APSP can be solved in $O(Mn^\omega)$ time for undirected graphs \cite{shoshan_zwick, alon1997exponent} and in subcubic time in directed graphs \cite{zwick2002bridgingsets}.

In this paper, we concentrate on undirected graphs, and henceforth all references to graphs indicate undirected graphs if not explicitly mentioned otherwise. 
Interest in computing APSP has naturally led to the study of approximation algorithms. 
An estimate $\hat{\delta}: V \times V \rightarrow \R$ is an $(\alpha, \beta)$-approximation of the actual shortest path metric $\delta: V \times V \rightarrow \mathbb{R}$ if $\delta(u,v) \leq \hat{\delta}(u, v) \leq \alpha \delta(u,v)+\beta $ for all $u,v \in V \times V$. Therefore an $(\alpha,0)$-approximation implies a pure multiplicative approximation whereas a $(1, \beta)$-approximation implies a pure additive approximation.
There is a huge body of literature on approximating APSP, from multiplicative and additive approximation in sub-cubic time \cite{aingworth1999fast, dor2000apasp, cohen2001smallstretch, baswana2010fasterapasp, berman2007approxapsp, deng2022apasp, roditty2023newapasp}, computing distance oracle that trades off preprocessing with query time, to developing space efficient data structures  
 \cite{mendel2007ramsey, baswana2010fasterapasp, patrascu2010distanceoracles, wulff2012approximateoracles, agarwal2013stretch2oracle, chechik2014approximate, chechik2015oracles, elkin2016linearoracle, elkin2016spaceoracle, sommer2016distanceoracle,akav2020almost2, chechik2022nearly2}.
 In a seminal work, Zwick gave an $((1+\epsilon),0)$-approximation algorithm for weighted APSP that runs in $\tilde{O}(\frac{n^{\omega}}{\epsilon} \log{W})$ time \cite{zwick_apsp_weighted_directed} where $W$ is the largest edge weight. Dependency on $W$ was later removed to obtain a strong polynomial running time of $O\left(\frac{n^\omega}{\epsilon} \text{polylog}\left(\frac{n}{\epsilon}\right)\right)$ \cite{bringmann2019noscaling}. Moreover, even for unweighted graphs, a better than $(2,0)$-approximation in $o(n^{\omega})$ time is not possible \cite{dor2000apasp}. Naturally this leads to the question whether a $(2,0)$ approximation is possible in $o(n^\omega)$ time and even better in $O(n^2)$ time. While designing an $\tilde{O}(n^2)$ time  algorithm for a $(2,0)$-approximation still remains open, it is possible to get a $(3,0)$-approximation in $\tilde{O}(n^2)$ time \cite{dor2000apasp,cohen2001smallstretch}.

 Let us first consider unweighted graphs. So far, the best running time to achieve a $(2,0)$-approximation is due to Roditty \cite{roditty2023newapasp}. Roditty gave an algorithm with a running time of $\tilde{O}(n^{2.25})$ for a $(2,0)$ approximation which improves upon the $\tilde{O}(n^{2.2593})$ running time previously known \cite{deng2022apasp, durr2023rect_monotone}.  In fact, both these results are based on Dor et al.'s work \cite{dor2000apasp}. Dor, Halperin and Zwick gave an algorithm to achieve a $(1,2)$-approximation that runs in $\tilde{O}(n^{7/3})$ time and a $(1,4)$-approximation that runs in $\tilde{O}(n^{9/4})$ time among other results. Roditty utilizes the $\tilde{O}(n^{9/4})$ time algorithm and brings in new ideas to show that on paths of length $3$, it is possible to get a $+2$-additive approximation. Moreover, paths of length $1$ can trivially be found exactly in $O(m)$ time, and paths of length $2$ can be approximated within $+2$-additive errors from Dor et al.'s work \cite{dor2000apasp}. All, these together imply a $(2,0)$-approximation in $\tilde{O}(n^{9/4})$ time by Roditty \cite{roditty2023newapasp}. On the other hand, Deng et al. showed the first step of the $\tilde{O}(n^{7/3})$ time algorithm of Dor et al. can be made faster by utilizing fast algorithms for bounded-difference $(\min, +)$ product \cite{chi2022monotone, chi2022boundeddifference, williams2020truly, bringmann2019boundeddifference}, therefore essentially giving a faster $(1,2)$-approximation algorithm. These lead to several interesting open questions.

 \begin{center}{\it Can we use algebraic methods to get a faster $(2,0)$-approximation?}
 \end{center}
 
 Using bounded-difference $(\min, +)$ product in the first step of $\tilde{O}(n^{9/4})$ does not help, as that running time itself is quite large. Dor et al. provided an entire trade-off between running time and additive error. They showed for every even $\beta$, it is possible to approximate APSP within $+\beta$ additive error in time $\tilde{O}(\min(n^{2-\frac{2}{\beta+2}}m^{\frac{2}{\beta+2}}, n^{2+\frac{2}{3\beta-2}}))$. While Deng et al.'s work \cite{deng2022apasp} improved the running time for a $(1,2)$-approximation, it left open the scope of improving the running time for algorithms that allow higher additive errors. The best known bounds for those still stand at where they were more than two decades back. Employing the bounded-difference $(\min,+)$-product as the first step in Dor et al.'s algorithm for higher additive errors provide no improvements.

  \begin{center}{\it Can we use algebraic methods to get faster $(1,\beta)$-approximation for all $\beta >0$?}
 \end{center}

 As will become apparent, the challenge in computing a fast $(2,0)$-approximation lies in handling paths of short lengths, for which a multiplicative $2$-approximation implies a good additive approximation. Another interesting contribution of Roditty's work \cite{roditty2023newapasp} is to provide an improved running time when a $(2,0)$-approximation is sought only for path lengths greater than a certain threshold. In particular, they show that a $(2,0)$-approximation can be obtained for vertex pairs at distance at least $k$ in time $\tilde{O}\left(\min{\left(n^{2-\frac{2}{k+4}}m^{\frac{2}{k+4}}, n^{2+\frac{2}{3k-2}}\right)}\right)$. This improves upon the previous bound of  $\tilde{O}\left(\min{\left(n^{2-\frac{2}{k+2}}m^{\frac{2}{k+2}}, n^{2+\frac{2}{3k-2}}\right)}\right)$ by Dor et al. for sparse graphs while leaving the same bounds for dense graphs. Clearly, this raises the question whether it is possible to get  improved bounds for dense graphs.

 \begin{center}{\it Can we get a faster $(2,0)$-approximation for vertex pairs at distance at least $k$ for dense graphs?}
 \end{center}

 Moving to weighted graphs, as stated before a $((1+\epsilon),0)$-approximation, $\epsilon >0$, is possible in $O\left(\frac{n^\omega}{\epsilon} \text{polylog}\left(\frac{n}{\epsilon}\right)\right)$ time \cite{zwick_apsp_weighted_directed, bringmann2019noscaling}. Multiple works have studied a natural question whether a bi-criteria $(\alpha, \beta)$-approximation, $\alpha >0, \beta >0$, can have better time complexity \cite{baswana2010fasterapasp,baswana2005nearly2approxapasp,berman2007approxapsp,elkin2005asp}. 
 Baswana and Kavitha \cite{baswana2010fasterapasp} and Berman and Kasiviswanathan \cite{berman2007approxapsp} obtained $(2, w_{u, v})$-approximations in $\tO{n^2}$ time where $w_{u, v}$ is the largest weight in the shortest path between $u, v$. Berman and Kasiviswanathan also showed an $(1+\epsilon, 2w(u,v))$-approximation with a running time of $\tO{\frac{n^\omega}{\epsilon^3}\log{\frac{n}{\epsilon}}}$ \cite{berman2007approxapsp}. This later result improved upon a prior work of Elkin where a $(1 + \eps, M \beta(\eps, \rho, \zeta))$-approximation is obtained in time $O(m n^{\rho} + n^{2 + \zeta})$ with $M$ being the ratio between the heaviest and lightest edge in the graph \cite{elkin2005asp}.
The constant $\beta(\eps, \rho, \zeta)$ depends on $\zeta$ as $(1/\zeta)^{\log 1/\zeta}$, inverse exponentially on $\rho$, and inverse polynomially on $\eps$. These lead to an interesting question, if we fix $\alpha=(1+\epsilon)$, can we show a running time trade-offs with varying $\beta$? A similar question applies for all $\alpha < 2$.

\begin{center}{\it Can we get a faster $(\alpha,\beta)$-approximation for weighted graphs where $\alpha=(1+\epsilon)$ and $\beta >0$?}
 \end{center}

\subsection{Our Contributions}

In this paper, we provide multitude of results on approximation APSP on undirected graphs answering all of the above questions. We start with our contributions on unweighted graphs.

\subsubsection{Multiplicative Approximation on Unweighted Graphs (Section \ref{sec:mult-approx-bounded})}

We significantly improve upon the current best known bound of $\tilde{O}(n^{2.25})$ \cite{roditty2023newapasp} for a $(2,0)$ approximation of unweighted APSP. Specifically, we obtain the following theorems. 
Theorem \ref{thm:rand-2-approx} is also obtained in a concurrent work by Dory, Forster, Kirkpatrick, Nazari, Vassilevska Williams, and de Vos \cite{DBLP:journals/corr/DoryFKNWV}.

\begin{restatable*}{theorem}{randtwomultapprox}
    Let $G$ be an undirected, unweighted graph with $n$ vertices.
    \Cref{alg:random-2-approx-apasp} computes a $(2, 0)$-approximate APSP solution in expected time $\bigtO{n^{2.03184039}}$.
    \label{thm:rand-2-approx}
\end{restatable*}
\begin{restatable*}{theorem}{dettwomultapprox}
    Let $G$ be an undirected, unweighted graph with $n$ vertices.
    \Cref{alg:2-approx-apasp} deterministically computes a $(2, 0)$-approximate APSP solution in time $\tO{n^{2.07203166}}$.
    \label{thm:2-approx-apsp}
\end{restatable*}

Our results use FMM with the current best known bounds \cite{duan2022matrixmult, gall2018rectangularmm} and several new combinatorial insights to bring down the running time very close to $O(n^2)$. We also observe that a $(7/3,0)$-approximation on unweighted graphs can be computed in $\tilde{O}(n^2)$ time (see Appendix~\ref{sec:7/3-det-approx}).

\subsubsection{Multiplicative Approximations for Long Paths~(Section \ref{sec:bounded-additive-approx})}

We improve the bounds for a $(2,0)$-approximation on paths of length at least $k$, for all $k \geq 4$ on dense graphs even just using combinatorial techniques and then further using algebraic methods. Combining with Roditty's results \cite{roditty2023newapasp}, these imply an improvement for all cases (sparse and dense) over Dor et al.'s result \cite{dor2000apasp} for paths of length at least $4$.

Below, we state the combinatorial results and the further improvements using FMM are stated in Section 5.

\begin{table}[H]
    \begin{center}
        \begin{tabular}{ |c|c|c|c|c| } 
             \hline
             \multicolumn{4}{|c|}{$(2, 0)$-Multiplicative Approximation for $\delta(u, v) \geq \multapproxlimit$} \\
             \hline
             $\multapproxlimit$ & \cite{dor2000apasp}, \cite{roditty2023newapasp} & \Cref{alg:mult-approx-bk} (Combinatorial) & \Cref{alg:mult-approx-bk-fmm} (uses FMM) \\
             \hline
             $4$ & $n^{11/5} = n^{2.200}$ & $n^{15/7} = n^{2.1429}$ (\Cref{alg:short-additive-apasp}) & $n^{2.01973523}$ \\
             $6$ & $n^{17/8} = n^{2.125}$ & $n^{21/10} = n^{2.1000}$ & $n^{2.01084688}$ \\
             $8$ & $n^{23/11} = n^{2.091}$ & $n^{29/14} = n^{2.0715}$ & $n^{2.00745825}$ \\
             $10$ & $n^{29/14} = n^{2.072}$ & $n^{37/18} = n^{2.0556}$ & $n^{2.00573823}$ \\
             $12$ & $n^{35/17} = n^{2.059}$ & $n^{45/22} = n^{2.0455}$ & $n^{2.00462679}$ \\
             \hline
        \end{tabular}
    \end{center}
    \caption{Improvements in computing $(2, 0)$-approximate APSP for $\delta(u, v) \geq \multapproxlimit$ on undirected, unweighted graphs with $n$ vertices and $m$ edges.
    For $\multapproxlimit = 4$, \Cref{alg:short-additive-apasp} is more efficient than \Cref{alg:mult-approx-bk}.
    While only a few examples are shown above, we obtain improvements for all $\multapproxlimit$.
    See \Cref{prop:mult-approx-bk-fmm-examples} for the derivation of some running times for \Cref{alg:mult-approx-bk-fmm}.
    }
    \label{tbl:2-approx-geq-beta}
\end{table}

\begin{theorem} [Stated as Corollary~\ref{cor:2-approx-d-geq-k}]

    Let $\multapproxlimit \geq 4$ be an even integer.
    Then, we can compute a $(2, 0)$-approximation for distances $\delta(u, v) \geq \multapproxlimit$ combinatorially in expected time
    \begin{equation*}
        \bigtO{\min \left(n^{2 - \frac{2}{\multapproxlimit + 4}} m^{\frac{2}{\multapproxlimit + 4}}, n^{2 + \frac{1}{2(\multapproxlimit - 1)}}, n^{2 + \frac{2}{3\multapproxlimit + 2}}\right)}
    \end{equation*}
    In particular, we output $\hat{\delta}$ such that $\delta(u, v) \leq \hat{\delta}(u, v)$ for all $u, v$ and $\hat{\delta}(u, v) \leq 2 \delta(u, v)$ whenever $\delta(u, v) \geq \multapproxlimit$.
\end{theorem}
In contrast, Roditty achieves a bound of  
\begin{equation*}
    \bigtO{\min \left(n^{2 - \frac{2}{k + 4}} m^{\frac{2}{k + 4}}, n^{2 + \frac{2}{3 \multapproxlimit - 2}}\right)}
\end{equation*} (Corollary 2.6 of Roditty \cite{roditty2023newapasp}).

\Cref{tbl:2-approx-geq-beta} illustrates the comparison in running time between our algorithms and the previous results of Roditty \cite{roditty2023newapasp} and Dor et al. \cite{dor2000apasp} for dense graphs.

\subsubsection{Additive Approximation (Section \ref{sec:monotone-min-plus})}

We show it is possible to get better additive approximations $(1,\beta)$ for all $\beta > 0$. 
Previously such a result was known only for $\beta=2$ \cite{deng2022apasp}.

\begin{restatable*}{theorem}{genadditiveapasp}
    Let $\beta \geq 4$ be an even integer.
    Let $G$ be an undirected, unweighted graph with $n$ vertices.
    \Cref{alg:k-additive-apasp} computes $\hat{\delta}$  such that $\delta(u, v) \leq \hat{\delta}(u, v) \leq \delta(u, v) + \beta$ for all $u, v \in V$ in time,
       $\bigtO{n^{2 + \frac{2x}{\beta + 2}}}$
    where $x$ is the solution to,
       $ \omega\left(1 - \frac{\beta - 2}{\beta + 2} x, 1 - x, 1 - \frac{\beta - 4}{\beta + 2} x\right) = 1 + \frac{4 + 2 \beta}{\beta + 2} x.$\footnote{$\omega(a, b, c)$ is the minimum value such that the product of a $\ceil{n^a} \times \ceil{n^b}$ matrix by a $\ceil{n^b} \times \ceil{n^c}$ matrix can be computed in $O(n^{\omega(a, b, c) + \eps})$ arithmetic operations for any constant $\eps > 0$. Note $\omega = \omega(1, 1, 1)$.}
    \label{thm:gen-additive-apasp}
\end{restatable*}

\Cref{tbl:k-additive-apasp} shows the improvement in running time for various $(1,\beta)$-approximation errors beyond the work of Dor, Halperin, and Zwick \cite{dor2000apasp}. 
The concurrent work \cite{DBLP:journals/corr/DoryFKNWV} obtains a $(1 + \eps, \beta)$-approximation faster than \cite{dor2000apasp} for $\beta \leq 9$. 
We obtain $(1, \beta)$-approximation improving upon \cite{dor2000apasp} for all $\beta$ and do so without incurring any multiplicative error. 
Our algorithm is also faster than the $(1 + \eps, \beta)$-approximation of \cite{DBLP:journals/corr/DoryFKNWV} for $\beta \geq 8$.
Due to the additional multiplicative error, we instead compare \cite{DBLP:journals/corr/DoryFKNWV} with Algorithm \ref{alg:k-weighted-additive-apasp} in Table \ref{tbl:weighted-k-additive-apasp}.

\begin{table}[H]
    \begin{center}
        \begin{tabular}{ |c|c|c| } 
             \hline
             \multicolumn{3}{|c|}{$+\beta$-Additive Approximation} \\
             \hline
             $\beta$ & \cite{dor2000apasp} & Algorithm \ref{alg:k-additive-apasp} \\ 
             \hline
             $4$ & $n^{11/5} = n^{2.2}$ & $n^{2.15506251}$ \\
             $6$ & $n^{17/8} = n^{2.125}$ & $n^{2.10300405}$ \\
             $8$ & $n^{23/11} = n^{2.0909}$  & $n^{2.07733373}$ \\
             $10$ & $n^{29/14} = n^{2.0715}$ & $n^{2.06196791}$ \\
             \hline
        \end{tabular}
    \end{center}
    \caption{Improvements in computing $+\beta$ approximation on undirected, unweighted graphs with $n$ vertices.
    While a few examples are shown above, we obtain improvements for all $\beta \geq 4$.
    All running times are computed with \cite{Complexity}.}
    \label{tbl:k-additive-apasp}
\end{table}

\subsubsection{Weighted Graphs \& \texorpdfstring{$(\alpha, \beta)$}{(alpha, beta)}-approximations (Section \ref{sec:weighted-additive-apasp} \& Section \ref{sec:sub-2-mult-approx})}
We show new results that allow for $(1+\epsilon)$-multiplicative approximation of APSP, and some additive errors to go significantly below $n^\omega$ for weighted graphs. We also show interesting new trade-offs between $\alpha < 2$, and $\beta$ for both unweighted and weighted graphs (see \Cref{thm:sub-2-mult-add-approx} and \Cref{thm:mult-add-weighted-approx}).

\begin{table}[H]
    \begin{center}
        \begin{tabular}{ |c|c|c| } 
             \hline
             \multicolumn{3}{|c|}{Approximation Algorithms on Weighted Graphs} \\
             \hline
             Work & Approximation Factor & Time \\ 
             \hline
             \cite{zwick2002bridgingsets} & $(1 + \eps, 0)$ & $n^{\omega}$ \\
             \hline
             \cite{elkin2005asp} & $(1 + \eps, M \beta(\zeta, \rho, \eps))$ & $m n^{\rho} + n^{2 + \zeta}$ \\
             \hline
             \cite{berman2007approxapsp} & $(1 + \eps, 2 w_{u, v})$ & $n^{2.24}$ \\
             \Cref{alg:2-weighted-additive-apasp} & $(1 + \eps, 2 w_{u, v})$ & $n^{2.1519}$ \\
             \Cref{alg:k-weighted-additive-apasp} & $(1 + \eps, 2 w_{u, v}(\beta))$ & $n^{2 + x/(\beta + 1)}$ \\
             \hline
        \end{tabular}
    \end{center}
    \caption{Comparison with previous results for $(\alpha,\beta)$-approximations on weighted graphs.
    $M$ denotes the ratio between the heaviest and lightest edge in the graph $G$.
    $w_{u, v}(\beta)$ denotes the weight of the $\beta$ heaviest edges on the shortest path between $u, v$ and $w_{u, v} = w_{u, v}(1)$.}
    \label{tbl:2-weighted-additive-apasp}
\end{table}

\begin{restatable*}{theorem}{weightedgenadditiveapasp}
    Let $\beta \geq 2$ be an integer and $\eps > 0$.
    Let $G$ be an undirected, unweighted graph with $n$ vertices.
    \Cref{alg:k-weighted-additive-apasp} computes $\hat{\delta}$ such that $\delta(u, v) \leq \hat{\delta}(u, v) \leq (1 + \eps) \delta(u, v) + 2 w_{u, v}(\beta)$ in time
       $\bigtO{\frac{n^{2 + \frac{x}{\beta + 1}}}{\eps}}$
    where $x$ is the solution to,
      $  \omega\left(1 - \frac{\beta - 1}{\beta + 1} x, 1 - x, 1 - \frac{\beta - 2}{\beta + 1} x\right) = 2 + \frac{x}{\beta + 1}.$
    Here, $w_{u, v}(\beta)$ denotes the total weight of the $\beta$ heaviest edges of a shortest path $P$.
    \label{thm:weighted-additive-apasp}
\end{restatable*}

We compare Algorithm \ref{alg:k-weighted-additive-apasp} with concurrent work \cite{DBLP:journals/corr/DoryFKNWV}. 
In addition to handling weighted graphs, our algorithm is faster than both \cite{dor2000apasp} and \cite{DBLP:journals/corr/DoryFKNWV} for all $\beta \geq 4$.

\begin{table}[H]
    \begin{center}
        \begin{tabular}{ |c|c|c|c| } 
             \hline
             \multicolumn{4}{|c|}{$(1+\eps, \beta)$-Additive Approximation} \\
             \hline
             $\beta$ & \cite{dor2000apasp} (Weighted) & \cite{DBLP:journals/corr/DoryFKNWV} (Unweighted) & Algorithm \ref{alg:k-weighted-additive-apasp} (Weighted) \\ 
             \hline
             $2$ & $n^{7/3} = n^{2.34}$ & $n^{2.152}$ & $n^{2.152}$ \\ 
             $4$ & $n^{11/5} = n^{2.200}$ & $n^{2.119}$ & $n^{2.094}$ \\
             $6$ & $n^{17/8} = n^{2.125}$ & $n^{2.098} $ & $n^{2.058}$ \\
             $8$ & $n^{23/11} = n^{2.0909}$ & $n^{2.084}$ & $n^{2.043}$ \\
             $10$ & $n^{29/14} = n^{2.0715}$ & \cite{dor2000apasp} & $n^{2.034}$ \\
             \hline
        \end{tabular}
    \end{center}
    \caption{Comprison with \cite{DBLP:journals/corr/DoryFKNWV} of computing $(1+\eps, \beta)$ approximation on undirected, unweighted graphs with $n$ vertices.
    While a few examples are shown above, we obtain improvements for all $\beta \geq 4$.
    \cite{dor2000apasp} (with adaptions from \cite{cohen2001smallstretch}) and Algorithm \ref{alg:k-weighted-additive-apasp} additionally handle weighted graphs.}
    \label{tbl:weighted-k-additive-apasp}
\end{table}

\subsection{Other Related Work}

Dor, Halperin and Zwick's results to additively approximate APSP \cite{dor2000apasp} improves upon an earlier work of Aingworth et al. \cite{aingworth1999fast} where a $+2$ additive approximation was obtained in $O(n^{2.5})$ time. Cohen and Zwick \cite{cohen2001smallstretch} observed that the algorithm of Dor et al. \cite{dor2000apasp} obtain $2 w_{u, v}(\beta)$ additive approximations for weighted graphs where $w_{u, v}(\beta)$ denote the weights of the $\beta$ heaviest edges on the shortest path. Cohen and Zwick \cite{cohen2001smallstretch} obtain a variety of multiplicative approximation algorithms for weighted graphs with stretch factors 2, 7/3, 3, which are later improved upon by Baswana and Kavitha \cite{baswana2010fasterapasp}. For directed graphs with real weights in $[0, M]$, Yuster \cite{yuster2012approximate} obtained an additive approximation with error $\eps M$ in time $O(n^{(3 + \omega)/2})$.
Building upon this work, Chan \cite{chan2021all} improved the running time on undirected graphs to $O(n^{(3 + \omega^2)/(\omega + 1)})$.

There is a long line of work investigating approximate distance oracles, where the goal is to trade-off the pre-processing time with query time along with further considerations such as space complexity. Thorup and Zwick gave a stretch $2k - 1$ distance oracle with $O(k)$ query time, $O(k n^{1 + 1/k})$ space, and $O(k m n^{1/k})$ pre-processing time \cite{thorup2005approximate}. Since then a rich literature of work has followed with improvements in pre-processing and query time, space complexity and bi-criteria approximations  \cite{mendel2007ramsey, baswana2010fasterapasp, wulff2012approximateoracles, agarwal2013stretch2oracle, chechik2014approximate, chechik2015oracles, elkin2016linearoracle, elkin2016spaceoracle, patrascu2010distanceoracles, sommer2016distanceoracle,akav2020almost2, chechik2022nearly2}.

Comparisons with the concurrent work of \cite{DBLP:journals/corr/DoryFKNWV} can be found in Theorem \ref{thm:rand-2-approx}, Table \ref{tbl:k-additive-apasp}, and Table \ref{tbl:weighted-k-additive-apasp}.

\section{Preliminaries}

Let $G = (V, E)$ be an undirected, unweighted graph with vertices $V$ and edges $E$.
Let $n$ denote the number of vertices and $m$ the number of edges.
Given a pair of vertices $u, v$, let $\delta(G, u, v)$ denote the distance in $G$ between $u, v$ i.e. the length of the shortest path connecting $u$ and $v$.
When the underlying graph $G$ is clear, we omit this parameter and write $\delta(u, v)$.
A path $P$ can be denoted by a sequence of vertices $(u, u_2, \dotsc, v)$ or by its endpoints $P_{u, v}$.
Given two paths $P, Q$ that share an endpoint (and no other vertices), let $P \circ Q$ denote the concatenation of the two paths.
Let $N(u) = \set{v \in V \given (u, v) \in E}$ denote the neighborhood of $u$ and $N(u, d) = \set{v \in V \given \delta(u, v) \leq d}$ denote the depth $d$ neighborhood of $u$.

\begin{definition}
    \label{def:approx-def}
    Let $G$ be an graph and $\delta(G, u, v)$ denote the length of the shortest path from $u$ to $v$.
    A distance estimate $\hat{\delta}(u, v): V \times V \rightarrow \R$ is,
    \begin{enumerate}
        \item an {\bf $\alpha$ multiplicative approximation} if $\delta(u, v) \leq \hat{\delta}(u, v) \leq \alpha \delta(u, v)$ for all $u, v \in V$.
        This can also sometimes referred to as an $\alpha$ {\bf stretch} approximation.
        \item a {\bf $\beta$ additive approximation} if $\delta(u, v) \leq \hat{\delta}(u, v) \leq \delta(u, v) + \beta$ for all $u, v \in V$.
        This can sometimes be denoted as a {\bf $+\beta$ approximation}.
        \item an {\bf $(\alpha, \beta)$ approximation} if $\delta(u, v) \leq \hat{\delta}(u, v) \leq \alpha \delta(u, v) + \beta$ for all $u, v \in V$.
    \end{enumerate}
\end{definition}

On an unweighted graph, let $\bfs(G, w)$ denote running breadth-first search on graph $G$ from root node $w$.
When the graph $G$ does not need to be specified, this may also be denoted $\bfs(w)$.
We also make use of a truncated $\bfs$, which is an execution of $\bfs$ with bounded depth.
A depth bounded $\bfs$ will be denoted $\bfs(G, w, k)$ for depth $k$ or $\bfs(w, k)$ when the graph $G$ is clear.

\begin{definition}
    \label{def:dominating-set}
    Let $G = (V, E)$ be an undirected, unweighted graph.
    A set of vertices $D$ {\bf dominates} $U \subset V$ if every $u \in U$ is either in $D$ or has a neighbor in $D$.

    For a given vertex $u \in U$, define the {\bf representative of $u$ in $D$}, denoted $r(u, D)$, be an arbitrary vertex $z \in D \cap N(u)$.

    For a given vertex $z \in D$, define the {\bf constituency of $z$ in $U$}, denoted $q(z, D)$, as the set $\set{u \in U \given r(u, D) = z}$.
\end{definition}

\begin{definition}
    \label{def:degree-threshold}
    Let $G = (V, E)$ be an undirected, unweighted graph.
    Let $s$ be a degree threshold.
    Define $V_s = \set{v \in V \given \deg(v) \geq s}$.
    Define $E_s = \set{(u, v) \in E \given \min(\deg(u), \deg(v)) < s}$.
\end{definition}

\begin{restatable}{lemma}{hittingsetlemma}
    \label{lemma:hitting-set}
    Let $U$ be a universe of $n$ elements.
    Let $\family = \set{S_1, \dotsc, S_n}$ denote a collection of subsets $S_i \subset U$ such that $|S_i| \geq s$ for all $i$.
    Then, there is a deterministic algorithm $\hittingSet$ that computes a hitting set $X$ of size $\bigO{\frac{n \log n}{s}}$ of $\family$ in time $\tO{n s}$.
    There is also a randomized algorithm $\rHittingSet$ that with high probability computes a hitting set $X$ of size $\bigO{\frac{n \log n}{s}}$ of $\family$ in time $O(n)$.
\end{restatable}

We note that we can easily verify that a randomly sampled set is indeed a hitting set in $O(n s)$ time, by checking each set $S_i$ for an element in $D$.
For any vertex $v \in V_s$, we may interpret $N(v)$ as a subset of $V$ of size at least $s$.
In particular, applying the above lemma we can immediately obtain the following.

\begin{restatable}{lemma}{dominatingsetlemma}
    \label{lemma:dominating-set}
    Let $G$ be a graph on $n$ vertices.
    There is a deterministic algorithm $\dominate$ and randomized algorithm $\rDominate$ that computes a dominating set $D$ of size $\bigO{\frac{n \log n}{s}}$ of $V_s$ in time $O(m + n s)$ and $O(n)$ respectively.
\end{restatable}

\begin{restatable}{lemma}{degreedecompositionlemma}
    \label{lemma:degree-decomposition}
    Let $G$ be a graph on $n$ vertices.
    Given degree thresholds $s_1 > s_2 > \dotsc > s_{k - 1}$, there is a deterministic algorithm $\decompose$ and a randomized algorithm $\rDecompose$ that outputs edge sets $\set{E_i}_{i = 1}^{k}$, edge set $E^*$, and vertex sets $\set{D_i}_{i = 1}^{k}$ satisfying,
    \begin{enumerate}
        \item $E_i = \set{(u, v) \in E \given \min(\deg(u), \deg(v)) < s_{i - 1}}$ is the set $E_{s_{i - 1}}$.
        \item $D_i$ dominates $V_{s_i} = \set{v \in V \given \deg(v) \geq s_i}$ and $|D_i| = \bigtO{\frac{n}{s_i}}$. 
        For convenience, $V_{s_i}$ may also be denoted $V_i$.
        \item $D_1 \subset D_2 \subset \dotsc \subset D_k = V$ and $E_k \subset E_{k - 1} \subset \dotsc \subset E_1 = E$.
        \item $E^* = \bigcup_{i = 1}^{k} E_i^{*}$ where each $E_i^* \subset E$ has for every $v \in V_{i}$, at least one edge $(v, w) \in E_i^{*}$ for $w \in D_i$.
    \end{enumerate}
    Furthermore, $\decompose$ runs in $\tO{k n^2}$ time and $\rDecompose$ runs in $\tO{k n}$ time.
    Note that $\rDecompose$ satisfies the above conditions with high probability.

    For a given vertex $v \in V$, define the {\bf level of $v$}, denoted $\level(v)$, as the integer $i$ such that $s_i \leq \deg(v) < s_{i - 1}$.
    For a given edge $e \in E$, define the {\bf level of $e$}, denoted $\level(e)$, as the integer $i$ such that $e \in E_i \setminus E_{i + 1}$.
\end{restatable}

For sake of completeness, these proofs are given in 
\Cref{app:degree-decomposition}.

We also quote a result from Baswana and Kavitha \cite{baswana2010fasterapasp} that we use crucially for our randomized algorithms (\Cref{sec:rand-2} and \Cref{sec:bounded-additive-approx}).

\begin{theorem}[\cite{baswana2010fasterapasp}]
    \label{thm:baswana}
    There exists a randomized algorithm $\bkAPASP$ that returns a $(2,0)$-approximation of APSP on undirected unweighted graphs with an expected running time complexity $\bigtO{m\sqrt{n}+n^2}$ and space complexity $O(n^2)$.
\end{theorem}

\section{Improved \texorpdfstring{$(2,0)$}{(2, 0)}-Approximation}
\label{sec:mult-approx-bounded}

In this section, we provide improved running time bounds for a $(2, 0)$-approximation of APSP on unweighted, undirected graphs.
We begin with a simple randomized algorithm followed by a slightly slower deterministic algorithm.

\subsection{A Randomized \texorpdfstring{$(2,0)$}{(2, 0)}-Approximate APSP}
\label{sec:rand-2}

We prove the following theorem

\randtwomultapprox

\paragraph{High Level Overview}

We begin by briefly recapping the powerful idea of degree decomposition in computing APSP approximations, first initiated by the work of Aingworth, Chekuri, Indyk, and Motwani \cite{aingworth1999fast}.
The idea is roughly as follows.
If we consider all vertices with degree at least $n^{1/2}$, there will be a small set $D$ of size $n^{1/2}\log{n}$ such that every vertex of high degree has a neighbor in $D$ or is itself in $D$.
$D$ is known as a dominating set.
We compute exact distances from $D$, and then for each pair of vertices, we can take the minimum length path passing through a vertex in $D$.
On paths with at least one high degree vertex, this gives a $+2$ approximation as a vertex in $D$ is adjacent to some vertex in the shortest path $P$.
On the other hand, for each pair of vertices whose shortest path contains no high degree vertices, we can compute APSP in time $O(n^{5/2})$ by computing $\bfs$ (or \dijkstra) from every vertex on the graph $G' \subset G$ containing only edges adjacent to low degree vertices. 

Dor, Halperin and Zwick extended this idea to consider three degree thresholds $s_0=n^{2/3}, s_1=n^{1/3}, s_2=0$ \cite{dor2000apasp}. These degree thresholds are used to compute the dominating sets $D_0 \subset D_1 \subset D_2=V$ according to \decompose ~(\Cref{lemma:degree-decomposition}) where $D_i$ dominates $V_i = \set{v \in V \given \deg(v) \geq s_i}$ and $E_i = \set{(u, v) \in E \given \min(\deg(u), \deg(v)) < s_{i - 1}}$ for $i=0,1,2$ (take $s_{-1}=n+1$). They use \dijkstra 
 (instead of \bfs) from all vertices at the lowest level, i.e., in $G_2=(V,E_2)$, resulting in a running time of $\bigtO{n^{2+1/3}}=\bigtO{n^{7/3}}$. 

For a $(2,0)$-approximation, it is enough to have a $(1,2)$-approximation on paths of length up to $3$. 
Using this crucial observation, Roditty was able to extend the decomposition to 4 levels, and then call a $(1,4)$-approximation algorithm of \cite{dor2000apasp} for paths of length at least $4$. 
This reduces the running time from $\bigtO{n^{7/3}}$ to $\bigtO{n^{2+1/4}}=\bigtO{n^{9/4}}$.

To improve the bound, we instead proceed as follows. We consider shortest paths of length $\geq C$, and less than $C$ separately where $C$ is some threshold which we later set to $C=22$. 
We also consider another threshold $x \in (0,1)$, and define $V_0=\{v \,\, \text{ s.t. } deg(v) \geq n^x\}$. We will set $x$ appropriately later.

\paragraph{High Degree Vertices: Paths of length $\leq C$.}
Consider first all shortest paths $P$ with at least one vertex from $V_0$ that is of degree $\geq n^x$. Moreover, concentrate only on paths of length less than $C=22$. Decompose $G$ into $(1-x)\log n$ levels with degree thresholds $t_j = \frac{n}{2^j}$ for $1 \leq j \leq (1 - x) \log n$ and compute dominating sets $C_j$ of size $\tO{\frac{n}{t_j}}$ and edge sets $F_j$ of size $\bigO{n t_{j - 1}}$.
Computing $\bfs$ on the subgraph $(V, F_j)$ from each vertex of $C_j$ therefore altogether requires $\tO{n^2}$ time. Moreover $|\cup_{j} C_j|=\bigtO{n^{1-x}}$.
For each pair of vertices, compute the estimate $\hat{\delta}(u, v) = \min_{w \in \cup_{j} C_j} \delta(w, u) + \delta(w, v)$ with a \minplus product.
Since, we concentrate on paths of length up to a constant $C$, we can compute the above efficiently with the $\boundedMinPlus$ algorithm (\Cref{thm:bounded-min-plus}).
For the value $j$ such that $P \subset F_j$ but $P \not\subset F_{j + 1}$, there is some vertex $w \in P$ such that $t_j \leq \deg(w) < t_{j - 1}$ and $w^* \in C_j$ is a neighbor of $w$.
Thus, the computed estimate is a $+2$ approximation, which is also a $(2, 0)$-approximation. Computation of the \boundedMinPlus~requires time $\bigtO{\omega(1,1-x,1)}$.

\paragraph*{Paths of length $\geq C$.}

For paths of lengths at least $C = 22$, we are allowed to obtain estimates with large additive errors, which can be computed efficiently using known results from \cite{dor2000apasp}.
We can choose any value of $C$ large enough so that executing the additive approximation of \cite{dor2000apasp} does not affect the overall running time.

\begin{restatable}{lemma}{dhzapasp}
    (\cite{dor2000apasp})
    \label{lemma:dhz-apasp}
    Let $\beta \geq 2$ be even.
    Let $G$ be an undirected, unweighted graph with $n$ vertices and $m$ edges.
    There is an algorithm computing a $+\beta$-approximate APSP solution in time
    \begin{equation*}
        \tO{\min(n^{2 - \frac{2}{\beta+2}} m^{\frac{2}{\beta+2}}, n^{2 + \frac{2}{3 \beta-2}})}
    \end{equation*}

    Denote $\sparseAPASP$ the algorithm running in time $\tO{n^{2 - \frac{2}{\beta+2}} m^{\frac{2}{\beta+2}}}$ and $\denseAPASP$ the algorithm running in time $\tO{n^{2 + \frac{2}{3\beta-2}}}$.
\end{restatable}

Consider paths $P$ such that at least one vertex on $P$ has degree $ \geq n^x$. We run \denseAPASP~with a running time of $\bigtO{n^{2+2/(3C-2)}}=\bigtO{2.03125}$.

\paragraph*{Low Degree Vertices.}

Consider the graph $G_x=(V,E_{n^{x}})$ where $E_{n^{x}}=\{(u,v) \,\,s.t.\,\, deg(u) < n^x \text{ or } deg(v) < n^x\}$. Then $|E_{n^{x}}|=O(n^{1+x})$. We simply run an algorithm by Baswana and Kavitha that satisfies Theorem~\ref{thm:baswana} \cite{baswana2010fasterapasp}. This computes a $(2,0)$-approximation of paths contained in $G_x$ with an expected time of $\bigtO{n^{1.5+x}+n^2}$.

\paragraph*{Time Complexity} Balancing $n^{1.5+x}$ with $\omega(1,1-x,1)$, we get $x=0.53184039$, and an overall running time of $\bigtO{n^{2.03184}}$. Note that the running time bound holds on expectation.

\paragraph{Algorithm}

\Cref{alg:random-2-approx-apasp} begins by initializing the distance matrix with adjacency matrix and setting parameters $x, C$ which will be chosen to optimize the running time.

Phase 1 (Lines 4-6), uses \boundedMinPlus~on matrix $M_t$ to estimate distances for paths of length at most $C$ containing at least one vertex of degree $\geq n^x$. 
To do so, Line 4 considers $\log n$ degree thresholds, each of the form $t = \frac{n}{2^j}$. 
Within each threshold, a dominating set $C_{t}$ is computed (\Cref{line:2-rand-approx:dominating-set}) and a edge set $E_{2n/t}$ such that the degree of each vertex does not exceed $\frac{2n}{t}$.
\Cref{line:dominating-apasp:bfs} of $\dominatingAPASP$ computes a $\bfs$ within the sub-graph $G_t = (V, E_{2n/t})$ from each node in the dominating set $C_t$, keeping only distances up to $C + 1$.
Then, for each pair of vertices, $u, v$, a bounded \minplus product is computed in \Cref{line:dominating-apasp:bounded-min-plus} of $\dominatingAPASP$ to compute the shortest path between $u, v$ passing through a vertex $w \in C_t$.
We repeat this computation for all $t$, balancing the thresholds so that each $\bfs$ search can be computed efficiently.
We repeat for all $t$ large enough such that the dominating set $C_t$ is small.

Phase 2 uses $\bkAPASP$ to compute $(2, 0)$ approximations for paths with maximum degree $n^x$.

Phase 3 uses $\denseAPASP$ to compute a $+C$ approximation, which implies a $(2, 0)$ approximation for all paths of length at least $C$.

Next, we give the pseudocode of $\randomApproxAPASP$ and its correctness proof along with a time complexity analysis.

\paragraph{Correctness}

\sloppy

\begin{proof}
    First,  $\delta(u, v) \leq \hat{\delta}(u, v)$ for all $u, v$. 
    For each call to $\dominatingAPASP$, this follows from \Cref{lemma:dominating-apasp-feasible} as $\hat{\delta}(u, v) \geq \delta(G_t, u, v) \geq \delta(u, v)$ as $G_t \subset G$.
    The claim then follows from the correctness of $\bkAPASP$ and $\denseAPASP$.

    We now show that $\hat{\delta}(u, v) \leq 2 \delta(u, v)$ for all $u, v$.
    Let $P$ be a shortest path and $\deg(P) = \max_{v \in V} \deg(v)$ be the maximum degree of a vertex in path $P$.
    
\IncMargin{1em}
\begin{algorithm}
\SetKwInOut{Input}{Input}\SetKwInOut{Output}{Output}
\Input{Unweighted, undirected Graph $G = (V, E)$ with $n$ vertices}
\Output{Distance estimate $\hat{\delta}: U \times V \rightarrow \Z$ such that $\delta(u, v) \leq \hat{\delta}(u, v) \leq 2 \delta(u, v)$ for all $u, v \in V$}
\BlankLine
Fix parameters $x \gets 0.53184039$ and $C \gets 22$

$\hat{\delta}(u, v) \gets \begin{cases}
    1 & (u, v) \in E \\
    \infty & \otherwise
\end{cases}$

\textcolor{blue}{Phase 1: Compute +2 Approximate Distances of High Max Degree Paths}

\For{$t = 2^j$ for $0 \leq j \leq \ceil{(1 - x) \log n}$}{
    $C_t \gets \dominate \left(G, \frac{n}{t} \right)$ and $G_t = \left(V, E_{\frac{2n}{t}}\right)$
    \label{line:2-rand-approx:dominating-set}
    
    $\hat{\delta} \gets \min(\hat{\delta}, \dominatingAPASP(G_t, V, V, C_t, C + 1))$
    \label{line:2-rand-approx:dominating-apasp}
}

\textcolor{blue}{Phase 2: Compute Approximate Distances of Low Degree Paths}

$\hat{\delta} \gets \min(\hat{\delta}, \bkAPASP(G_x))$ where $G_x = (V, E_{n^x})$
\label{line:2-rand-approx:bgs-apasp}

\textcolor{blue}{Phase 3: Compute Approximate Distances of Long Paths}

$\hat{\delta} \gets \min(\hat{\delta}, \denseAPASP(G, C))$
\label{line:2-rand-approx:dense-apasp}

\caption{$\randomApproxAPASP(G)$}
\label{alg:random-2-approx-apasp}
\end{algorithm}
\DecMargin{1em}

\IncMargin{1em}
\begin{algorithm}[H]

\SetKwInOut{Input}{Input}\SetKwInOut{Output}{Output}
\Input{Unweighted, undirected Graph $G = (V, E)$, source and target subsets $V_1, V_2 \subset V$ and dominating subset $W \subset V$, and distance bound $C$}
\Output{Distance estimate $\hat{\delta}: U \times V \rightarrow \Z$ such that $\delta(u, v) \leq \hat{\delta}(u, v) \leq \min_{w \in W} \delta(u, w) + \delta(w, v)$ for all $u, v$}

\BlankLine

\For{$w \in W$}{
    $\hat{\delta}(w) \gets \bfs(G, w)$
}
\label{line:dominating-apasp:bfs}

Construct $V_1 \times W$ matrix $A$ where $A(v, w) = \begin{cases}
    \hat{\delta}(w, v) & \hat{\delta}(w, v) \leq C \\ 
    \infty & \otherwise
\end{cases}$
\label{line:dominating-apasp:matrix-a}

Construct $W \times V_2$ matrix $B$ where $B(w, v) = \begin{cases}
    \hat{\delta}(w, v) & \hat{\delta}(w, v) \leq C \\ 
    \infty & \otherwise
\end{cases}$
\label{line:dominating-apasp:matrix-b}

$\hat{\delta} \gets \boundedMinPlus(A, B, C)$
\label{line:dominating-apasp:bounded-min-plus}

\caption{$\dominatingAPASP(G, V_1, V_2, W, C)$}
\label{alg:dominating-set-apasp}
\end{algorithm}
\DecMargin{1em}

    \paragraph{Case 1: $\delta(u, v) \geq C$}

    In this case, $\denseAPASP$ computes a $+C$ approximation so that,
    \begin{equation*}
        \hat{\delta}(u, v) \leq \delta(u, v) + C \leq 2 \delta(u ,v)
    \end{equation*}

    \paragraph{Case 2: $\deg(P) < n^{x}$}

    In this case, $P \subset G_x = (V, E_{n^x})$ so that $\bkAPASP$ computes,
    \begin{equation*}
        \hat{\delta}(u, v) \leq 2 \delta(u, v)
    \end{equation*}

    \paragraph{Case 3: $\deg(P) \geq n^x$ and $\delta(u, v) \leq C$}

    We handle this case with Phase 1 and obtain a $+2$ approximation.
    Note that if $\delta(u, v) = 1$, then we already computed an exact distance when initializing $\hat{\delta}$.

    Let $w$ be the vertex of maximum degree in $P$.
    Let $j$ be the integer such that $\frac{n}{2^j} \leq \deg(w) < \frac{n}{2^{j - 1}}$, noting that $j \leq \ceil{(1 - x) \log n}$ since $\deg(w) > n^x$.
    Consider the iteration where $t = 2^j$ of $\dominatingAPASP$ (\Cref{alg:dominating-set-apasp}).
    Then, $P \subset G_t = (V, E_{\frac{2n}{t}})$.
    Then, since $\deg(w) \geq \frac{n}{t}$ we have that there is some neighbor of $w$, say $w^* = r(w, C_t) \in C_t$.
    Since $\delta(u, v) \leq C$, by the triangle inequality,
    \begin{equation*}
        \delta(G_t, u, w^*) \leq \delta(u, w) + 1 \leq C + 1
    \end{equation*}
    and similarly $\delta(G_t, w^*, v) \leq C + 1$ so that both entries are finite in matrices $A, B$.
    
    Then, computing the \minplus product,
    \begin{equation*}
        \hat{\delta}_t(u, v) \leq \delta(G_t, u, w^*) + \delta(G_t, w^*, v) \leq \delta(u, w) + \delta(w, v) + 2 = \delta(u, v) + 2
    \end{equation*}

    Then, we immediately obtain,
    \begin{align*}
        \hat{\delta}(u, v) \leq \delta(u, v) + 2 \leq 2 \delta(u, v)
    \end{align*}
\end{proof}
\paragraph*{Time Complexity}
\begin{proof}
    We analyze the time complexity of our algorithm.
    $x, C$ are parameters that will be optimized.

    \paragraph{Phase 1}
    Fix some $t = 2^j$.
    We compute set $C_t$ of size $|C_t| = \tO{t}$ in time $O(m) = O(n^2)$ by Lemma \ref{lemma:dominating-set}.
    The graph $G_t$ has at most $\frac{2n^2}{t}$ edges, so execution of all $|C_t|$ $\bfs$ searches requires $\tO{n^2}$-time.
    We can also upper bound $t$ as $t = O(n^{1 - x})$.
    Since $C = O(1)$, computing the $\boundedMinPlus(M_t, M_t^T, C + 1)$ requires time $C n^{\omega(1, 1 - x, 1)} \geq n^2$.
    Then, since there are at most $\log n$ such $j$, Phase 1 requires overall time,
    \begin{equation*}
        \bigtO{C n^{\omega(1, 1-x, 1)}}
    \end{equation*}

    \paragraph{Phase 2}
    In Phase 2, we call $\bkAPASP$ on the graph $G_x$ which requires $\tO{m \sqrt{n}} = \tO{n^{1.5 + x}}$ expected time.

    \paragraph{Phase 3}
    In Phase 3, we call $\denseAPASP$ which requires time $\tO{n^{2 + 2/(3C - 2)}}$.
    We can equivalently use $\additiveAPASP$ from \Cref{sec:monotone-min-plus} but this will not affect the overall running time of our algorithm.

    Then, balancing $\omega(1, 1-x, 1) = 1.5 + x$, we choose $x = 0.53184039$ \cite{Complexity}, this leads to a time complexity of $\tO{n^{2.03184039}}$.
    
    Then, we simply choose $C$ such that $2 + \frac{2}{3C - 2} \leq 2.03184039$, noting that $C = 22$ suffices.
\end{proof}

\noindent{\bf Note.} Note that if we replace the $\boundedMinPlus$ computation by vanilla $(\min,+)$-product, and set $x=\frac{3}{4}$, then we get a combinatorial (randomized) algorithm matching $\tilde{O}(n^{2.25})$ time bound obtained by Roditty \cite{roditty2023newapasp}. However, Roditty's algorithm is deterministic. In the following \Cref{sec:2-det-approx}, we give a deterministic algorithm that significantly improves upon Roditty's bound.

\subsection{A Deterministic \texorpdfstring{$(2, 0)$}{(2, 0)}-Approximate APSP}
\label{sec:2-det-approx}

We now move to our deterministic algorithm which gives a slightly weaker running time bound, but brings in many new ideas that will be useful in later sections. We prove the following theorem.

\dettwomultapprox

We extend Roditty's deterministic algorithm that uses $4$ levels significantly to consider $k+1=8$ levels. 
We handle paths of length longer than twelve using a $(1,12)$-approximation algorithm of \cite{dor2000apasp} which requires time $n^{2+1/17}$ and does not dominate the total running time. 
We consider the following degree thresholds $s_0 = n^x, s_1 = n^{\frac{6}{7}x}, s_2 = n^{\frac{5}{7}x}, \dotsc, s_6 = n^{\frac{1}{7}x}$ and $s_7=0$.
We will leave $x$ as a parameter to be optimized later.
Given these degree thresholds, we compute dominating sets $D_0 \subset D_1 \subset \dotsc \subset D_6 \subset D_7 = V$ according to $\decompose$ (\Cref{lemma:degree-decomposition}) where $D_i$ dominates $V_i = \set{v \in V \given \deg(v) \geq s_i}$ and $E_i = \set{(u, v) \in E \given \min(\deg(u), \deg(v)) < s_{i - 1}}$ for $i=0,1,..,7$ (take $s_{-1}=n+1$).

\paragraph{High Degree Vertices: Paths of length $\leq C$.}
We handle high degree vertices in a similar way as we did for our randomized algorithm, that is we compute a dominating set of size $\bigtO{n^{1-x}}$ in $\bigtO{n^2}$ time, and compute $\boundedMinPlus$ to obtain a $+2$-additive approximation of all paths of length at most $C$ going through some vertex in $V_0$. This step requires $\bigtO{n^{\omega(1,1-x,1)}}$ time.
\paragraph*{Paths of length $\geq C$.}
For paths of lengths at least $C = 12$, we will use known results from \cite{dor2000apasp}.
Consider paths $P$ such that all vertices on $P$ have degree less than $n^x$. Here, we have decomposed the remaining graph into $k = 7$ levels. $\sparseAPASP$ with $k = 7$ levels achieves a $+12$ approximation. So, we can simply run it.
The algorithm $\sparseAPASP$ executes $\dijkstra(G_{i,w}, w, \hat{\delta})$ from all $w \in D_i$ where $G_{i, w} = (V, E_i \cup E^* \cup (w \times V))$ is equipped with the current estimates $\hat{\delta}$ as edge weights, which requires $\tO{n^{2 + \frac{x}{7}}}$ time as $D_i$ has size $\tO{n^{1 - \frac{7 - i}{7} x}}$ and the graph $G_{i, w}$ has $\bigO{n^{1 + \frac{7 - (i - 1)}{7} x}}$ edges. If path $P$ has a vertex in $V_0$, then we can run $\denseAPASP$ with a running time of $\bigtO{n^{2+1/17}}=\bigtO{n^{2.0588}}$.
Again, since we only need $\denseAPASP$ to compute good approximations on paths with high degree vertices, we can choose any constant value of $C$ large enough such that executing $\denseAPASP$ does not dominate the running time.

\paragraph*{Low Degree Vertices: Paths of length $\leq C$.}

\paragraph{Paths of length 6 or more}
For paths of length at least 6, we compute a $+6$ approximation following $\boundedAdditiveAPASP$ in \Cref{sec:additive-approx-constant}.
We give a brief overview here with special attention to the specific application required with more details provided in \Cref{sec:additive-approx-constant}. We use fast rectangular matrix multiplication judiciously to improve upon \cite{dor2000apasp}.

We compute $\boundedAdditiveAPASP$ on the original graph $G$ and therefore choose new degree thresholds and use a new degree decomposition of the graph $G$.
We first describe our algorithm combinatorially, and then mention where fast matrix multiplication can be applied.
For a full discussion, see \Cref{sec:additive-approx-constant}.
To obtain a $+6$ approximation, choose $7$ degree thresholds $r_1, \dotsc, r_{7}$ and decompose the vertex and edge sets of $G$ as $C_1 \subset \dotsc \subset C_8 = V$ and $F_8 \subset F_7 \subset \dotsc \subset F_1 = E$.
At level $i$, compute $\dijkstra$ from each $w \in C_i$ on the graph
\begin{equation*}
    G_{i, w} = \left(V, F_i \cup \left( \bigcup_{i + j_1 + j_2 \leq 17} C_{j_1} \times C_{j_2} \right) \cup F^* \cup (w \times V) \right)
\end{equation*}
noting that both $F_i, C_{j_1} \times C_{j_2}$ have size $\bigO{n^{1 + \frac{8 - (i - 1)}{8}}}$.
Here $F^* = \bigcup_{i = 1}^{7} F_i^*$ is output by $\decompose$.

Consider a shortest path $P$. We use the following definition of blocking vertices and blocking levels.

\begin{restatable}{definition}{blockingvertices}
    \label{def:blocking-vertices}
    Let $u, v$ be vertices in a graph $G$ and $P$ a path between $u, v$. 
    Let $(D_1, \dotsc, D_k)$, $(E_1, \dotsc, E_k)$ be the outputs of a call to $\decompose$ with degree thresholds $s_1, \dotsc, s_{k - 1}$.
    If there is an edge $(a, b) \in P$ such that $(a, b) \notin E_{\level(u)}$, the {\bf blocking vertex} of $P$ from $u$, denoted $b(u, P)$, is the closest vertex to $v$ that is an endpoint to such an edge.
    If no such edge exists, then $b(u, P) = v$.

    Then, the {\bf blocking vertices of $P$} is the set $B(P) = \set{x_0, x_1, \dotsc, x_t}$ defined in the following manner:
    \begin{enumerate}
        \item $x_0 = u$
        \item $x_1 = v$
        \item $x_i = b \left(x_{i - 1}, P_{x_{i - 1}, x_{i - 2}}\right)$ where $P_{x_{i - 1}, x_{i - 2}}$ is the sub-path of $P$ between $x_{i - 1}, x_{i - 2}$.
    \end{enumerate}

    For any $j \geq \level(v)$, {\bf blocking levels of $P$ at level $j$} is the set $L_B(P, j) = \set{\level(x_i) \given x_i \in B(P) \andT \level(x_i) < j}$.
    Denote $L_B(P) = L_B(P, \level(v))$.
\end{restatable}

We prove that the number of blocking levels directly determines the additive error accumulated on a given path.

\begin{restatable}{lemma}{blockapproxerror1}
    \label{lemma:block-approx-error-1}
    Let $u, v$ be vertices in a graph $G$ and $P$ be a shortest path between $u, v$.
    Let $B(P)$ be the blocking vertices of $P$.
    Then, the distance estimate $\hat{\delta}$ satisfies
    \begin{equation*}
        \hat{\delta}(v^*, u)  \leq \delta(v, u) + 2 |L_B(P)| + 1
    \end{equation*}

    For any level $j \geq \level(v)$, if $v \in D_j$,
    \begin{equation*}
        \hat{\delta}(v, u) \leq \hat{\delta}_{j}(v, u) \leq \delta(v, u) + 2 |L_B(P, j)|
    \end{equation*}
\end{restatable}

The proof is deferred to \Cref{sec:sparse-apasp-analysis}.

If $P \subset F_{5}$ then we obtain a $+6$ additive approximation as a Corollary to \Cref{lemma:block-approx-error-1} as $|L_B(P)| \leq 3$. 
Thus, suppose $P \subset F_{j}$ but $P \not\subset F_{j + 1}$ for some $j \leq 4$.

\begin{figure}[ht]
    \centering
    \includegraphics[width=0.7\textwidth]{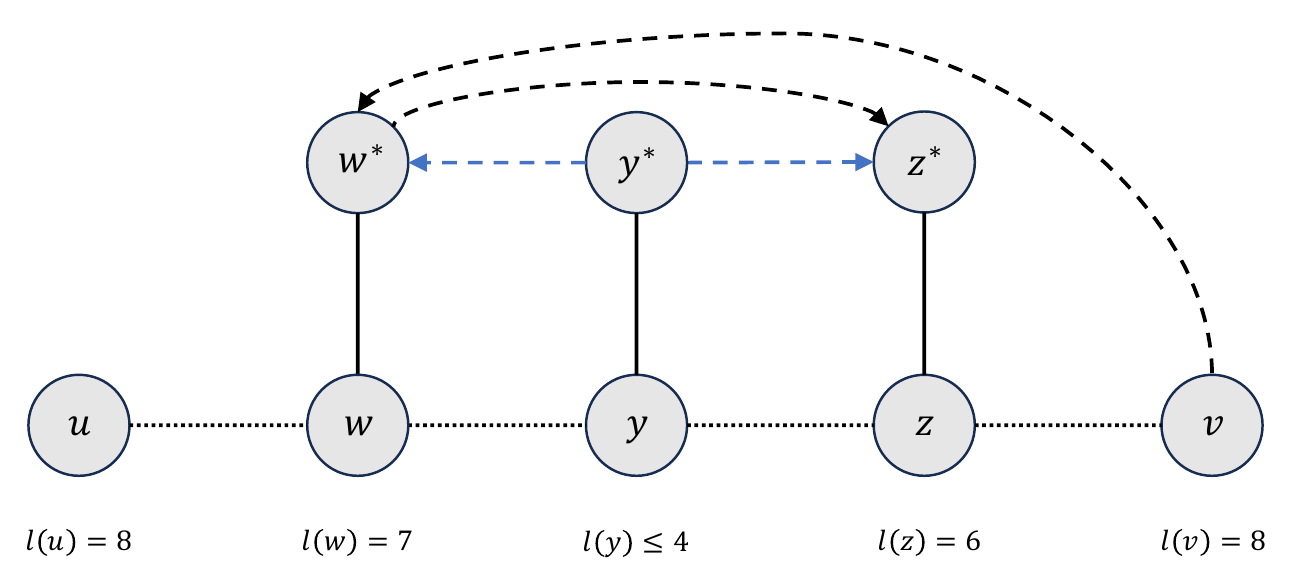}
    \caption{Solid lines denote edges and dotted lines represent paths.
    Blue dashed arrows denote distance estimates used in the \minplus product. 
    Black dashed arrows denote edges $w \times V$ used in $G_{j, w}$.}
    \label{fig:+6-min-plus-approx}
\end{figure}

For any vertex $w \in V_7$ and $w^* = r(w, C_7)$, the graph $G_{7, w^*}$ contains edge sets $F_7 \cup (w^* \times V) \cup \left(\bigcup_{j = 1}^{4} C_j \times C_6\right)$.
For $v \in V$, let $z$ denote the first vertex on $P_{w, v}$ such that the sub-path $P_{z, v} \subset F_7$ and $y \in P_{w, v}$ be the vertex of maximum degree. See \Cref{fig:+6-min-plus-approx}.
Then $\level(z) \leq 6$, and $\level(y) \leq 6$. Suppose $\level(z)=6$ (other cases are only easier).
Let $y^* = r(y, C_j)$ where $j \leq 6$ and $z^* = r(z, C_6)$.
Then, since $P \subset F_j$, $\hat{\delta}(y^*, z^*) \leq \delta(y, z) + 2$ and $\hat{\delta}(y^*, w^*) \leq \delta(y, w) + 2$.
Executing $\dijkstra$ from $w^*$, the path $(w^*, y^*, z^*, z) \circ P_{z, v} \subset G_{7, w^*}$ so that the distance estimate,
\begin{equation*}
    \hat{\delta}(w^*, v) \leq \hat{\delta}(w^*, y^*) + \hat{\delta}(y^*, z^*) + 1 + \delta(z, v) \leq \delta(w, v) + 5
\end{equation*}

Then, returning to the shortest path $P$ with endpoints $u, v$, the graph $G_{8, u}$ contains $F_8 \cup (u \times v)$.
If $w$ denotes the closest vertex to $v$ on $P$ such that $P_{w, u} \subset F_8$ and $w^* = r(w, C_7)$, then the path $(v, w^*, w) \circ P_{w, u} \subset G_{8, v}$ so that the distance estimate computed is at most,
\begin{equation*}
    \hat{\delta}(v, u) \leq \hat{\delta}(v, w^*) + 1 + \delta(w, u) \leq \delta(v, u) + 6
\end{equation*}

Note that we have not actually required the use of every $C_{j_1} \times C_{j_2}$ such that $i + j_1 + j_2 \leq 17$ when executing $\dijkstra$ from $C_j$.
Instead, we have only used the sets $C_i \times C_6$ for $i \leq 4$ when executing $\dijkstra$ from $C_7$.
Instead of adding these edges, we can compute the distances that would have been obtained with these edges with a \minplus product and encode these distances into the smaller edge set $w \times V$.
We encode the distances $C_i \times C_7$ into a matrix $A_i$ and the distances $C_i \times C_6$ into a matrix $B_i$ and compute $A_i * B_i$ with entries bounded by $14$ for $i=1,2,3,4$.
In the example above, this gives the estimate $\hat{\delta}(w^*, z^*) \leq \delta(w, z) + 4$.
Then, when executing $\dijkstra$ from $w^* \in C_7$, the path $(w^*, z^*, z) \circ P_{z, v} \subset G_{7, w^*}$ so the distance estimate is again at most,
\begin{equation*}
    \hat{\delta}(w^*, v) \leq \hat{\delta}(w^*, z^*) + 1 + \delta(z, v) \leq \delta(w, v) + 5
\end{equation*}

\paragraph{Paths of length 5 or less}
We now describe the most challenging cases: obtaining a $(2, 0)$ approximation on shortest paths $P_{u, v}$ of length at most 5.
Let $G_x = (V, E_{n^x})$ denote the graph $G$ pruned so that the maximum degree of $G_x$ is strictly less than $n^x$. We can compute a depth 2 \bfs (that is \bfs tree upto two levels) in time $n^{2x}$ from each vertex, with a total time of $n^{1+2x}$. Our choice of $x$ will ensure this is allowed.

As a warm-up, consider obtaining a $+2$ approximation on paths of length at most 3.
For 2 edge paths, a depth 2 $\bfs$ computes exact distances for all paths of length 2 in $G_x$.
For 3 edge paths $P = (u, u_2, v_2, v)$, assume $\deg(u) \geq \deg(v)$.
A depth 2 $\bfs$ from $u$ implies $\hat{\delta}(u, v_2) = 2$ and if $u^* = r(u, D_{\level(u)})$ then we can efficiently update $\hat{\delta}(u^*, v_2) \leq 1 + \hat{\delta}(u, v_2) = 3$.
Then, in the graph $G_{\level(u), u^*}$, there is the path $(u^*, v_2, v)$ so $\dijkstra$ computes a distance estimate $\hat{\delta}(u^*, v) \leq 4$.
Finally, in the graph $G_{k, v}$ (that is at the lowest level) there is the path $(v, u^*, u)$ so $\dijkstra$ computes a distance estimate $\hat{\delta}(u, v) \leq 5$, which is a $+2$ approximation.

Finally, we consider paths $P = (u, u_2, u_3, v_3, v_2, v)$ of 5 edges.
Paths of length 4 are easier to handle and we will mention the necessary modifications in the full proof.
Recall that we have decomposed the graph $G_x$ into $k = 7$ levels with degree thresholds $s_1, s_2,...,s_6, s_7=0$.

Consider the case where both endpoints have low degree $\deg(u), \deg(v) < s_{6}$, or identically $\level(u), \level(v)=7$.
Let w.l.o.g, $\deg(u_2) \geq \deg(v_2)$. 
Let $u_2^* = r(u_2, D_{\level(u_2)})$.
If $P \subset E_{\level(u_2)}$, then $\hat{\delta}(u_2^*, v) \leq \delta(u_2, v) + 1$.
Otherwise, both $\level(u_3), \level(v_3) < \level(u_2)$ (since even if one of them is at $\level(u_2)$ or higher, we have $P \subset E_{\level(u_2)}$).
Therefore, $\hat{\delta}(v_3^*, u_2^*) \leq \delta(v_3, u_2) + 2$ so the path $(u_2^*, v_3^*, v_3) \circ P_{v_3, v} \subset G_{\level(u_2), u_2^*}$ and $\dijkstra$ computes the desired distance estimate: $\hat{\delta}(u_2^*, v)\leq \delta(u_2,v)+3$.
Then, as $(v, u_2^*, u_2, u) \subset G_{7, v}$, $\dijkstra$ computes a $+4$ additive approximation.
\Cref{fig:u-v-level-k} gives an illustration.

\begin{figure}[ht]
    \centering
    \includegraphics[width=0.7\textwidth]{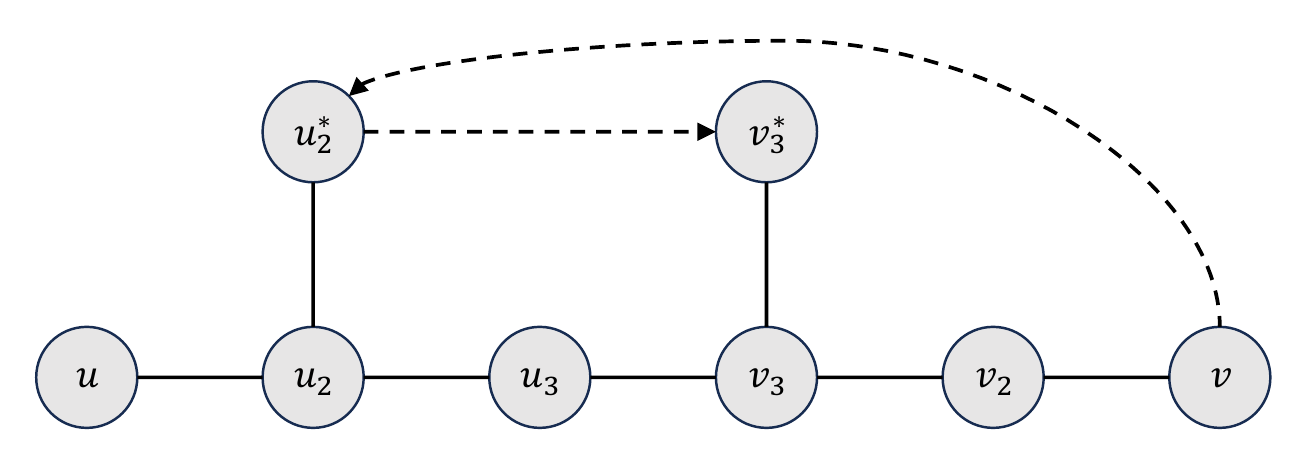}
    \caption{Case 2.3.1: $\level(u) = \level(v) = k$.
    Solid lines denote edges and dotted lines represent paths.
    Black dashed arrows denote edges $w \times V$ in $G_{j, w}$.}
    \label{fig:u-v-level-k}
\end{figure}

Then, assume at least one endpoint (say $u$) has degree $ \geq s_{k - 1}$.
Suppose there is some vertex (say $v_3$) with $n^x > \deg(v_3) \geq s_1$.
Using a bounded rectangular \minplus product for matrices of size $D_1 \times V$ and $D_1 \times D_{6}$, the \minplus product yields an estimate $\hat{\delta}(u^*, v) \leq \delta(u*, v) + 2 \leq \delta(u, v) + 3$.
Then, as $(v, u^*, u) \subset G_{7, v}$, $\dijkstra$ computes a $+4$ additive approximation.
\Cref{fig:2-approx-min-plus} gives an illustration.

\begin{figure}[ht]
    \centering
    \includegraphics[width=0.7\textwidth]{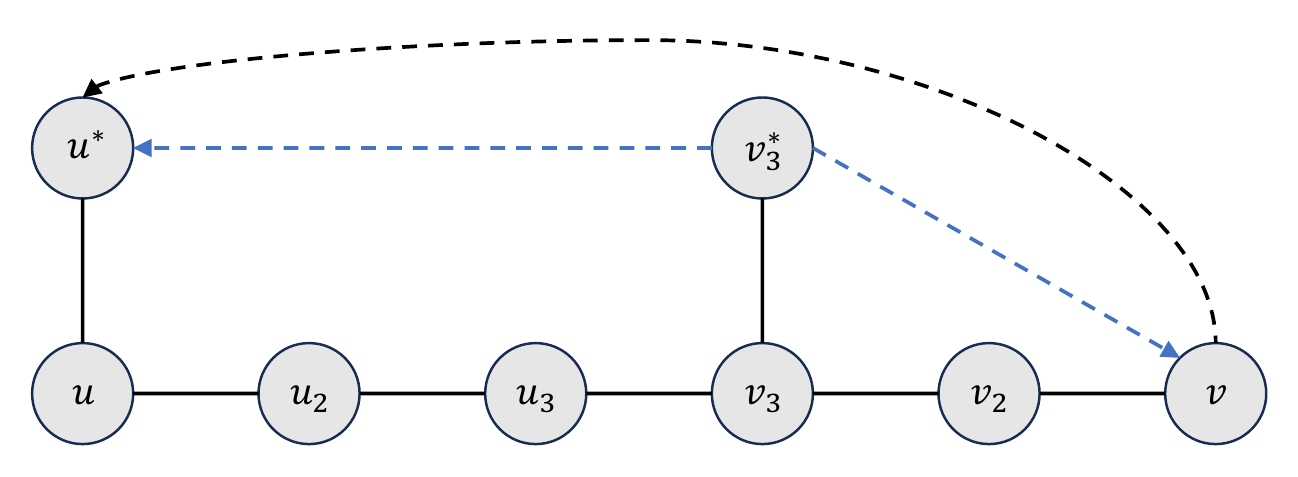}
    \caption{Case 2.3.2: $\deg(u) \geq s_{6}$ and $\deg(v_3) > s_1$.
    Solid lines denote edges in $G$.
    Blue dashed arrows denote distance estimates used in the \minplus product.
    Black dashed arrows denote edges $w \times V$ in $G_{j, w}$.}
    \label{fig:2-approx-min-plus}
\end{figure}

We can now assume that every vertex has degree $< s_1$.
We now compute depth 3 $\bfs$ in $G_2 = (V, E_2)$ from every vertex with degree at most $s_2$. 
One depth 3 $\bfs$ requires $s_2s_1^2=n^{\frac{3k - 4}{k} x}$ time (contrast to $n^{3x}$ if naively applying depth 3 $\bfs$ on the graph $G_x$). The correctness depends on extensive case analysis.

In Case 2.3.3, suppose $\deg(u_2), \deg(v_2) \geq s_2$.
Suppose $\deg(v) \geq \deg(u)$.
From our earlier discussion, since $P \subset E_2$, $\hat{\delta}(u_2^*, v) \leq \delta(u, v) + 1$.
Then, since the path $(v^*, u_2^*, u_2, u) \subset G_{\level(v), v^*}$, $\dijkstra$ computes an estimate,
\begin{equation*}
    \hat{\delta}(v^*, u) \leq \hat{\delta}(v^*, u_2^*) + 2 \leq \delta(u, v) + 3
\end{equation*}
Then, as $(u, v^*, v) \subset G_{8, u}$, $\dijkstra$ computes a $+4$ additive approximation.

\begin{figure}[ht]
    \centering
    \begin{subfigure}[b]{0.45\textwidth}
         \centering
         \includegraphics[width=\textwidth]{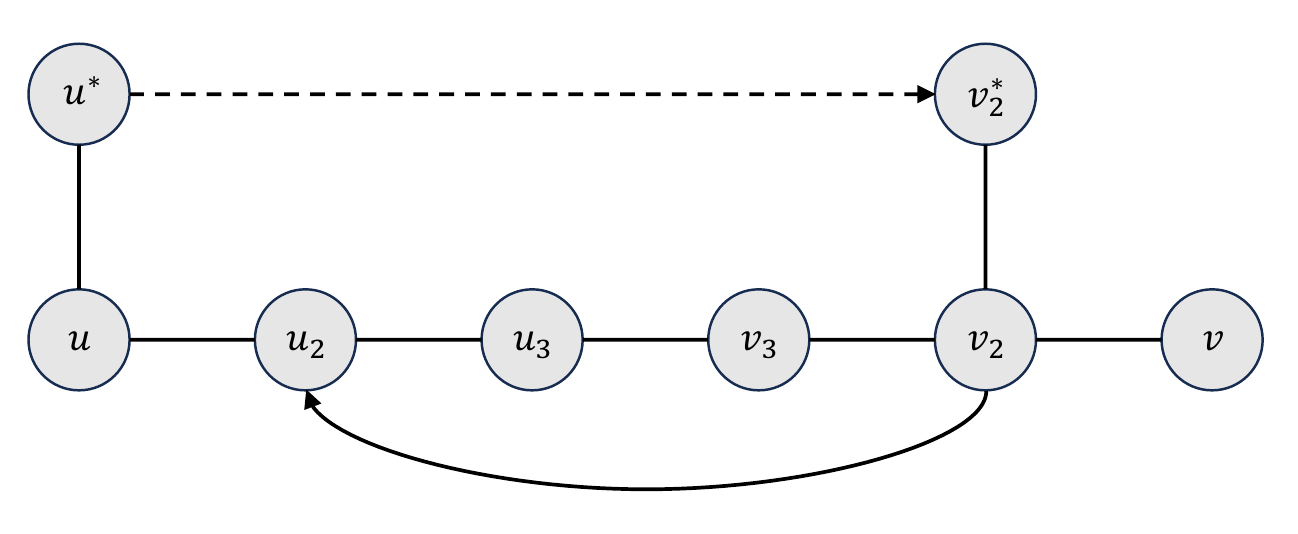}
         \caption{Depth 3 $\bfs$ search from $v_2$.}
         \label{fig:depth-3-v2}
     \end{subfigure}
     \hfill
     \begin{subfigure}[b]{0.45\textwidth}
         \centering
         \includegraphics[width=\textwidth]{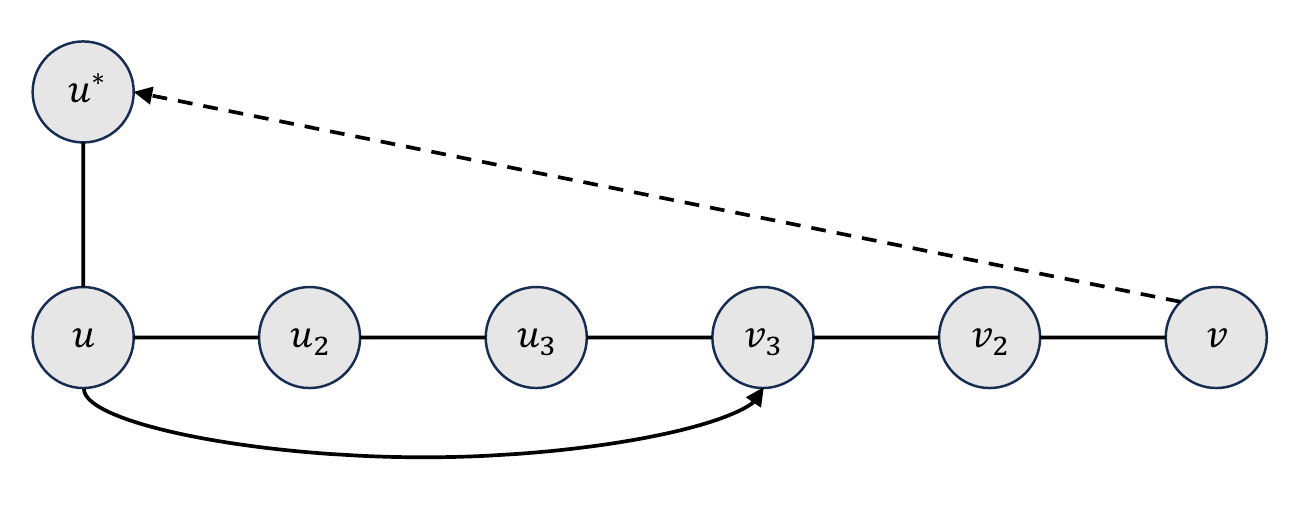}
         \caption{Depth 3 $\bfs$ search from $u$.}
         \label{fig:depth-3-u}
     \end{subfigure}
     \caption{Case 2.3.4 and 2.3.5: We use depth 3 $\bfs$ in \Cref{line:2-approx:depth-3-bfs} to obtain more accurate distance estimates. 
     Solid lines denote edges in $G$.
     Solid arrows denote depth 3 $\bfs$ executions.
     Dotted arrows denote edges $w \times V$ used in executing $\dijkstra$ on the graph $G_{j, w}$.}
    \label{fig:2-approx-depth-3}
\end{figure}

In Case 2.3.4, suppose $\deg(u_2) \geq s_2$ but $\deg(v_2) < s_2$.
If $\deg(v) \geq \deg(u)$, then we can proceed exactly as above, so assume $\deg(u) > \deg(v)$.
The depth 3 $\bfs$ from $v_2$ computes $\hat{\delta}(v_2, u_2) = 3$ and \Cref{line:2-approx:truncated-bfs-hitting-set} updates $\hat{\delta}(v_2^*, u_2) = 1 + \hat{\delta}(v_2, u_2) = 4$. If $\deg(v_2) \geq \deg(u)$, then, the path $(v_2^*, u_2, u, u^*) \subset G_{\level(v_2), v_2^*}$, $\dijkstra$ computes a distance estimate at most,
\begin{equation*}
    \hat{\delta}(v_2^*, u^*) \leq \hat{\delta}(v_2^*, u_2) + 2 \leq 6
\end{equation*}
Since $\deg(u) > \deg(v)$, $(u^*, v_2^*, v_2, v) \subset G_{\level(u), u^*}$, $\dijkstra$ computes a distance estimate at most,
\begin{equation*}
    \hat{\delta}(u^*, v) \leq \hat{\delta}(u^*, v_2^*) + 2 \leq 8
\end{equation*}
Then, as $(v, u^*, u) \subset G_{7, u}$, $\dijkstra$ computes a distance estimate at most 9, a $+4$ additive approximation.
See \Cref{fig:depth-3-v2} for an illustration.

On the other hand $\deg(u) > \deg(v_2)$.
If $\deg(u) > s_2$, we obtain a $+2$ approximation following the previous discussion of Case 2.3.3.
Otherwise, if $\deg(u) < s_2$ then the depth 3 $\bfs$ from $u$ computes $\hat{\delta}(u, v_3) = 3$ and \Cref{line:2-approx:truncated-bfs-hitting-set} updates $\hat{\delta}(u^*, v_3) = 1 + \hat{\delta}(u, v_3) = 4$.
Since $\deg(u) > \deg(v_2)$, $(u^*, v_3, v_2, v) \subset G_{\level(u), u^*}$ so $\dijkstra$ computes a distance estimate at most,
\begin{equation*}
    \hat{\delta}(u^*, v) \leq \hat{\delta}(u^*, v_3) + 3 \leq 6
\end{equation*}
Then, as $(v, u^*, u) \subset G_{7, u}$, $\dijkstra$ computes a distance estimate at most 7, which is a $+2$ additive approximation.
See \Cref{fig:depth-3-u} for an illustration.

In the full proof we give a detailed argument on all possible variations of length 5 paths as well as the modifications necessary to compute a $+4$ additive approximation on length 4 paths.

\paragraph*{Time Complexity}

The time-complexity to compute \boundedMinPlus~to handle paths with at least one vertex in $V_0$ is $\tilde{O}(\omega(1,1-x,1))$. This computation dominates the \boundedMinPlus~computation in Case 2.3.2 \Cref{fig:2-approx-min-plus} (see \Cref{prop:2-approx-min-plus-convexity}). The time complexity to handle paths with low degree vertices for paths of length at most $5$ is 
\begin{equation*}
    \bigtO{n^{1 + 2x} + n^{1 + \frac{3k - 4}{k}x} + n^{2 + \frac{x}{k}}}
\end{equation*}
where the first term is from running depth 2 \bfs, the second term is from running depth 3 \bfs~from vertices with degree at most $s_2$, and the last term is from running \dijkstra. If we balance these terms, we get we choose $x = 0.4414248$ and $k = 7$ to obtain the time complexity $\tO{n^{2.07203166}}$. The non-dominating terms include time to run \boundedAdditiveAPASP on paths of length $[6,12]$ which requires $\tilde{O}(n^{2.0638})$ by \Cref{cor:bounded-additive-examples} and \denseAPASP(G,12) that requires time $\tO{n^{35/17}} = \tO{n^{2.0589}}$ by Lemma \ref{lemma:dhz-apasp}.

\paragraph{Algorithm}
\Cref{alg:2-approx-apasp} computes a deterministic $(2,0)$-approximation of APSP on unweighted graphs. 
Phase 1 (Lines \ref{line:2-approx:dominating-loop}-\ref{line:2-approx:dominating-apasp}), uses \boundedMinPlus on matrix $M_t$ to estimate distances for paths of length at most $C$ containing at least one vertex of degree $\geq n^x$. 
To do so, Line \ref{line:2-approx:dominating-loop} considers $\log n$ degree thresholds, each of the form $t = \frac{n}{2^j}$. 
Within each threshold, a dominating set $C_{t}$ is computed (\Cref{line:2-approx:dominating-set}) and a edge set $E_{2n/t}$ such that the degree of each vertex does not exceed $\frac{2n}{t}$.
\Cref{line:dominating-apasp:bfs} of $\dominatingAPASP$ computes a $\bfs$ within the sub-graph $G_t = (V, E_{2n/t})$ from each node in the dominating set $C_t$, keeping only distances up to $C + 1$.
Then, for each pair of vertices, $u, v$, a bounded \minplus product is computed in \Cref{line:dominating-apasp:bounded-min-plus} of $\dominatingAPASP$ to compute the shortest path between $u, v$ passing through a vertex $w \in C_t$.
We repeat this computation for all $t$, balancing the thresholds so that each $\bfs$ search can be computed efficiently.
We repeat for all $t$ large enough such that the dominating set $C_t$ is small.

Phase 2 handles paths with low degree vertices. \Cref{line:2-approx:decompose-k} uses $\decompose$ on the remaining sparse graph into $k = 7$ degree levels specified in \Cref{line:2-approx:degree-thresholds}.

Phase 2.1, computes $\bfs$ from each vertex in $D_1$ (Lines \ref{line:2-approx:define-g-s-1}-\ref{line:2-approx:2.1-dominating-set}) and use a bounded \minplus product to find approximate distances between each pair of vertices $V \times D_{k - 1}$ that goes through a vertex in $D_1$ in \Cref{line:dominating-apasp:bounded-min-plus} of $\dominatingAPASP$.

In Phase 2.2, a depth 2 $\bfs$ is computed from each vertex in $G_x$ (\Cref{line:2-approx:depth-2-bfs}), as well as depth 3 $\bfs$ from each vertex of degree at most $s_2$ in the graph $G_2 = (V, E_2)$ (\Cref{line:2-approx:depth-3-bfs}).

\IncMargin{1em}
\begin{algorithm}[H]
\SetKwInOut{Input}{Input}\SetKwInOut{Output}{Output}
\Input{Unweighted, undirected Graph $G = (V, E)$ with $n$ vertices}
\Output{Distance estimate $\hat{\delta}: U \times V \rightarrow \Z$ such that $\delta(u, v) \leq \hat{\delta}(u, v) \leq 2 \delta(u, v)$ for all $u, v \in V$}
\BlankLine
Fix parameters $x \gets 0.45509125$ and $C \gets 12$

$\hat{\delta}(u, v) \gets \begin{cases}
    1 & (u, v) \in E \\
    \infty & \otherwise
\end{cases}$

\textcolor{blue}{Phase 1: Compute +2 Approximate Distances of High Max Degree Paths}

\For{$t = 2^j$ for $0 \leq j \leq \ceil{(1 - x) \log n}$}{
    \label{line:2-approx:dominating-loop}

    $C_t \gets \dominate \left(G, \frac{n}{t} \right)$ and $G_t = \left(V, E_{\frac{2n}{t}}\right)$
    \label{line:2-approx:dominating-set}
    
    $\hat{\delta} \gets \min(\hat{\delta}, \dominatingAPASP(G_t, V, V, C_t, C + 1))$
    \label{line:2-approx:dominating-apasp}
}

\textcolor{blue}{Phase 2: Compute Approximate Distances of Low Degree Paths}

$s_i \gets n^{\frac{k - i}{k} x}$ for all $1 \leq i \leq k - 1 = \frac{C}{2}$
\label{line:2-approx:degree-thresholds}

$(D_{1}, D_{2}, \dotsc, D_{k}), (E_{1}, E_{2}, \dotsc, E_{k}), E^* \gets \decompose(G_x, (s_{1}, s_{2}, \dotsc, s_{k - 1}))$ where $G_{x} \gets (V, E_{n^{x}})$
\label{line:2-approx:decompose-k}

\textcolor{blue}{Phase 2.1: Computing Approximate Distances of Paths with Vertex in $V_{s_1}$}

$G_{s_1} \gets (V, E_{n^x} \cup E^*)$
\label{line:2-approx:define-g-s-1}

$\hat{\delta} \gets \min(\hat{\delta}, \dominatingAPASP(G_{s_1}, V, D_{k - 1}, D_1, C + 2))$
\label{line:2-approx:2.1-dominating-set}

\textcolor{blue}{Phase 2.2: Compute Approximate Distances for $\delta(u, v) \leq 5$}

\For{$v \in V$}{
    $M_{v, 2} \gets \bfs(G_x, v, 2)$ \label{line:2-approx:depth-2-bfs}

    \lIf{$\deg(u) \leq s_2$}{$M_{u, 3} \gets \bfs(G_2, u, 3)$ where $G_2 = (V, E_2)$} \label{line:2-approx:depth-3-bfs}

    $\hat{\delta}(v, u) \gets \min(\hat{\delta}(u, v), M_{v, 2}(u), M_{v, 3}(u))$
}

\textcolor{blue}{Phase 2.3: Compute Approximate Distances of Low Degree Paths}

\For{$1 \leq i \leq k$}{
    \For{$w \in D_i$}{
        $\hat{\delta}(w, v) \gets \min(\hat{\delta}(w, v), \min_{u \in q(w, D_i)} 1 + \hat{\delta}(u, v))$ for all $v \in V$ \label{line:2-approx:truncated-bfs-hitting-set}
        
        $\hat{\delta} \gets \dijkstra(G_{i, w}, w, \hat{\delta})$ where $G_{i, w} \gets \left(V, E_i \cup (w \times V) \cup E^* \right)$
    }
}

\textcolor{blue}{Phase 3: Compute Approximate Distances of Long Paths}

$\hat{\delta} \gets \min(\hat{\delta}, \boundedAdditiveAPASP(G, 6, C), \denseAPASP(G, C))$

\caption{$\twoApproxAPASP(G)$}
\label{alg:2-approx-apasp}
\end{algorithm}
\DecMargin{1em}
Phase 2.3 (Lines 26-29), performs $\dijkstra$ from each vertex $w \in D_i$ on the graph $(V, E_i \cup E^* \cup (w \times V))$ with the current estimates $\hat{\delta}$ equipped as edge weights.
Before $\dijkstra$ is executed, the distances from a vertex $w \in D_i$ in the dominating set is updated by examining distances from its constituency (see Definition \ref{def:dominating-set}).

\sloppy

Finally, Line 31 handles paths of distance between 6 and $12$ with a call to $\boundedAdditiveAPASP$ and all paths of distance larger than $12$ with a call to $\denseAPASP$.
Note that we can also easily use \Cref{alg:k-additive-apasp} in place of $\denseAPASP$, although this will not affect the asymptotic performance of our final algorithm.

\subsubsection*{Proof of \texorpdfstring{\Cref{thm:2-approx-apsp}}{}}

We start with the correctness proof, followed by a run time analysis.
\paragraph*{Correctness.}
\begin{proof}

    Consider a vertex pair $u, v \in V \times V$.
    Let $\delta(u, v)$ denote the length of a shortest path between $u, v$.
    From Lemma \ref{lemma:2-approx-feasible}, we have that $\delta(u, v) \leq \hat{\delta}(u, v)$ for all $u, v$. 
    We therefore focus on proving $\hat{\delta}(u, v) \leq 2 \delta(u, v)$ for all $u, v$.
    We can assume $\delta(u, v) \geq 2$ since 1 edge paths can easily be found by examining the adjacency matrix.

    \paragraph{Case 1: $\deg(P) \geq n^x$}

    Suppose $\delta(u, v) \geq C$.
    By the correctness of $\denseAPASP$, we have $\hat{\delta}(u, v) \leq \delta(u, v) + C \leq 2 \delta(u, v)$.
    Thus, in the following, assume $\delta(u, v) \leq C$

    The remaining proof is similar to Case 3 of the randomized algorithm.

    \paragraph{Case 2: $\deg(P) < n^x$}

    We now prove the correctness of Algorithm \ref{alg:2-approx-apasp} on low-degree paths.

    \paragraph{Case 2.1: $\delta(u, v) \geq C$}

    Correctness immediately follows from $\denseAPASP$.
    In particular,
    \begin{equation*}
        \hat{\delta}(u, v) \leq \delta(u, v) + C \leq 2 \delta(u, v)
    \end{equation*}

    \paragraph{Case 2.2: $6 \leq \delta(u, v) \leq C$}

    Correctness immediately follow from $\boundedAdditiveAPASP$.
    In particular,
    \begin{equation*}
        \hat{\delta}(u, v) \leq \delta(u, v) + 6 \leq 2 \delta(u, v)
    \end{equation*}

    \paragraph{Case 2.3: $4 \leq \delta(u, v) \leq 5$}

    Denote the path $P = (u, u_2, u_3, v_3, v_2, v)$.
    Note that if $\delta(u, v) = 4$, we in fact have $u_3 = v_3$.

    Recall that for any vertex, we define the level of vertex $v$ as $\level(v)$ to be the index $i$ where $s_i \leq \deg(v) < s_{i - 1}$. 
    For any path $P$, define $\level(P) = \min_{v \in P} \level(v)$ to be the minimum level of any vertex in $P$, or alternatively the level of the maximum degree vertex of $P$.
    Let $\hat{\delta}_{j}$ denote the estimate after executing $\dijkstra$ from all $w \in D_j$.
    We now proceed by case analysis on the levels of the vertices comprising shortest path $P$.

    \paragraph{Case 2.3.1: $\level(u) = \level(v) = k$}

    We show a $+4$ additive error in this case.
    Without loss of generality, assume $\level(u_2) \leq \level(v_2)$.
    Then, for length $4$ paths when $u_3 = v_3$, we actually have $P \subset E_{\level(u_2)}$ so that after the execution of $\dijkstra$ at the $\level(u_2)$ iteration, $\hat{\delta}_{\level(u_2)}(u_2^*, v) \leq \delta(u_2, v) + 1$ where $u_2^* = r\left(u_2, D_{\level(u_2)}\right)$.
    Then, in the $k$-th iteration, when we execute $\dijkstra$ from $v$, we have access to the edges $(v, u_2^*), (u_2^*, u_2), (u_2, u)$ to find a path of length at most $\delta(u, v) + 2$.

    Now, suppose we have a length 5 path so that $u_3 \neq v_3$.
    If $P \notin E_{\level(u_2)}$ then $\level(u_3), \level(v_3) < \level(u_2)$.
    Then, $P \subset E_{\level(v_3)}$ so that from $v_3^* = r\left(v_3, D_{\level(v_3)}\right)$ the path $(v_3^*, v_3) \circ P_{v_3, u_2} \circ (u_2, u_2^*) \subset G_{\level(v_3), v_3^*}$ and $\dijkstra$ computes an estimate,
    \begin{equation*}
        \hat{\delta}_{\level(v_3)}(v_3^*, u_2^*) \leq \delta(v_3, u_2) + 2
    \end{equation*}

    Then, from $u_2^*$, since $(u_2^*, v_3^*, v_3) \circ P_{v_3, v} \subset G_{\level(u_2), u_2^*}$, $\dijkstra$ computes,
    \begin{equation*}
        \hat{\delta}_{\level(u_2)}(u_2^*, v) \leq \hat{\delta}_{\level(v_3)}(u_2^*, v_3^*) + 1 + \delta(v_3, v) \leq \delta(u_2, v_3) + 3 + \delta(v_3, v) \leq \delta(u_2, v) + 3
    \end{equation*}

    Finally, in the $k$-th iteration $(v, u_2^*, u_2, u) \subset G_{k, v}$ so $\dijkstra$ finds a path of length at most $\delta(u, v) + 4$.

    \paragraph{Case 2.3.2: $\level(P) = 1$ and $\min(\level(u), \level(v)) \leq k - 1$}

    Recall that $P = (u, u_2, u_3, v_3, v_2, v)$. 
    For $\delta(u, v) = 4$, note that $u_3 = v_3$.

    This case is handled by Phase 2.1.
    Without loss of generality, assume $\deg(u) \geq \deg(v)$ so that $\deg(u) \geq s_{k - 1}$.
    Since $\level(P) = 1$, there is some vertex $w$ such that $\deg(w) \geq s_1$.
    Then, let $w^* = r(w, D_1) \in D_1$ and $u^* = r(u, D_{k - 1}) \in D_{k - 1}$.
    Since $P \subset G_x$,
    \begin{align*}
        B(w^*, u^*) =  \hat{\delta}_{s_1}(w^*, u^*) &\leq \delta(w, u) + 2 \leq C + 2 \\
        A(w^*, v) = \hat{\delta}_{s_1}(w^*, v) &\leq \delta(w, v) + 1 \leq C + 1
    \end{align*}

    Therefore, since both entries are finite, we obtain in the \minplus product an estimate at most,
    \begin{equation*}
        \hat{\delta}(u^*, v) \leq \delta(w, u) + \delta(w, v) + 3 \leq \delta(u, v) + 3
    \end{equation*}

    Then, in Phase 2.3, when we executing $\dijkstra$ from $G_{k, v}$, we take the edge $(v, u^*)$ and $(u^*, u)$ to find,
    \begin{equation*}
        \hat{\delta}(v, u) \leq \hat{\delta}(u^*, v) + 1 \leq \delta(u, v) + 4 \leq 2 \delta(u, v)
    \end{equation*}

    \paragraph{Case 2.3.3: $\level(P) \geq 2$ and $\level(u_2) = \level(v_2) = 2$}

    Recall that $P = (u, u_2, u_3, v_3, v_2, v)$. 
    For $\delta(u, v) = 4$, note that $u_3 = v_3$.

    We handle this special case by noting that the number of distinct levels in the path $P$ is small.
    Without loss of generality, suppose $\level(v) \leq \level(u)$.
    Then, after the 2nd iteration,
    \begin{equation*}
        \hat{\delta}_{2}(u_2^*, v) \leq \delta(u_2, v) + 1
    \end{equation*}

    as $P \subset E_2$ and $E_2 \cup E^* \subset G_{2, u_2^*}$.
    Then, after the $\level(v)$-th iteration, we take the path $(v^*, u_2^*, u_2, u)$ to obtain,
    \begin{equation*}
        \hat{\delta}_{\level(v)}(v^*, u) \leq \hat{\delta}_{2}(v^*, u_2^*) + 2 \leq \delta(u_2, v) + 4 \leq \delta(u, v) + 3
    \end{equation*}

    In the final iteration, we take the path $(v, u^*, u)$ to obtain a $+4$ approximation.
    
    \paragraph{Case 2.3.4: $\level(P) \geq 2$ and $\min(\level(u_2), \level(v_2)) = 2, \max(\level(u_2), \level(v_2)) \geq 3$}

    Recall that $P = (u, u_2, u_3, v_3, v_2, v)$. 
    For $\delta(u, v) = 4$, note that $u_3 = v_3$.

    Without loss of generality, assume $\level(u_2) = 2$ and $\level(v_2) \geq 3$.
    If $\level(v) \leq \level(u)$, then we can proceed in the above case to obtain a $+4$ approximation.
    Thus, assume $\level(u) < \level(v)$.
    We can now break our analysis into two sub-cases.
    In the remaining cases, we use the Depth 3 BFS search from Line \ref{line:2-approx:depth-3-bfs} to ensure correctness.

    Suppose $\level(v_2) \leq \level(u) < \level(v)$.
    This case is illustrated in Figure \ref{fig:depth-3-v2}.
    Since $\level(v_2) \geq 3$, we have $\deg(v_2) < s_2$.
    Then, $\bfs(G_2, v_2, 3)$ finds $u_2$ so that $\hat{\delta}(v_2, u_2) = \delta(v_2, u_2) = 3$.
    Then, denote $v_2^* = r(v_2, D_{\level(v_2)})$.
    After Line \ref{line:2-approx:truncated-bfs-hitting-set} since $v_2 \in q(v_2^*, D_{\level(v_2)})$, we obtain $\hat{\delta}_{\level(v_2)}(v_2^*, u^*) \leq \delta(v_2^*, u) + 1$ via the path $(v_2^*, u_2, u, u^*)$.
    Then in the $\level(u)$-th iteration, since $\level(u) < \level(v)$, we execute $\dijkstra$ from $u^*$ and use the path $(u^*, v_2^*, v_2, v)$ to find $\hat{\delta}_{\level(u)}(u^*, v) \leq \delta(u, v) + 3$.
    Then, executing $\dijkstra$ from $v$ in the $k$-th iteration, we take the path $(v, u^*, u)$ and find an estimate, $\hat{\delta}(u, v) \leq \delta(u, v) + 4$.

    Finally, suppose $\level(u) < \level(v_2)$.
    First, if $\level(u) = 2 = \level(P)$, then from Lemma \ref{lemma:sparse-apasp-approx} we have that after the second iteration $\hat{\delta}_{2}(u^*, v) \leq \delta(u, v) + 1$ and obtain a $+2$ approximation in the $k$-th iteration via the path $(v, u^*, u)$.
    We can therefore assume $\level(u) \geq 3$.

    Then, we execute a Depth 3 $\bfs$ from $u$ as illustrated in Figure \ref{fig:depth-3-u} so that after Line \ref{line:2-approx:depth-3-bfs} we have $M_{u, 3}(v_3) = 3$ (for length 4 paths, we in fact obtain $M_{u, 3}(v_2) = 3$).
    In particular, after Line \ref{line:2-approx:truncated-bfs-hitting-set}, $\hat{\delta}(u^*, v_3) \leq \delta(u, v_3) + 1$ where $u^* = r(u, D_{\level(u)})$.

    Since $\level(u) < \level(v_2), \level(v)$, in the $\level(u)$-th iteration, we find $v$ from $u^*$ via the path $(u^*, v_3, v_2, v)$ so that $\hat{\delta}_{\level(u)}(u^*, v) \leq \delta(u, v) + 1$.
    Finally, in the $k$-th iteration, from $v$ we take the path $(v, u^*, u)$ to find,
    \begin{equation*}
        \hat{\delta}(u, v) \leq \delta(u, v) + 2
    \end{equation*}
    and obtain a $+2$ approximation.

    \paragraph{Case 2.3.5: $\level(P) \geq 2$ and $\min(\level(u_2), \level(v_2)) \geq 3$}

    Recall that $P = (u, u_2, u_3, v_3, v_2, v)$. 
    For $\delta(u, v) = 4$, note that $u_3 = v_3$.

    Without loss of generality, assume $\level(u_2) \leq \level(v_2)$.
    Since $u, u_2, v_2, v$ all have degree at most $s_2$, we can freely afford Depth 3 $\bfs$ from any of these vertices in the graph $G_2$.

    Without loss of generality, suppose $\level(u) \leq \level(v)$.
    For length 4 paths, we have computed a $\hat{\delta}(u, v_2) = 3$ after Line \ref{line:2-approx:depth-3-bfs} so that after $\level(u)$-th iteration, $\hat{\delta}(u^*, v) \leq \delta(u, v) + 1$ and we obtain a $+2$ approximation in the $k$-th iteration via the path $(v, u^*, u)$.
    Similarly, if $\level(u) \leq \level(v_2)$, we obtain a $+2$ approximation for length 5 paths.

    Thus, we have $\level(v_2) < \level(u)$.
    Then, as in the above case, we compute a Depth 3 $\bfs$ from $v_2$ so that after the $\level(v_2)$-th iteration, $\hat{\delta}_{\level(v_2)}(v_2^*, u) \leq \delta(v_2, u) + 1$.
    Then, after the $\level(u)$-th iteration, $\hat{\delta}_{\level(u)}(u^*, v) \leq \delta(u, v) + 3$ and we obtain a $+4$ approximation in the $k$-th iteration.

    \paragraph{Case 2.4: $\delta(u, v) = 3$}

    For length 3 paths, we prove that we find a $+2$ approximation.
    Denote the path $P = (u, u_2, v_2, v)$.

    We can assume that the maximum degree vertex of $P$ is either $u_2, v_2$.
    Otherwise, let $u$ be the endpoint of maximum degree.
    We compute $\hat{\delta}_{\level(u)}(u^*, v) \leq \delta(u, v) + 1$ and obtain a $+2$ approximation in the $k$-th iteration via path $(v, u^*, u)$.
    
    Suppose without loss of generality $\level(u) \leq \level(v)$.
    Then, $M_{u, 2}(v_2) = 2$ so that after Line \ref{line:2-approx:truncated-bfs-hitting-set},
    \begin{equation*}
        \hat{\delta}(u^*, v_2) \leq 3
    \end{equation*}
    where $u^* = r(u, D_{\level(u)})$.
    Executing $\dijkstra$ from $u^*$, we can take the path $(u^*, v_2, v)$ and obtain $\hat{\delta}_{\level(u)}(u^*, v) \leq \delta(u, v) + 1$.
    Then, in the $k$-th iteration, we take the path $(v, u^*, u)$ and obtain a $+2$ approximation.

    \paragraph{Case 2.5: $\delta(u, v) = 2$}

    For length 2 paths, we in fact find the exact distance.
    Since $P \subset G_x$, in Line \ref{line:2-approx:depth-2-bfs} we compute depth 2 $\bfs$ from each vertex, and therefore find $\hat{\delta}(u, v) = \delta(u, v) = 2$.
\end{proof}

\paragraph*{Time Complexity}
\begin{proof}

    We analyze the time complexity of our algorithm.
    $x, C$ are parameters that will be optimized.

    \paragraph{Phase 1}
    Following the same arguments as Phase 1 of the randomized algorithm, we can bound the running time of Phase 1 as,
    \begin{equation*}
        \bigtO{n^{\omega(1, 1-x, 1)}}
    \end{equation*}

    \paragraph{Phase 2}
    
    In Phase 2, we have a graph $G_x$ of maximum degree $n^x$ and therefore $m \leq n^{1 + x}$.
    From Lemma \ref{lemma:degree-decomposition}, the execution of $\decompose$ requires $O(m) = O(n^{1+x})$ time.

    In Phase 2.1, we can upper bound $|D_1| = \tO{n^{1 - \frac{(k - 1)}{k}x}}$ by Lemma \ref{lemma:degree-decomposition}.
    Likewise, $|D_{k - 1}| = \tO{n^{1 - \frac{x}{k}}}$.
    Executing $\bfs$ on graph $G_{s_1}$ with $O(n^{1 + x})$ edges from each $w \in D_1$ requires $\tO{n^{2 + \frac{x}{k}}}$ time.

    Now, computing the \minplus product of $A, B$ requires time,
    \begin{equation*}
        \bigtO{n^{\omega\left(1, 1 - \frac{k - 1}{k} x, 1 - \frac{x}{k} \right)}} = \bigtO{n^{\omega(1, 1 - x, 1)}}
    \end{equation*}
    by Proposition \ref{prop:2-approx-min-plus-convexity}.

    In Phase 2.2, computing Line \ref{line:2-approx:depth-2-bfs} requires time $O(n^{2x})$ while computing Line \ref{line:2-approx:depth-3-bfs} requires time $\bigO{n^{\frac{3k - 4}{k}x}}$.
    Overall, Phase 2.2 requires time,
    \begin{equation*}
        \bigO{n^{1 + 2x} + n^{1 + \frac{3k - 4}{k}x}}
    \end{equation*}

    Finally, we turn our attention to Phase 2.3.
    Fix some iteration $i$.
    Since the sets $q(w, D_i)$ are disjoint, $\sum_{w \in D_i} |q(w, D_i)| \leq n$ so that Line \ref{line:2-approx:truncated-bfs-hitting-set} requires time $O(n^2)$ over all $w \in D_i$.
    We bound $|D_i| \leq \tO{n^{1 - \frac{k-i}{k}x}}$ by Lemma \ref{lemma:degree-decomposition}.
    Each invocation of $\dijkstra$ requires time $|E_i| = O(n^{1 + \frac{k-i+1}{k}x})$ so that across all $w \in D_i$, all invocations of $\dijkstra$ can be completed in time $\tO{n^{2 + \frac{x}{k}}}$.

    We ignore Phase 3 for now as we will find that it is not the computational bottleneck.
    Phases 1 and 2 require time,
    \begin{equation*}
        \bigtO{n^{\omega(1, 1-x, 1)} + n^{1 + 2x} + n^{1 + \frac{3k - 4}{k}x} + n^{2 + \frac{x}{k}}}
    \end{equation*}

    Using \cite{Complexity}, we choose $x = 0.4414248$ and $k = 7$ to obtain the time complexity $\tO{n^{2.07203166}}$ for Phase 1 and 2.
    Finally, we conclude the proof by noting that $\boundedAdditiveAPASP(G, 6, C)$ requires time $\tO{n^{2.05794292}}$ by Corollary \ref{cor:bounded-additive-examples} and $\denseAPASP(G, C)$ requires time $\tO{n^{35/17}} = \tO{n^{2.0589}}$ by Lemma \ref{lemma:dhz-apasp}.
\end{proof}

\section{\texorpdfstring{$(2, 0)$}{(2, 0)}-Multiplicative Approximations for Long Paths}
\label{sec:bounded-additive-approx}

In this section, we present algorithms for computing $(2, 0)$ approximations for paths of at least some length.
We begin with a combinatorial result that gives a $(2, 0)$ multiplicative approximation for paths of length at least $k$ for even values of $k \geq 4$.

\subsection{Multiplicative Approximation for Long Paths}

\begin{theorem}
    \label{thm:mult-approx-bk}
    Let $k \geq 4$ be an even integer.
    \Cref{alg:mult-approx-bk} computes a $(2, 0)$ approximation for all paths of length $\delta(u, v) \geq k$ in expected time $\bigtO{n^{2 + \frac{1}{2(k - 1)}}}$.
\end{theorem}

In \Cref{sec:application-mult-approx-long-path}, we improve the above bound for $k=4$ and $5$.

\paragraph{High Level Overview}

As a simple example, we begin with by showing that we can obtain a $(2, 0)$ approximation for $\delta(u, v) \geq 4$ in time $\tO{n^{13/6}}$.
Decompose the graph with degree thresholds $s_1 = n^{5/6}, s_2 = n^{4/6}, \dotsc, s_5 = n^{1/6}$.
Let $P$ be a shortest path of length at least 4.
If $P \subset E_3$, then we execute $\bkAPASP(G_3)$ in time $n^{13/6}$ where $G_3 = (V, E_3)$ has maximum degree $s_2 = n^{4/6}$.
Otherwise, $P \not\subset E_3$ and vertex $w \in P$ of maximum degree has $\deg(w) \geq s_2$ and is dominated by $w^* = r(w, D_{\level(w)})$.
Suppose $\deg(v) \geq \deg(u)$.
Let $z = b(v, P)$ be the blocking vertex so $\deg(z) \geq s_5$ and $z^* = r(z, D_{\level(z)}) \in D_5$.
From $w^*$, since $P \subset E_{\level(w)}$,
\begin{align*}
    \hat{\delta}(w^*, z^*) &\leq \delta(w, z) + 2 \\
    \hat{\delta}(w^*, v) &\leq \delta(w, v) + 1
\end{align*}
Since $D_{\level(w)} \subset D_2$, and $2 + 5 + 6 \leq 13$, the path $(v, w^*, z^*, z) \circ P_{z, u}$ exists in $G_{6, v}$ (see Line 49 of \Cref{alg:mult-approx-bk}) and $\dijkstra$ computes a $+4$ approximation.

For the remaining overview, assume $k\geq 6$.
We want to obtain a $(2, 0)$ approximation for paths of length at least $\multapproxlimit$.
Of course, it suffices to compute a $+\multapproxlimit$ approximation. 
Let $P$ be a path.
We will choose $l$ levels with degree thresholds $s_1, s_2, \dotsc, s_{l - 1}$ (note that $s_l=0$).
In order to compute a $+\multapproxlimit$ approximation, it suffices to have for all $w$ with $\deg(w) \geq s_{l - \frac{\multapproxlimit - 4}{2}}$ and for all $v \in V$, an estimate
\begin{equation*}
    \hat{\delta}(w^*, v) \leq \delta(w, v) + 5
\end{equation*}
Then, following a similar argument to \Cref{lemma:block-approx-error}, we obtain a $+\multapproxlimit$ additive approximation at the $l$-th level.
In order to guarantee a $+5$ error at the $(l - \frac{\multapproxlimit - 4}{2})$ level, we require the edge set $D_{l_0} \times D_{l - \frac{\multapproxlimit - 2}{2}}$ for some choice of $l_0$.

Suppose $P \subset E_{l_0 + 1}$.
Then, we run $\bkAPASP(G_{l_0 + 1})$ and compute a $(2, 0)$ approximation on $P$.
Otherwise, $P \not\subset E_{l_0 + 1}$ so the maximum degree vertex $z \in P$ has $\deg(z) \geq s_{l_0}$ and there is some vertex $z^* = r(z, D_{\level(z)})$.
Let $P_{w, v}$ be the sub-path from $w$ to $v$ and let $y = b(w, P_{w, v})$ be the blocking vertex so that $\level(y) \leq l - \frac{\multapproxlimit - 2}{2}$ and there is $y^* = r(y, D_{l - \frac{\multapproxlimit - 2}{2}})$.
Since $D_{\level(z)} \subset D_{l_0}$ and $P \subset E_{\level(z)}$,
\begin{align*}
    \hat{\delta}(w^*, z^*) &\leq \delta(w, z) + 2 \\
    \hat{\delta}(z^*, y^*) &\leq \delta(z, y) + 2
\end{align*}
If $D_{l_0} \times D_{l - \frac{\multapproxlimit - 2}{2}} \subset G_{l - \frac{\multapproxlimit - 4}{2}, w^*}$, the path $(w^*, z^*, y^*, y) \circ P_{y, v} \subset G_{l - \frac{\multapproxlimit - 4}{2}, w^*}$ and $\dijkstra$ computes an estimate at most,
\begin{equation*}
    \hat{\delta}(w^*, v) \leq \delta(w, v) + 5
\end{equation*}

The only restriction on $l_0$ therefore is that
\begin{equation*}
    l_0 + \left(l - \frac{\multapproxlimit - 4}{2} \right) + \left(l - \frac{\multapproxlimit - 2}{2} \right) \leq 2 l + 1
\end{equation*}
or
\begin{equation*}
    l_0 \leq \multapproxlimit - 2
\end{equation*}
which justifies our choice of $l_0 = \multapproxlimit - 2$.
We then choose $l$ to optimize the running time trade-off between the executions of $\dijkstra$ which requires time $\tO{n^{2 + 1/l}}$ and $\bkAPASP$ which requires time $\bigtO{n^{1.5} s_{l_0}} = \bigtO{n^{1.5} s_{\multapproxlimit - 2}} = \bigtO{n^{2.5 - \frac{\multapproxlimit - 2}{l}}}$.

\paragraph{Algorithm}

Our algorithm begins by initializing the distance estimates with the adjacency matrix and setting parameters $l_0, l$.
Next, \Cref{line:mult-approx-bk:decompose} decomposes the graph with $l - 1$ degree thresholds.

In Phase 1, at each level $j$, we compute $\dijkstra$ from $w \in D_j$ on the graph $G_{j, w}$ in \Cref{line:mult-approx-bk:dijkstra}.
The graph $G_{j, w}$ consists of edge set $E_j \cup E^* \cup (w \times V)$ and $D_{j_1} \times D_{j_2}$ for all $j + j_1 + j_2 \leq 2l - 1$.

Finally, we take the minimum between the estimate computed above and the estimate given by $\bkAPASP$ on the graph $G_{k_0 + 1} = (V, E_{k_0 + 1})$.

\IncMargin{1em}
\begin{algorithm}[ht]
\SetKwInOut{Input}{Input}
\SetKwInOut{Output}{Output}
\Input{Unweighted, undirected graph $G = (V, E)$ with $n$ vertices; even integer $k$}
\Output{Distance estimate $\hat{\delta}: U \times V \rightarrow \Z$ such that $\delta(u, v) \leq \hat{\delta}(u, v)$ for all $u, v \in V$ and $\hat{\delta}(u, v) \leq 2 \delta(u, v)$ whenever $\delta(u, v) \geq k$}

\BlankLine

\textcolor{blue}{Phase 0: Set up and Decompose Graph}

$\hat{\delta}(u, v) \gets \begin{cases}
    1 & (u, v) \in E \\
    \infty & \otherwise
\end{cases}$

$l \gets 2(k - 1)$

$l_0 \gets k - 2$

$s_{i} \gets n^{1 - \frac{i}{l}}$ for all $1 \leq i \leq l - 1$

$(D_{1}, D_{2}, \dotsc, D_{l}), (E_{1}, E_{2}, \dotsc, E_{l}), E^* \gets \decompose(G, (s_{1}, s_{2}, \dotsc, s_{l-1}))$
\label{line:mult-approx-bk:decompose}

\textcolor{blue}{Phase 1: Compute Distance Estimates on High Degree Paths}

\For{$1 \leq j \leq l$}{
    \For{$w \in D_{j}$}{
        $G_{j, w} \gets \left(V, E_j \cup \left( \bigcup_{j + j_1 + j_2 \leq 2l + 1} D_{j_1} \times D_{j_2} \right) \cup E^* \cup (w \times V) \right)$
        
        $\hat{\delta} \gets \dijkstra(G_{j, w}, w, \hat{\delta})$
        \label{line:mult-approx-bk:dijkstra}
    }
}

$\hat{\delta} \gets \min \left( \hat{\delta}, \bkAPASP(G_{l_0 + 1}) \right)$ where $G_{l_0 + 1} \gets (V, E_{l_0 + 1})$
\label{line:mult-approx-bk:bk-apasp}

\caption{$\longMultiplicativeAPASP(G, k)$} 
\label{alg:mult-approx-bk}
\end{algorithm}
\DecMargin{1em}

\paragraph{Warm-Up: $k = 4$}

We first show that we can obtain a $\tO{n^{13/6}}$ algorithm for computing $(2, 0)$ approximation for paths of length at least 4.

\begin{proof}
    Since $k = 4$, we have $l = 6$ and $l_0 = 2$.
    We decompose the graph with degree thresholds $s_1 = n^{5/6}, s_2 = n^{4/6}, \dotsc, s_5 = n^{1/6}$.
    
    We execute $\bkAPASP$ on $G_{3} = (V, E_3)$.
    Therefore, if $P \subset E_3$, we obtain a $(2, 0)$ approximation via $\bkAPASP$.
    
    Otherwise, $P \not\subset E_3$.
    In this case, we claim to compute a $+4$ approximation, which implies a $(2, 0)$ approximation for $\delta(u, v) \geq 4$.
    The vertex $w \in P$ of maximum degree has $\deg(w) \geq s_2$ and is dominated by $w^* = r(w, D_{\level(w)})$.
    
    Without loss of generality, suppose $\deg(v) \geq \deg(u)$.
    Let $z = b(v, P)$ be the blocking vertex from $v$.
    Since $P \not\subset E_3$, $z \neq u$ and $\deg(z) \geq s_5$.
    
    If $w = z$ then since $P \subset E_{\level(w)}$, $B(P) = \set{u, v, w}$ (\Cref{def:blocking-vertices}) and we obtain a $+4$ approximation.
    In particular,
    \begin{equation*}
        \hat{\delta}_{\level(w)}(w^*, v) \leq \delta(w, v) + 1
    \end{equation*}
    where $w^* = r(w, D_{\level(w)})$.
    If $\level(v) = 6$, then,
    \begin{equation*}
        \hat{\delta}(v, u) \leq \hat{\delta}_{\level(w)}(v, w^*) + 1 + \delta(w, u) \leq \delta(u, v) + 2
    \end{equation*}
    Otherwise, if $v^* = r(v, D_{\level(v)})$,
    \begin{equation*}
        \hat{\delta}_{\level(v)}(v^*, u) \leq \hat{\delta}_{\level(w)}(v^*, w^*) + 1 + \delta(w, u) \leq \delta(v, u) + 3
    \end{equation*}
    Of course, this implies a $+4$ approximation when executing $\dijkstra$ on $G_{6, u}$ by taking the path $(u, v^*, v)$.

    Therefore, $w \neq z$.
    From $w^*$, since $P \subset E_{\level(w)}$,
    \begin{align*}
        \hat{\delta}(w^*, z^*) &\leq \delta(w, z) + 2 \\
        \hat{\delta}(w^*, v) &\leq \delta(w, v) + 1
    \end{align*}
    Since $D_{\level(w)} \subset D_2$ and $z^* \in D_5$, and $2 + 5 + 6 \leq 13$, the path $(v, w^*, z^*, z) \circ P_{z, u}$ exists in $G_{6, v}$ and $\dijkstra$ computes a $+4$ approximation.
    \begin{equation*}
        \hat{\delta}(v, u) \leq \hat{\delta}(v, w^*) + \hat{\delta}(w^*, z^*) + 1 + \delta(z, u) = \delta(v, u) + 4
    \end{equation*}

    To analyze the running time, it suffices to observe that for each $i$, $|D_i| = \tO{n^{i/6}}$ while $G_{i, w}$ has $O(n^{2 - (i-1)/6})$ edges so that the computation of $\dijkstra$ requires $\tO{n^{13/6}}$.
    Finally, computing $\bkAPASP$ on $G_3 = (V, E_3)$ requires $\tO{m \sqrt{n}} = \tO{n^{13/6}}$.
\end{proof}

\paragraph{Correctness}

Finally, we prove \Cref{thm:mult-approx-bk} for general $k \geq 6$.

\begin{proof}

First, we note that $\delta(u, v) \leq \hat{\delta}(u, v)$ since $\bkAPASP$ is correct and every edge weight in $G_{j, w}$ is a distance estimate computed by some path in the original graph $G$.

Suppose $P \subset E_{l_0 + 1}$.
By the correctness of $\bkAPASP$, $\hat{\delta}(u, v) \leq 2 \delta(u, v)$.

Then, suppose $P \not\subset E_{l_0 + 1}$.
Suppose $\level(v) \leq \level(u)$.
Recall from Lemma \ref{lemma:block-approx-error} the blocking vertices $B(P) = \set{x_0, x_1, \dotsc, x_t}$ and levels $L_B(P)$ of path $P$.

Let $a = \min_{\level(x_i) \leq l - \frac{l_0}{2} + 1} i$ be the minimum index of an element in the blocking set such that $\level(x_i) \leq l - \frac{l_0}{2} + 1$. Since $P$ is not contained in $E_{l_0 + 1}$, we can assume $a$ exists and $a \geq 1$.
Otherwise, every edge is in $E_{l - \frac{l_0}{2} + 2} = E_{(3\multapproxlimit+2)/2} \subset E_{\multapproxlimit - 1} = E_{l_0 + 1}$.
Since $l \geq \level(x_1) > \level(x_2) > \dotsc > \level(x_t) \geq 1$ (\Cref{lemma:prop-blocking}), we can upper bound,
\begin{equation*}
    a \leq l - \left(l - \frac{l_0}{2} + 1\right) = \frac{l_0}{2} - 1 = \frac{\multapproxlimit}{2} - 2
\end{equation*}

Let $x_a \in B(P)$ be the corresponding vertex in $B(P)$.
Since $P$ has an edge not in $E_{l_0 + 1}$, the last blocking vertex of minimum level must have $\level(x_t) \leq l_0$.
Recall that we denote $v^* = r\left(v, D_{\level(v)}\right)$ for any vertex $v$.
Since $P \subset E_{\level(x_t)}$, we again have $\hat{\delta}_{\level(x_t)}(x_t^*, x_a^*) \leq \delta(x_t, x_a) + 2$ and $\hat{\delta}_{\level(x_t)}(x_t^*, x_{a+1}^*) \leq \delta(x_t, x_{a + 1}) + 2$.

Consider the $\level(x_a)$-th iteration.
The edges $D_{\level(x_t)} \times D_{\level(x_{a + 1})}$ are in $G_{\level(x_a), x_a}$ as,
\begin{equation*}
    \level(z) + \level(x_{a + 1}) + \level(x_{a}) \leq l_0 + \left(l - \frac{l_0}{2}\right) + \left(l - \frac{l_0}{2} + 1\right) = 2 l + 1
\end{equation*}
Next observe that $x_t \in P_{x_a, x_{a + 1}}$ (\Cref{lemma:prop-blocking}), as each blocking vertex $x_{i + 1} = b(x_i, P_{x_i, x_{i - 1}})$ is by definition in the sub-path $P_{x_i, x_{i - 1}}$.
Therefore we have the path $(x_a^*, x_t^*, x_{a + 1}^*, x_{a + 1}) \circ P_{x_{a + 1}, x_{a - 1}}$ so $\dijkstra$ computes an estimate at most,
\begin{align*}
    \hat{\delta}_{\level(x_a)}(x_a^*, x_{a - 1}) &\leq \hat{\delta}_{\level(x_t)}(x_a^*, x_t^*) + \hat{\delta}_{\level(x_t)}(x_t^*, x_{a + 1}^*) + 1 + \delta(x_{a + 1}, x_{a - 1}) \\ 
    &\leq \delta(x_a, x_t) + \delta(x_t, x_{a + 1}) + \delta(x_{a + 1}, x_{a - 1}) + 5 \\ 
    &\leq \delta(x_a, x_{a - 1}) + 5
\end{align*}

Then, following a similar argument to the inductive step of Lemma \ref{lemma:block-approx-error}, we claim the following for all $1 \leq j \leq a$.
\begin{equation*}
    \hat{\delta}_{\level(x_j)}(x_j^*, x_{j - 1}) \leq \delta(x_j, x_{j - 1}) + 2(2 + (a - j)) + 1
\end{equation*}
where we have established the base case $j = a$ above.
We now proceed by induction for $j < a$.
Consider an execution of $\dijkstra$ from $x_j^*$ in $G_{\level(x_j), x_j^*}$.
Let $x_{j + 1}$ be the blocking vertex from the previous iteration.
We take the edges $(x_j^*, x_{j + 1}^*)$, $(x_{j + 1}^*, x_{j + 1}) \in E^*$, and the remaining edges in $E_{\level(x_j)}$.
By induction, we have,
\begin{align*}
    \hat{\delta}_{\level(x_j)} (x_j^*, x_{j - 1}) &\leq \hat{\delta}_{\level(x_{j + 1})} (x_j^*, x_{j + 1}^*) + 1 + \delta(x_{j + 1}, x_{j - 1}) \\
    &\leq \delta(x_j, x_{j + 1}) +  2 (2 + (a - (j + 1))) + 3 + \delta(x_{j + 1}, x_{j - 1}) \\
    &\leq \delta(x_{j}, x_{j - 1}) + 2 (2 + (a - j)) + 1
\end{align*}

Thus, we have,
\begin{equation*}
    \hat{\delta}_{\level(v)}(v^*, u) = \hat{\delta}_{\level(x_1)} (x_1^*, x_0) \leq \delta(v, u) + 2 (a + 1)  + 1
\end{equation*}

From $u$, we take the edge $(u, v^*)$, followed by $(v^*, v) \in E^*$ so that,
\begin{align*}
    \hat{\delta}(u, v) &\leq \hat{\delta}_{\level(v)}(u, v^*) + 1 \\
    &\leq \delta(u, v) + 2 (a + 1) + 2 \\
    &\leq \delta(u, v) + \multapproxlimit
\end{align*}
\end{proof}

\paragraph{Time Complexity}

\begin{proof}
    We set $l_0 = k - 2$ and leave $l$ to be optimized.

    $\decompose$ requires $\tO{l n^2}$ time.
    Over all $D_j$, the invocations of $\dijkstra$ requires $\tO{n^{2 + 1/l}}$ time, as for each i, $|D_i| = \tO{n/s_i}$ and $|E_i| = O(n s_{i - 1})$ while $|D_{j_1} \times D_{j_2}| = \bigtO{\frac{n}{s_{j_1}} \times \frac{n}{s_{j_2}}} = \bigtO{n^{\frac{j_1 + j_2}{l}}} = \bigtO{n^{2 - \frac{j - 1}{l}}}$ as $j_1 + j_2 \leq 2l + 1 - j$.
    
    $\bkAPASP$ on $G_{l_0 + 1}$ requires $s_{l_0} n^{1.5} = n^{1.5 + 1 - \frac{k - 2}{l}}$ time.
    Balancing,
    \begin{equation*}
        2 + \frac{1}{l} = 2.5 - \frac{k - 2}{l}
    \end{equation*}
    we solve to obtain,
    \begin{align*}
        2 l + 1 &= \frac{5}{2} l - k + 2 \\
        \frac{l}{2} &= k - 1 \\
        l &= 2 (k - 1)
    \end{align*}

    Thus, the algorithm runs in time $\tO{n^{2 + 1/l}} = \tO{n^{2 + \frac{1}{2 k - 1}}}$
\end{proof}

\subsection{Faster Multiplicative Approximations for Long Paths with FMM}

In this section, we improve upon \Cref{thm:mult-approx-bk} using fast matrix multiplication.

\begin{theorem}
    \label{thm:mult-approx-bk-fmm}
    Let $k \geq 4$ be an even integer.
    \Cref{alg:mult-approx-bk} computes a $(2, 0)$ approximation for all paths of length $\delta(u, v) \geq k$ in expected time $\tO{n^{2 + \frac{x}{l - k + 2}}}$ where $x, l$ are chosen to minimize,
    \begin{equation*}
        \max \left( \omega\left( 1 - \frac{k - 2}{2(l - k + 2)} x, 1 - x, 1 - \frac{k - 4}{2(l - k + 2)} x \right), 2 + \frac{x}{l - k + 2}, 1.5 + x \right)
    \end{equation*}
\end{theorem}

In \Cref{tbl:2-approx-geq-beta} and \Cref{prop:mult-approx-bk-fmm-examples}, we exhibit a few running times for $6 \leq k \leq 10$.

\paragraph{High Level Overview}

We begin by showing that we can obtain a $(2, 0)$ approximation for $\delta(u, v) \geq 4$ in time $\tO{n^{2.01973523}}$.
Set parameters $x = 0.51973523, M = 35$ and $l = 29$ so that $l_0 = 2$.
Decompose the graph with degree thresholds $s_{2} = n^x,  s_3 = n^{26x/27}, \dotsc, s_{28} = n^{x/27}$.
Let $P$ be a shortest path of length at least 4.
If $P \subset E_3$, then we execute $\bkAPASP(G_3)$ in time $n^{1.5 + x} = n^{2.01973523}$ where $G_3 = (V, E_3)$ has maximum degree $s_2 = n^{x}$.
If $\delta(u, v) \geq M$, then $\denseAPASP$ computes a $+M$ approximation and we obtain a $(2, 0)$ approximation.
Otherwise, $P \not\subset E_3$ is of length at most $M$ and vertex $w \in P$ of maximum degree has $\deg(w) \geq s_2$.
Suppose $\deg(v) \geq \deg(u)$.
Let $z = b(v, P)$ be the blocking vertex so $\deg(z) \geq s_{28}$ and $z^* = r(z, D_{\level(z)}) \in D_{28}$.
Let $\frac{n}{2^i} \leq \deg(w) \leq \frac{n}{2^{i - 1}}$.
Since $P \subset F_i$, we obtain the estimates from $w^*$,
\begin{align*}
    \hat{\delta}(w^*, z^*) &\leq \delta(w, z) + 2 \leq M + 2 \\
    \hat{\delta}(w^*, v) &\leq \delta(w, v) + 1 \leq M + 1
\end{align*}
Then, the bounded \minplus product computes,
\begin{equation*}
    \hat{\delta}(v, z^*) \leq \delta(v, z) + 3
\end{equation*}
as $v \in V, z^* \in D_{28}, w^* \in C_i$.
Finally, the path $(v, z^*, z) \circ P_{z, u}$ exists in $G_{6, v}$ and $\dijkstra$ computes a $+4$ approximation.

For the remaining overview, assume $k \geq 6$.
In \Cref{alg:mult-approx-bk}, we obtained a $+5$ approximation after the $\left(l - \frac{\multapproxlimit - 4}{2}\right)$ iteration by including edge sets $D_{i} \times D_{l - \frac{\multapproxlimit - 2}{2}}$ for all $i \leq l_0$ in the edges of $G_{l - \frac{\multapproxlimit - 4}{2}, w}$ for $w \in D_{l - \frac{\multapproxlimit - 4}{2}}$.
Using fast matrix multiplication, we will pre-compute the distances obtained from these edges and instead encode them in smaller edge set $w \times V$.

Consider a shortest path $P$.
If $P \subset E_{l_0 + 1}$, then we again invoke $\bkAPASP$ to compute a $(2, 0)$ approximation.
Therefore, suppose $P \not\subset E_{l_0 + 1}$.
For a large enough constant $M$, we can invoke $\denseAPASP(G, M)$ (\Cref{lemma:dhz-apasp}) or $\additiveAPASP$ (\Cref{thm:gen-additive-apasp}) to compute a $+M$ approximation without affecting the overall running time.
Therefore, we can also assume that $\delta(u, v) \leq M$.

We decompose the graph $G$ into $(1 - x) \log n$ levels with each degree threshold defined as $t_i = \frac{n}{2^i}$.
We compute dominating sets $C_i$ of size $\tO{n/t_i}$ and edge sets $F_i$ of size at most $O(n t_{i - 1}) = O(n t_i)$.
Again, to obtain a $+k$ additive approximation, it suffices to have a $+5$ approximation after the $l - \frac{\multapproxlimit - 4}{2}$ iteration of computing $\dijkstra$.
Consider the blocking vertices $B(P) = \set{x_0, x_1, \dotsc, x_t}$ (\Cref{def:blocking-vertices}).
Let $a$ be the minimum index such that $\level(x_a) \leq l - \frac{\multapproxlimit - 4}{2}$ so that $\level(x_{a + 1}) \leq l - \frac{\multapproxlimit - 2}{2}$ and $\level(x_t) \leq l_0$ as we have $P \not\subset E_{k_0 + 1}$.
Then, if $\frac{n}{2^{i^*}} \leq \deg(x_t) \leq \frac{n}{2^{i^* - 1}}$, we have $P \subset F_{i^*}$ and in the $i^*$-th call to $\dominatingAPASP$, 
\begin{align*}
    \hat{\delta}(x_t^*, x_a^*) &\leq \delta(x_t, x_a) + 2 \\
    \hat{\delta}(x_t^*, x_{a+ 1}^*) &\leq \delta(x_t, x_{a + 1}) + 2
\end{align*}
By our assumption $\delta(u, v) \leq M$ both entries in the matrices constructed by $\dominatingAPASP$ are finite.
Furthermore, $x_t \in P_{x_a, x_{a + 1}}$ (\Cref{lemma:prop-blocking}) so the bounded \minplus product computes,
\begin{equation*}
    \hat{\delta}(x_a^*, x_{a + 1}^*) \leq \delta(x_a, x_{a + 1}) + 4
\end{equation*}

Then, in the graph $G_{\level(x_a), x_a^*}$ there exists the path $(x_a^*, x_{a + 1}^*, x_{a + 1}) \circ P_{x_{a + 1}, x_{a - 1}}$ and we obtain an estimate at most,
\begin{equation*}
    \hat{\delta}(x_a^*, x_{a - 1}^*) \leq \delta(x_{a}, x_{a - 1}) + 5
\end{equation*}

\paragraph{Algorithm}

\Cref{alg:mult-approx-bk} and \Cref{alg:mult-approx-bk-fmm} are similar overall, so we only highlight the differences here.
In Phase 1, we use a bounded \minplus product to replace the edge sets $D_{j_1} \times D_{j_2}$.

In \Cref{line:mult-approx-bk-fmm:fmm-decompose}, we decompose the graph $G$ into $(1 - x) \log n$ degree thresholds and at each threshold invoke $\dominatingAPASP$ to compute approximate distances between $D_{l - \frac{l_0}{2}}$ and $D_{l - \frac{l_0}{2} + 1}$ in \Cref{line:mult-approx-bk-fmm:dominating-apasp}.
In Phase 2, for each level $j \geq l_0 + 1$, we execute $\dijkstra$ from each vertex $w \in D_j$ on the graph $G_{j, w} = (V, E_j \cup E^* \cup (w \times V))$ in \Cref{line:mult-approx-bk-fmm:dijkstra}.

Finally, we take the minimum of our estimate with the estimate produced by $\bkAPASP$ and $\denseAPASP$.

\paragraph{Analysis}

We briefly argue that \Cref{alg:mult-approx-bk-fmm} equipped with fast matrix multiplication should outperform the combinatorial \Cref{alg:mult-approx-bk}.
Indeed, suppose we use naive matrix multiplication, so that the largest matrix multiplication contains matrices of size $D_{l - \frac{l_0}{2}} \times D_{l_0}$ and $D_{l_0} \times D_{l - \frac{l_0}{2} + 1}$ which requires time,
\begin{align*}
    \bigtO{\frac{n^3}{s_{l_0} s_{l - \frac{l_0}{2}} s_{l - \frac{l_0}{2} + 1}}} &= \bigtO{n^{3 - \left( 1 + \frac{l_0}{2(l - l_0))} + \frac{l_0 - 2}{2(l - l_0}\right)x}} \\
    &= \bigtO{n^{3 - \left( \frac{2l - 2}{2(l - l_0)} \right) x}}
\end{align*}

Balancing the running times,
\begin{equation*}
    3 - \frac{l - 1}{l - l_0} x = 2 + \frac{x}{l - l_0}
\end{equation*}

we obtain,
\begin{align*}
    x &= \frac{l - l_0}{l}
\end{align*}

so that the overall running time is $2 + \frac{1}{l}$, which matches the combinatoiral algorithm \Cref{alg:mult-approx-bk}.

\IncMargin{1em}
\begin{algorithm}[ht]
\SetKwInOut{Input}{Input}
\SetKwInOut{Output}{Output}
\Input{Unweighted, undirected graph $G = (V, E)$ with $n$ vertices; approximation parameter $k$}
\Output{Distance estimate $\hat{\delta}: U \times V \rightarrow \Z$ such that $\delta(u, v) \leq \hat{\delta}(u, v)$ for all $u, v \in V$ and $\hat{\delta}(u, v) \leq 2 \delta(u, v)$ whenever $\delta(u, v) \geq k$}

\BlankLine

\textcolor{blue}{Phase 0: Set up and Decompose Graph}

$\hat{\delta}(u, v) \gets \begin{cases}
    1 & (u, v) \in E \\
    \infty & \otherwise
\end{cases}$

Parameters $x, l$ chosen to minimize 
\begin{equation*}
    t_{\max} = \max \left( \omega\left( 1 - \frac{k - 2}{2(l - k + 2)} x, 1 - x, 1 - \frac{k - 4}{2(l - k + 2)} x \right), 2 + \frac{x}{l - k + 2}, 1.5 + x \right)
\end{equation*}

$M$ chosen so that $2 + \frac{2}{3 M - 2} \leq t_{\max}$

$l_0 \gets \multapproxlimit - 2$

$s_{l - i} \gets n^{\frac{i}{l - l_0} x}$ for all $1 \leq i \leq l - l_0$

$(D_{l_0}, D_{l_0 + 1}, \dotsc, D_{l}), (E_{l_0}, E_{l_0 + 1}, \dotsc, E_{l}), E^* \gets \decompose(G, (s_{l_0 + 1}, s_{l_0 + 2}, \dotsc, s_{l - 1}))$

\textcolor{blue}{Phase 1: Compute Intermediate Distance Estimates on High Degree Paths}

$t_{i} \gets \frac{n}{2^i}$ for all $1 \leq i \leq b - 1 = \ceil{(1 - x) \log n}$

$(C_{1}, C_{2}, \dotsc, C_{b}), (F_{1}, F_{2}, \dotsc, F_{b}), F^* \gets \decompose(G, (t_{1}, t_{2}, \dotsc, t_{b - 1}))$
\label{line:mult-approx-bk-fmm:fmm-decompose}

\For{$1 \leq i \leq b$}{
    $\hat{\delta} \gets \min\left(\hat{\delta}, \dominatingAPASP\left(G_i, D_{l - \frac{l_0}{2}}, D_{l - \frac{l_0}{2} + 1}, C_i, M + 2\right)\right)$ where $G_i = (V, F_i)$.
    \label{line:mult-approx-bk-fmm:dominating-apasp}
}

\textcolor{blue}{Phase 2: Compute Distance Estimates}

\For{$l_0 + 1 \leq j \leq l$}{
    \For{$w \in D_{j}$}{
        $G_{j, w} \gets \left(V, E_j \cup E^* \cup (w \times V) \right)$
        
        $\hat{\delta} \gets \dijkstra(G_{j, w}, w, \hat{\delta})$ 
        \label{line:mult-approx-bk-fmm:dijkstra}
    }
}

$\hat{\delta} \gets \min \left( \hat{\delta}, \bkAPASP(G_{l_0 + 1}) \right)$ where $G_{l_0 + 1} \gets (V, E_{l_0 + 1})$

$\hat{\delta} \gets \min \left( \hat{\delta}, \denseAPASP(G, M) \right)$

\caption{$\longMultiplicativeAPASPFMM(G, k)$} 
\label{alg:mult-approx-bk-fmm}
\end{algorithm}
\DecMargin{1em}

\paragraph{Warm-Up: $k = 4$}

We begin by showing that we can obtain a $(2, 0)$ approximation for $\delta(u, v) \geq 4$ in time $\tO{n^{2.01973523}}$.

\begin{proof}

Using \cite{Complexity}, we set parameters $x = 0.51973523$ and $l = 29$ so that $l_0 = 2$.
Choosing $M = 35$ suffices to ensure $\denseAPASP$ requires time $\bigtO{n^{2.0195}}$ (\Cref{lemma:dhz-apasp}).

Decompose the graph with degree thresholds $s_{2} = n^x,  s_3 = n^{26x/27}, \dotsc, s_{28} = n^{x/27}$.
Let $P$ be a shortest path of length at least 4.
If $P \subset E_3$, then we execute $\bkAPASP(G_3)$ in time $n^{1.5 + x} = n^{2.01973523}$ where $G_3 = (V, E_3)$ has maximum degree $s_2 = n^{x}$.
If $\delta(u, v) \geq M$, then $\denseAPASP$ computes a $+M$ approximation and we obtain a $(2, 0)$ approximation.

Then, assume $P \not\subset E_3$ is of length at most $M$ and vertex $w \in P$ of maximum degree has $\deg(w) \geq s_2$.
Suppose $\deg(v) \geq \deg(u)$.
Let $z = b(v, P)$ be the blocking vertex so $\deg(z) \geq s_{28}$ and $z^* = r(z, D_{\level(z)}) \in D_{28}$.
If $z = w$, we obtain a $+4$ approximation as in \Cref{alg:mult-approx-bk}.

Assume $z \neq w$.
Let $\frac{n}{2^i} \leq \deg(w) \leq \frac{n}{2^{i - 1}}$.
Since $P \subset F_i$, we obtain the estimates,
\begin{align*}
    \hat{\delta}(w^*, z^*) &\leq \delta(w, z) + 2 \leq M + 2 \\
    \hat{\delta}(w^*, v) &\leq \delta(w, v) + 1 \leq M + 1
\end{align*}
Then, the bounded \minplus product computes,
\begin{equation*}
    \hat{\delta}(v, z^*) \leq \delta(v, z) + 3
\end{equation*}
as $v \in V, z^* \in D_{28}, w^* \in C_i$.
Finally, the path $(v, z^*, z) \circ P_{z, u}$ exists in $G_{6, v}$ and $\dijkstra$ computes a $+4$ approximation.

The computation of every invocation of $\dijkstra$ requires time $\tO{n^{2 + x/27}}$, since $l - l_0 = 27$.
The computation of $\bkAPASP$ requires time $\tO{n^{1.5 + x}}$.
The overall complexity of $\dominatingAPASP$ is dominated by the computation of bounded \minplus products, which requires time,
\begin{equation*}
    \bigtO{n^{\omega\left(1 - \frac{x}{27}, 1 - x, 1\right)}}
\end{equation*}
where we have again plugged in $l = 29, l_0 = 2$.
Plugging in $x = 0.51973523$ yields the desired complexity.
\end{proof}

\paragraph{Correctness}

We prove \Cref{thm:mult-approx-bk} for general $k \geq 6$.

\begin{proof}
Let $u, v$ be vertices and $P$ a shortest path of length $\delta(u, v)$.

First, note $\delta(u, v) \leq \hat{\delta}(u, v)$ for all $u, v$ as $\dominatingAPASP$ returns $\hat{\delta}(u, v) \geq \delta(G_i, u, v) \geq \delta(u, v)$ by \Cref{lemma:dominating-apasp-feasible}.
In Phase 2, every edge weight in $G_{j, w}$ is the length of some path previously found in the original graph $G$.
Finally, $\bkAPASP$ and $\denseAPASP$ return estimates $\hat{\delta}(u, v) \geq \delta(u, v)$.

Now, we only need to show $\hat{\delta}(u, v) \leq 2 \delta(u, v)$
If $P \subset E_{l_0 + 1}$, this immediately follows from the correctness of $\bkAPASP$.
If $\delta(u, v) \geq M$, then we have correctness by $\denseAPASP$, since,
\begin{equation*}
    \hat{\delta}(u, v) \leq \delta(u, v) + M \leq 2 \delta(u, v)
\end{equation*}

Therefore, assume $P \not\subset E_{l_0 + 1}$ and $\delta(u, v) \leq M$.
Without loss of generality, suppose $\level(v) \leq \level(u)$.
Recall from Lemma \ref{lemma:block-approx-error} the blocking vertices $B(P) = \set{x_0, x_1, \dotsc, x_t}$ and levels $L_B(P)$ of path $P$.

Let $a = \min_{\level(x_i) \leq l - \frac{l_0}{2} + 1} i$ be the minimum index of an element in the blocking set such that $\level(x_i) \leq l - \frac{l_0}{2} + 1$.
Since $l \geq \level(x_1) > \level(x_2) > \dotsc > \level(x_t) \geq 1$, we can upper bound,
\begin{equation*}
    a \leq l - \left(l - \frac{l_0}{2} + 1\right) = \frac{l_0}{2} - 1 = \frac{\multapproxlimit}{2} - 2
\end{equation*}
Since $P$ is not contained in $E_{l_0 + 1}$, we can assume $a$ exists and $a \geq 1$.

Let $x_a \in B(P)$ be the corresponding vertex in $B(P)$.
Since $P$ has an edge not in $E_{l_0 + 1}$, the last blocking vertex of minimum level must have $\level(x_t) \leq l_0$.
Recall that we denote $v^* = r\left(v, D_{\level(v)}\right)$ for any vertex $v$.

Let $\frac{n}{2^{i^*}} \leq \deg(x_t) < \frac{n}{2^{i^* - 1}}$ and $x_t^* = r(x_t, C_{i^*})$.
Since $P \subset F_{i^*}$, the $i^*$-th invocation of $\dominatingAPASP$ computes 
\begin{align*}
    \hat{\delta}_{i^*}(x_t^*, x_a^*) &\leq \delta(x_t, x_a) + 2 \leq M + 2 \\
    \hat{\delta}_{i^*}(x_t^*, x_{a + 1}^*) &\leq \delta(x_t, x_{a + 1}) + 2 \leq M + 2
\end{align*}
These entries are finite in the matrices $A_{i^*}, B_{i^*}$ constructed by $\dominatingAPASP$.

Note that $x_t \in P_{x_a, x_{a + 1}}$ (\Cref{lemma:prop-blocking}).
Furthermore, $\level(x_a) \leq l - \frac{l_0}{2} + 1$ implies $x_{a}^* \in D_{l - \frac{l_0}{2} + 1}$ and $\level(x_{a + 1}) < \level(x_a)$ implies $x_{a + 1}^* \in D_{l - \frac{l_0}{2}}$.
Then, the bounded \minplus product computes,
\begin{align*}
    \hat{\delta}_{i^*}(x_a^*, x_{a + 1}^*) &\leq \hat{\delta}_{i^*}(x_a^*, x_t^*) + \hat{\delta}_{i^*}(x_t^*, x_{a + 1}^*) \\ 
    &\leq \delta(x_a, x_t) + \delta(x_t, x_{a + 1}) + 4 \\ 
    &\leq \delta(x_a, x_{a + 1}) + 4
\end{align*}

Then, in the $\level(x_a)$-th iteration, we can take the edge sequence $(x_a^*, x_{a + 1}^*, x_{a + 1})$ and the remaining edges in $E_{\level(x_a)}$ to compute a distance estimate at most,
\begin{equation*}
    \hat{\delta}(x_a^*, x_{a - 1}^*) \leq \hat{\delta}_{i^*}(x_a^*, x_{a + 1}^*) + 1 + \delta(x_{a + 1}, x_{a - 1}) \leq \delta(x_a, x_{a - 1}) + 5
\end{equation*}

The remaining proof follows exactly as in \Cref{thm:mult-approx-bk}.

\end{proof}

\paragraph{Time Complexity}

\begin{proof}
    We set $l_0 \gets \multapproxlimit - 2$ and leave $x, l, M$ to be optimized.
    Phase 0 requires $\tO{l n^2}$ time overall as we only invoke $\decompose$.
    Invoking $\dijkstra$ requires $\tO{n^{2 + \frac{x}{l - l_0}}}$ time as $|D_j| = \tO{n/s_j}$ and $|E_j| \leq O(n s_{j - 1})$.
    $\bkAPASP$ requires $\tO{m \sqrt{n}} = \tO{n^{1.5} s_{l_0}} = \tO{n^{1.5 + x}}$ time as we call $\bkAPASP$ on graph $G_{l_0 + 1}$, a graph with maximum degree $s_{l_0} = n^x$.
    
    Phase 1 calls $\dominatingAPASP$ $b = \tO{1}$ times.
    In each call, we call $\bfs$ on the graph $(V, F_i)$ from dominating set $C_i$ in time $\bigtO{|C_i||F_i|} = \bigtO{\frac{n}{t_i} {2 n t_i}} = \bigtO{n^2}$.
    Then, we compute a $\boundedMinPlus$ on matrices of size at most $\left|D_{l - \frac{l_0}{2}}\right| \times |C_i|$ and $|C_i| \times \left|D_{l - \frac{l_0}{2} + 1}\right|$.
    Since $|C_i| = O(|C_b|) = \bigtO{\frac{n}{t_b}} = \bigtO{2^b} = \bigtO{n^{1 - x}}$, all computations of $\boundedMinPlus$ require time at most,
    \begin{align*}
        \bigtO{n^{\omega\left( 1 - \frac{\frac{l_0}{2}}{l - l_0} x, 1 - x, 1 - \frac{\frac{l_0}{2} - 1}{l - l_0} x \right)}} &= \bigtO{n^{\omega\left( 1 - \frac{l_0}{2(l - l_0)} x, 1 - x, 1 - \frac{l_0 - 2}{2(l - l_0)} x \right)}} \\
        &= \bigtO{n^{\omega\left( 1 - \frac{k - 2}{2(l - k + 2)} x, 1 - x, 1 - \frac{k - 4}{2(l - k + 2)} x \right)}}
    \end{align*}

    To optimize, we use \cite{Complexity} and choose $x, l$ minimizing,
    \begin{align*}
        t_{\max} = \max \left( \omega\left( 1 - \frac{k - 2}{2(l - k + 2)} x, 1 - x, 1 - \frac{k - 4}{2(l - k + 2)} \right), 2 + \frac{x}{l - k + 2}, 1.5 + x \right)
    \end{align*}

    Finally, note that $\denseAPASP$ requires $\tO{n^{2 + \frac{2}{3M - 2}}}$ time.
    Given our choice of $x, l$, we simply choose $M$ to be a large enough constant such that the running time of $\denseAPASP$ does not exceed $2 + \frac{2}{3M - 2} \leq T_{\max}$.
    We give a few example running times in \Cref{prop:mult-approx-bk-fmm-examples} and \Cref{tbl:2-approx-geq-beta}.
\end{proof}

\section{Additive Approximation via Monotone \texorpdfstring{\minplus}{(min, +)} Product}
\label{sec:monotone-min-plus}

In the following section, we will use the fast monotone \minplus product of Chi, Duan, Xie, and Zhang \cite{chi2022monotone} to the additive approximation framework of Dor et al. \cite{dor2000apasp}.
In particular, this will generalize the result of Deng et al. \cite{deng2022apasp} to all distances.
Whereas Deng et al. \cite{deng2022apasp} use fast rectangular matrix multiplication for bounded difference matrices, we will require a more general fast rectangular matrix multiplication.
To do so, let us begin by designing a fast algorithm for multiplying rectangular matrices, of which only one is monotone.

\subsection{Rectangular Monotone \texorpdfstring{\minplus}{(min, +)} Product}

Following the algorithm of \cite{chi2022monotone}, we give a simple extension of their result for monotone matrices to rectangular monotone matrices.
This result slightly extends the work of Durr \cite{durr2023rect_monotone} to handle the most general case of rectangular matrices, which is required in our application.

\begin{theorem}
    \label{thm:rect-monotone-min-plus}
    Let $A$ be a $n^{a} \times n^{b}$ integer matrix with non-negative entries bounded by $O(n^{\mu})$.
    Let $B$ be a $n^{b} \times n^{c}$ row-monotone integer matrix with non-negative entries bounded by $O(n^{\mu})$.
    Then, there is an algorithm $\monotoneMinPlus$ computing $C = A * B$ in time,
    \begin{equation*}
        \bigtO{n^{\frac{(a + b + \mu) + \omega(a, b, c)}{2}}}
    \end{equation*}
\end{theorem}

Since the algorithm and proof remain largely unchanged, we only provide a sketch of the proof.

Let $\alpha \in [0, 1]$ be a parameter to be fixed, and let $p$ be a uniformly random prime chosen in $[40 n^{\alpha}, 80 n^{\alpha}]$.
We make the same simplifying assumption as in \cite{chi2022monotone}.
By Lemma 3.4 of \cite{chi2022monotone}, this assumption is justified.

\begin{assumption}
    (Assumption 3.1 of \cite{chi2022monotone})
    \label{ass:monotone-assumption}
    Let $i \in [n^a], k \in [n^b], j \in [n^c]$.
    For all $i, k$, $A_{ik}$ is either $\infty$ or $(A_{ik} \mod p) < \frac{p}{3}$.
    For all $k, j$, $B_{kj}$ is either $\infty$ or $(B_{kj} \mod p) < \frac{p}{3}$.
    Each row of $B$ is monotone.
\end{assumption}

\begin{lemma}
    \label{lemma:assumption-reduction}
    (Lemma 3.4 of \cite{chi2022monotone})
    Let $A$ be a $n^{a} \times n^{b}$ integer matrix with non-negative entries bounded by $O(n)$.
    Let $B$ be a $n^{b} \times n^{c}$ row-monotone integer matrix with non-negative entries bounded by $O(n)$.
    The computation $A * B$ can be reduced to a constant number of $A^i * B^i$ where $A^i, B^i$ satisfy Assumption \ref{ass:monotone-assumption}.
\end{lemma}

\begin{proof}
    We only define the relevant matrices, leaving the proof to \cite{chi2022monotone}.
    Define for $i \in \set{1, 2, 3}$,
    \begin{align*}
        A_{ik}^{((i)} &= \begin{cases}
            A_{ik} & \frac{(i - 1) p}{3} < (A_{ik} \mod p) < \frac{i p}{3} \\
            \infty & \otherwise
        \end{cases} \\
        B_{kj}^{(1)} &= \begin{cases}
            B_{kj} & (B_{kj} \mod p) < \frac{p}{3} \\
            p \floor{B_{kj}/p + 1} & \frac{p}{3} \leq (B_{kj} \mod p) < \frac{2p}{3} \\
            p \floor{B_{kj}/p + 1} & \frac{2p}{3} \leq (B_{kj} \mod p)
        \end{cases} \\
         B_{kj}^{(2)} &= \begin{cases}
            p \floor{B_{kj}/p} + \ceil{p/3} & (B_{kj} \mod p) < \frac{p}{3} \\
            B_{kj} & \frac{p}{3} \leq (B_{kj} \mod p) < \frac{2p}{3} \\
            p \floor{B_{kj}/p + 1} & \frac{2p}{3} \leq (B_{kj} \mod p)
        \end{cases} \\
         B_{kj}^{(3)} &= \begin{cases}
            p \floor{B_{kj}/p} + \ceil{2p/3} & (B_{kj} \mod p) < \frac{p}{3} \\
            p \floor{B_{kj}/p} + \ceil{2p/3} & \frac{p}{3} \leq (B_{kj} \mod p) < \frac{2p}{3} \\
            B_{kj} & \frac{2p}{3} \leq (B_{kj} \mod p)
        \end{cases} 
    \end{align*}
    Then, the matrices,
    \begin{equation*}
        A^{(1)}, A^{(2)} - \ceil{p/3}, A^{(3)} - \ceil{2p/3}, B^{(1)}, B^{(2)} - \ceil{p/3}, B^{(3)} - \ceil{2p/3}
    \end{equation*}
    all satisfy Assumption \ref{ass:monotone-assumption}.
\end{proof}

Define $h$ such that $2^{h - 1} \leq p < 2^h$.
For all $0 \leq l \leq h$, define $A_{ik}^{(l)} = \floor{\frac{A_{ik} \mod p}{2^l}}$ if $A_{ik}$ is finite and $\infty$ otherwise.
Define $B_{kj}^{(l)} = \floor{\frac{B_{kj} \mod p}{2^l}}$ if $B_{kj}$ is finite and $\infty$ otherwise.
Define $A_{ik}^* = \floor{A_{ik}/p}$ if $A_{ik}$ is finite and $\infty$ otherwise.
Define $B_{kj}^* = \floor{B_{kj}/p}$ if $B_{kj}$ is finite and $\infty$ otherwise.
Since $A^*, B^*$ are monotone, $C^* = A^* * B^*$ can be computed in $\tO{n^{(a + b + \mu) - \alpha}}$ time using segment trees as described in \cite{chi2022monotone}.
If $C_{ij}$ is finite, then by Assumption \ref{ass:monotone-assumption} $C_{ij} \mod p < \frac{2p}{3}$ so that $C_{ij}^* = \floor{C_{ij} / p}$.

In the next phase, we compute $C^{(l)}$ for $l = h, h - 1, \dotsc, 0$.
Note that $C^{(l)}$ is not necessarily $A^{(l)} * B^{(l)}$.
We require $C^{(l)}$ to satisfying the following properties if $C_{ij}$ is finite:
\begin{enumerate}
    \item $\floor{\frac{(C_{ij} \mod p) - 2(2^{l} - 1)}{2^{l}}} \leq C_{ij}^{(l)} \leq \floor{\frac{(C_{ij} \mod p) + 2(2^{l} - 1)}{2^{l}}}$ 
    \item If $C_{i, j_0}^* = C_{i, j_1}^*$ for $j_0 < j_1$, the elements $C_{i, j_0}^{(l)}, C_{i, j_0 + 1}^{(l)}, \dotsc, C_{i, j_1}^{(l)}$ are monotonically non-decreasing.
\end{enumerate}

Notice that when $l = 0$, we have $C_{ij}^{(0)} = C_{ij} \mod p$, which combined with $C_{ij}^*$ allows us to recover $C_{ij}$.

The entries of $A^{(l)}, B^{(l)}, C^{(l)}$ are therefore non-negative integers at most $\bigO{\frac{n^{\alpha}}{2^l}}$ or $\infty$.
Since $B^{(l)}$ is monotone, every row is composed of at most $\bigO{\frac{n^{\alpha}}{2^l}}$ intervals, where all the values in each interval are identical.
Now, $C^*$ has $O(n^{\mu - \alpha})$ intervals in each row.
Within each interval, $C^{(l)}$ has $\bigO{\frac{n^{\alpha}}{2^l}}$ intervals, so that $C^{(l)}$ has $\bigO{\frac{n^{\mu}}{2^l}}$ such intervals in each row.

\begin{definition}
    \label{def:monotone-min-plus-segment}
    (Definition 3.1 of \cite{chi2022monotone})
    A segment $(i, k, (j_0, j_1))$ with respect to $B^{(l)}, C^{(l)}$ where $i \in [n^a], k \in [n^b], j_0, j_1 \in [n^c]$ and $j_0 \leq j_1$ satisfies that for all $j_0 \leq j \leq j_1$,
    $B_{kj}^{(l)} = B_{k, j_0}^{(l)}$, $B_{kj}^{*} = B_{k, j_0}^{*}$, $C_{ij}^{(l)} = C_{i, j_0}^{(l)}$, and $C_{ij}^{*} = C_{i, j_0}^{*}$.
\end{definition}

We will also maintain sets $T_{b}^{(l)}$ for $-10 \leq b \leq 10$ where $T_{b}^{(l)}$ consists of segments $(i, k, (j_0, j_1))$ satisfying $A_{ik} < \infty$ and $A_{ik}^* + B_{kj}^* \neq C_{ij}^*$ and $A_{ik}^{(l)} + B_{kj}^{(l)} = C_{ij}^{(l)} + b$ for all $j \in [j_0, j_1]$.
Finally, we turn to the computation of each $C^{(l)}$ matrix.

\subsubsection*{Computing \texorpdfstring{$C^{(l)}$}{C(l)}}

\begin{enumerate}
    \item Since $p < 2^h$, $A^{(h)}, B^{(h)}$ are zero-matrices.
    Likewise, $C^{(h)} = 0$ satisfies the required properties.
    For $b \neq 0$, $T_{b}^{(h)} = \emptyset$ and $T_{0}^{(h)}$ includes all segments $(i, k, (j_0, j_1))$ where $A_{ik} < \infty$ and $A_{i,k}^* + B_{k,j_0}^* \neq C_{i,j_0}^*$.

    \item Let $l < h$. 
    We will construct $C^{(l)}, T_b^{(l)}$ from $C^{(l + 1)}, T_{b}^{(l + 1)}$.
    \begin{enumerate}
        \item {\bf Polynomial Matrix Multiplication}
        
        Construct $A_{ik}^p = x^{A_{ik}^{(l)} - 2A_{ik}^{(l + 1)}} y^{A_{ik}^{(l + 1)}}$ if $A_{ik}$ finite and $0$ otherwise.
        Define $B_{kj}^p$ analogously.
        Compute $C^p = A^p B^p$.
        Since the $x$ degree is either 0 or 1, and the $y$ degree is at most $n^{\alpha}$, this phase requires $\tO{n^{\omega(a, b, c) + \alpha}}$ time.
        
        \item {\bf Subtracting Error Terms}

        If $C_{ij}^p = 0$, then $C_{ij}^{(l)} = \infty$.
        Otherwise, for all $b$, we collect all monomials $\lambda x^c y^d$ where $d = C_{ij}^{(l + 1)} + b$ and let $C_{ijb}^p(x)$ be the sum of all such $\lambda x^c$.
        Next, we compute,
        \begin{equation*}
            R_{ijb}^p(x) = \sum_{(i, k, (j_0, j_1)) \in T_{b}^{(l + 1)}, j \in [j_0, j_1]} x^{A_{ik}^{(l)} - 2 A_{ik}^{(l + 1)} + B_{kj}^{(l)} - 2 B_{kj}^{(l + 1)}}
        \end{equation*}

        Let $s_{ijb}$ be the minimum degree of $x$ in $C_{ijb}^p(x) - R_{ijb}^p(x)$ and compute $c_{ijb} = 2d + s_{ijb}$ (or $c_{ijb} = \infty$ if $s_{ijb} = 0$).
        Finally, output $C_{ij}^{(l)} = \min_b c_{ijb}$.

        Constructing $C_{ijb}^p(x)$ requires $\tO{n^{a+c+\alpha}}$ time.
        Notice each $T_{b}^{(l + 1)}$ has at most $2$ different $B_{kj}^{(l)}$ and thus two $R_{ijb}^p(x)$.
        Since each interval in $B_{kj}^{(l + 1)}$ is broken into at most $2$ intervals, we can compute $C_{ijb}^p(x) - R_{ijb}^p(x)$ using segment trees in $\tO{|T_{b}^{(l + 1)}|}$-time.
        The overall time required in this stage is therefore $\tO{n^{a + c + \alpha} + |T_{b}^{(l + 1)}|}$.
        
        \item {\bf Computing Triples $T_b^{(l)}$}

        It is shown in \cite{chi2022monotone} that each segment with respect to $B^{(l + 1)}, C^{(l + 1)}$ splits into $O(1)$ segments with respect to $B^{(l)}, C^{(l)}$ and furthermore $\bigcup_{b} T_{b}^{(l)} \subset \bigcup_{b} T_{b}^{(l + 1)}$, so it suffices to enumerate $T_{b}^{(l + 1)}$, breaking each segment into $O(1)$ sub-segments.
        The time required is $\tO{|T_{b}^{(l + 1)}|}$.
    \end{enumerate}
\end{enumerate}

It then suffices to bound the size of the sets $T_{b}^{(l)}$, which we do below in Lemma \ref{lemma:error-set-bound-monotone}.
Finally, we can optimize parameter $\alpha$ and compute the overall running time.
Note that the overall running time is,
\begin{equation*}
    \bigtO{n^{(a + b + \mu) - \alpha} + n^{\omega(a, b, c) + \alpha} + n^{a + c + \alpha}} = \bigtO{n^{(a + b + \mu) - \alpha} + n^{\omega(a, b, c) + \alpha}} 
\end{equation*}

giving the desired result.
The proof of correctness follows exactly as in \cite{chi2022monotone}.

\begin{lemma}
    \label{lemma:error-set-bound-monotone}
    The expected number of segments in $T_{b}^{(l)}$ is $\tO{n^{(a + b + \mu) - \alpha}}$
\end{lemma}

\begin{proof}
    Suppose first $2^l > \frac{p}{100}$.
    Then, $B^{(l)}$, $C^{(l)}$ have at most $\bigO{\frac{n^{\mu}}{2^l}} = \bigO{\frac{n^{\mu}}{p}} = \bigO{n^{\mu - \alpha}}$ intervals in a given row.
    The overall number of intervals is therefore bounded by $\bigO{n^{(a + b + \mu) - \alpha}}$

    Consider a segment $(i, k, (j_0, j_1))$ and arbitrarily pick $j \in [j_0, j_1]$ where $A_{ik}$ finite and $A_{ik}^* + B_{kj}^* \neq C_{ij}^*$.
    Then, by Assumption \ref{ass:monotone-assumption}, we have $|A_{ik} + B_{kj} - C_{ij}| \geq \frac{p}{3}$.
    We now bound the probability that $(i, k, (j_0, j_1)) \in T_{b}^{(l)}$.
    By definition, this implies that,
    \begin{equation*}
        \floor{\frac{A_{ik} \mod p}{2^l}} + \floor{\frac{B_{kj} \mod p}{2^l}} = C_{ij}^{(l)} + b 
    \end{equation*}
    Since $\floor{\frac{(C_{ij} \mod p) - 2(2^{l} - 1)}{2^{l}}} \leq C_{ij}^{(l)} \leq \floor{\frac{(C_{ij} \mod p) + 2(2^{l} - 1)}{2^{l}}}$ 
    \begin{equation*}
        -4 \leq \frac{A_{ik} \mod p}{2^l} + \frac{B_{kj} \mod p}{2^l} - \frac{C_{ij} \mod p}{2^l} - b \leq 4
    \end{equation*}
    Let $C_{ij} = A_{iq} + B_{qj}$, so that,
    \begin{align*}
        -4 \cdot 2^l \leq (A_{ik} + B_{kj} - A_{iq} - B_{qj}) \mod p - b \cdot 2^{l} \leq 4 \cdot 2^l
    \end{align*}
    so that $(A_{ik} + B_{kj} - A_{iq} - B_{qj}) \mod p$ takes one of $O(2^l)$ values $r$ in the range $[2^l (b - 4), 2^l (b + 4)]$.
    Since $|b| \leq 10$, the largest such value is bounded by $|r| \leq 14 \cdot 2^l < \frac{p}{6} \leq \frac{1}{2} |A_{ik} + B_{kj} - A_{iq} - B_{qj}|$ where we used $2^l < \frac{p}{100}$.

    Then, if $B_{kj}, B_{qj}$ are from $B$ in Lemma \ref{lemma:assumption-reduction}, then,
    \begin{equation*}
        |(A_{ik} + B_{kj} - A_{iq} - B_{qj}) - r| = O(n^{\mu})
    \end{equation*}
    The number of primes larger than $40n^{\alpha}$ that divide this quantity is therefore at most $\frac{\mu}{\alpha} = O(1)$.
    In particular, such a $p$ is chosen with probability $n^{-\alpha}$.

    On the other hand, if $B_{kj}, B_{qj}$ may have been set artificially to some number congruent to $0, \ceil{p/3}, \ceil{2p/3}$ in Lemma \ref{lemma:assumption-reduction}.
    We can still bound the probability that $p$ divides $|(A_{ik} + B_{kj} - A_{iq} - B_{qj}) - r|$, as in \cite{chi2022monotone}.

    Finally, we have $O(2^l)$ such remainders $r$, and in total $\bigO{\frac{n^{a + b + \mu}}{2^l}}$ segments, so that in expectation we can upper bound $|T_{b}^{(l)}| = \tO{n^{(a + b + \mu) - \alpha}}$.
\end{proof}

\subsection{Framework for Additive Approximations via Monotone \texorpdfstring{\minplus}{(min, +)} Product}

We now give an improvement on the generalized algorithm of Dor et al. \cite{dor2000apasp}, using fast matrix multiplication to give a better trade-off between additive approximation error and running time.
This extends the work of Deng et al. \cite{deng2022apasp} to additive error beyond +2.

\genadditiveapasp

Before we present a high level overview of Theorem \ref{thm:gen-additive-apasp}, let us recall some useful definitions.

\begin{definition}
    \label{def:euler-tour}
    For a spanning tree $T \subset E$ of a connected graph $G = (V, E)$ on $n$ vertices, an {\bf Euler Tour} of $T$ is a sequence of vertices $v_1, v_2, \dotsc, v_{2n - 1}$ where each vertex of $G$ appears at least once and the edges $(v_i, v_{i + 1}) \in T$ for all $1 \leq i \leq 2n - 2$.
\end{definition}

Given a spanning tree $T$, an Euler tour of $T$ can be found in $O(|T|)$ time by conducting depth first search on $T$.

\begin{definition}
    \label{def:euler-tour-matrix}
    Let $G = (V, E)$ be a graph on $n$ vertices and $D \subset V$.
    Consider an arbitrary Euler Tour on $G$, denoted $v_1, \dotsc, v_{2n - 1}$.
    $X$ is an {\bf Euler Tour Distance Matrix} of $D$ on $V$ if $A$ is a $|D| \times (2n - 1)$ matrix where $X(w, i) = \delta(w, v_i)$ for $w \in D$ and $v_i$ as ordered in the Euler Tour.
\end{definition}

It was first observed by Deng, Kirkpatrick, Rong, Vassilevska-Williams, and Zhong \cite{deng2022apasp} that it is possible to encode the results of a $\bfs$ search in a row bounded difference matrix by ordering the vertices according to an Euler tour.

We now present a high level description of our algorithm.

\paragraph{High Level Overview}

First, we recap the algorithm of Deng et al. \cite{deng2022apasp} and explain why it does not generalize to additive approximations beyond $+2$.
Consider a $+4$ approximation.
The natural approach is to use a bounded difference matrix multiplication to compute a $+2$ additive approximation on paths with some vertex of degree at least $n^x$ and to invoke $\sparseAPASP(G, 4)$ on paths where all vertices have degree at most $n^x$.
If $\omega_{BD}(1, 1 - x, 1)$ denotes the exponent required to compute a rectangular bounded difference \minplus matrix product with input $n \times n^{1 - x}$ and $n^{1 - x} \times n$, then this algorithm requires time $n^{2 + \frac{x}{3}}$ where $x$ is the solution to $\omega_{BD}(1, 1-x, 1) = 2 + \frac{x}{3}$.
For current bounds on $\omega_{BD}$, this is larger than the $\frac{11}{5}$ given by the combinatorial algorithm of Dor, Halperin, and Zwick \cite{dor2000apasp}.
Why? 
The $\denseAPASP$ algorithm decomposes the graph $G$ into 5 levels, whereas the above algorithm has only decomposed $G$ into 4 levels.

Indeed, only using rectangular matrix multiplication $\omega_{BD}(1, y, 1)$ for some $y \in [0, 1]$, we will not be able to obtain any improvement on $\denseAPASP(G, 4)$.
For example, recall that the degree thresholds are $s_1 = n^{4/5}, s_2 = n^{3/5}, s_3 = n^{2/5}, s_4 = n^{1/5}$.
In order to improve the overall running time, we must decrease $s_4$ so that the final iteration of computing $\dijkstra$ on $E_5$ from all $D_5 = V$ does not require $\tO{n^{11/5}}$ time.
If we wish to decrease $s_4$, then the size of $D_4$ increases.
To limit the overall time of running $\dijkstra$ from $D_4$ on $E_4$, we must therefore decrease $s_3$.
Following the same reasoning, we must decrease $s_2$ and $s_1$.
However, the combinatorial algorithm $\denseAPASP$ also includes the edge set $D_2 \times D_4$ when executing $\dijkstra$ from $D_5 = V$.
In particular, either $D_2$ or $D_4$ needs to reduce in size, which implies that at least one of $s_2, s_4$ needs to increase, contradicting our previous conclusions.

We cannot hope to succeed simply by using fast \minplus matrix multiplication to replace the first level $D_1$.
However, we use fast \minplus matrix multiplication to replace the edge set $D_2 \times D_4$ and more generally $D_{j_1} \times D_{j_2}$ for $j + j_1 + j_2 \leq 2k + 1$.
Consider $+4$ approximation as a basic example.
Let $P$ be a shortest path between $u, v$.
If $P \subset E_3$, then we compute a $+4$ approximation as in $\sparseAPASP$.
Therefore, suppose $P \subset E_j$ and $P \not\subset E_{j + 1}$ for some $j \in \set{1, 2}$.
Let $y$ be the vertex closest to $v$ such that $P_{u, y} \subset E_5$.
Let $z$ be the vertex of maximum degree in $P$.
Let $y^* = r(y, D_4)$ and $z^* = r(z, D_{\level(z)})$.
In the combinatorial algorithm, since $D_2 \times D_4 \subset G_{5, v}$, we find the path $(v, z^*, y^*, y) \circ P_{y, u}$ and compute a $+4$ approximation.
Instead, we use a fast \minplus product to find,
\begin{equation*}
    \hat{\delta}(v, y^*) \leq \min_{w \in D_2} \hat{\delta}(v, w) + \hat{\delta}(w, y^*) \leq \hat{\delta}(v, z^*) + \hat{\delta}(z^*, y^*) \leq \hat{\delta}(v, y) + 3
    \label{eq:+4-min-plus-approx}
\end{equation*}
to obtain the same approximation.
To utilize fast matrix multiplication, let $x$ be some parameter to be optimized.
Since $P \subset E_3$ is handled by $\sparseAPASP$, we decompose the graph $G_x = (V, E_{n^x})$ into 3 levels.
Then, for $1 \leq j \leq (1 - x) \log n$ thresholds $t_j = \frac{n}{2^j}$, we compute a dominating set $C_j$ and edge set $F_j$.
For each $j$, we construct matrix $A_j$ of dimension $D_4 \times C_j$ and matrix $B_j$ of dimension $C_j \times V$ where we ensure $B_j$ is a monotone matrix by using the Euler tour distance matrix construction and the transformation $\bdToMonotone$.
Filling in these matrices with the results computed by $\bfs$ from $C_j$ on the graph $(V, F_j)$, we compute the correct distance estimate required in \Cref{eq:+4-min-plus-approx}.
Since overall the $\bfs$ require time $\tO{n^2}$, the running time of \minplus matrix multiplication dominates.

\begin{figure}[ht]
\centering
\includegraphics[width=0.9\textwidth]{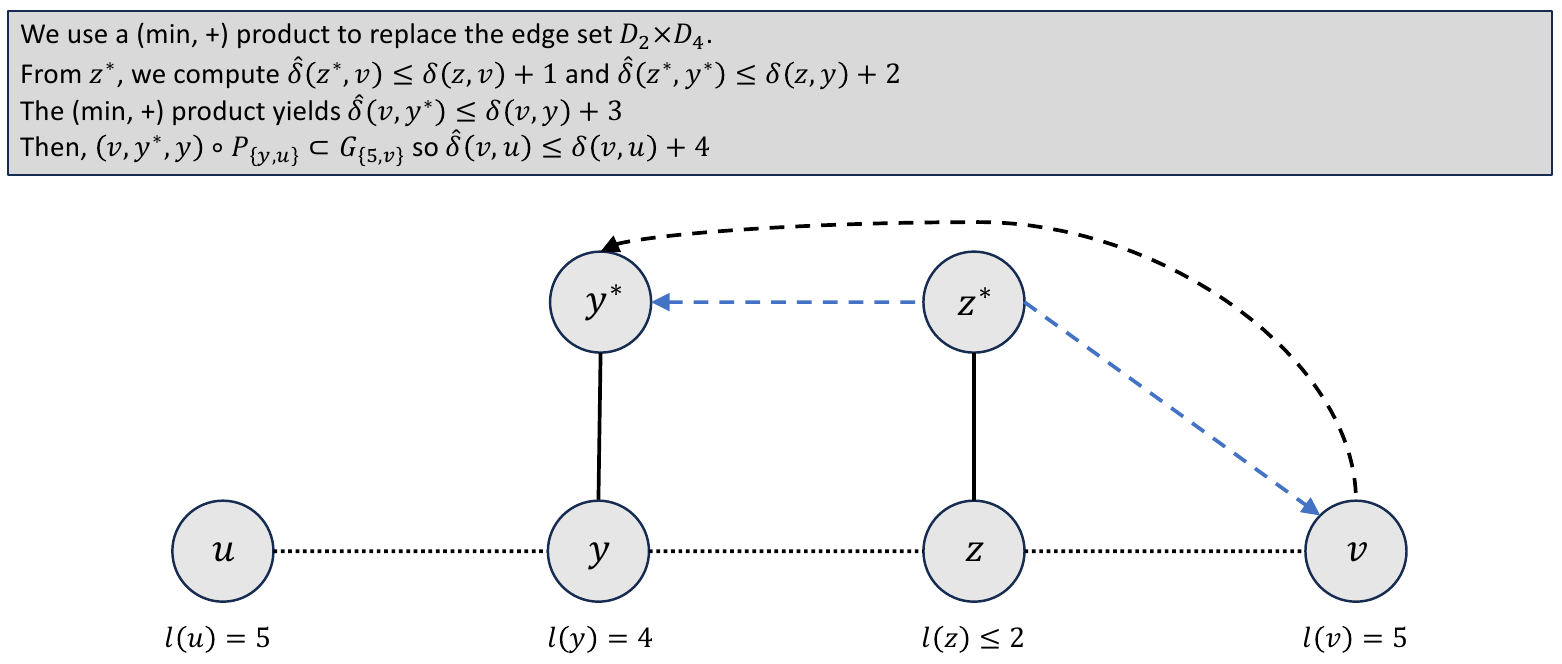}
\caption{Using \minplus product to replace edge set $D_2 \times D_4$.
Solid lines represent eges in $G$.
Dotted black lines represent paths in $G$.
Dashed blue arrows represent distance estimates used in the \minplus product.
Dashed black arrows represent edges $w \times V$ in the graph $G_{j, w}$.
}
\label{fig:min-plus-intro}
\end{figure}

We have already discussed $+6$ approximation in \Cref{sec:2-det-approx}.

For general $+\beta$ approximation, let $k = \frac{3 \beta - 2}{2}$ and $k_0 = k - \floor{k/3} - 2$.
By \Cref{lemma:sparse-apasp-approx}, if any path $P \subset E_{k_0 + 1}$, then we already obtain a $+\beta$-approximation as in $\sparseAPASP$.
Then, to ensure a $+\beta$ approximation for paths containing a higher degree vertex, we need to ensure that we have a $+5$ approximation after the $(k_0 + 3)$-th iteration, we can follow the arguments of \Cref{lemma:sparse-apasp-approx} to obtain a $+\beta$ additive approximation from this point.

We demonstrate this case in \Cref{fig:bounded-min-plus-approx}.
In this example, $a$ is a vertex on the shortest path $P$ between $u, v$ such that $\level(a) = k_0 + 3$.
While the diagram only shows $\hat{\delta}(a^*, v) \leq \delta(a, v) + 5$, a similar argument holds for all $x \in P$.

\begin{figure}[ht]
    \centering
    \includegraphics[width=0.7\textwidth]{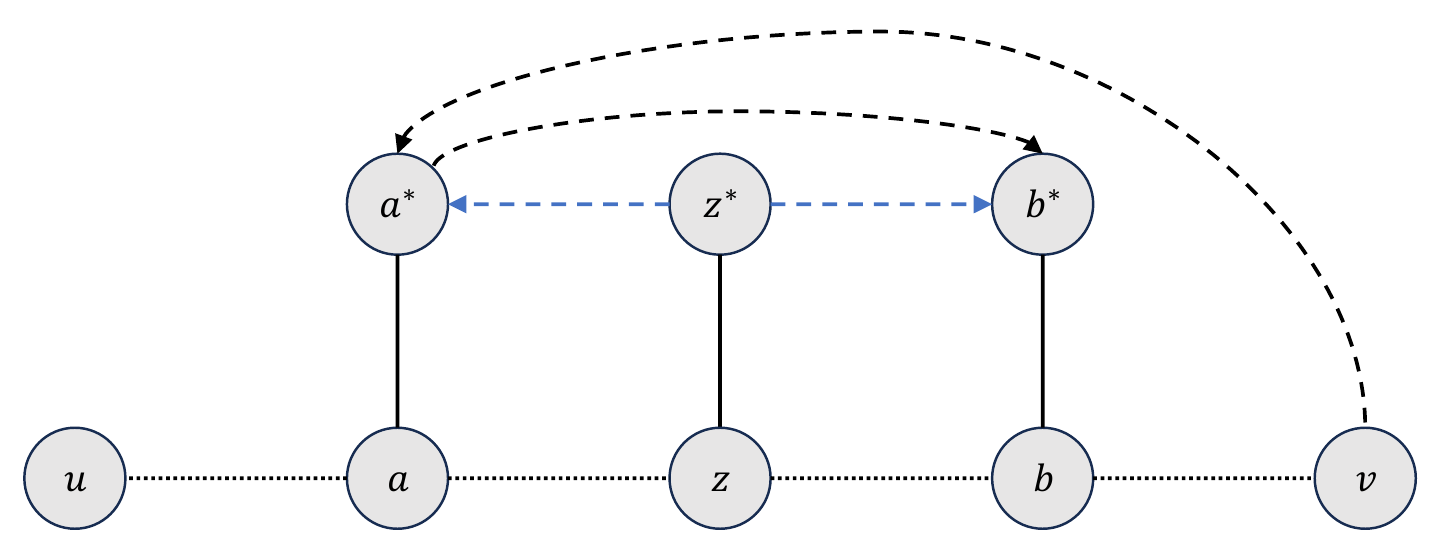}
    \caption{Computing a distance estimate between $u, v$.
    Solid black lines denote edges in $G$.
    Dotted black lines represent paths in $G$.
    Blue dashed arrows denote distance estimates used in the \minplus product (Phase 1). 
    Black dashed arrows denote edges $w \times V$ in the graph $G_{j, w}$ (Phase 2).}
    \label{fig:bounded-min-plus-approx}
\end{figure}

Let $b$ be the vertex closest to $a$ such that $P_{b, v} \subset E_{k_0 + 3}$.
Let $z$ be the highest degree vertex in $P_{a, b}$.
Suppose $\frac{n}{2^j} \leq \deg(z) < \frac{n}{2^{j - 1}}$.
Then, after computing the \minplus product $A_j * B_j$, we have an estimate
\begin{equation*}
    \hat{\delta}(a^*, b^*) \leq \delta(a, b) + 4
\end{equation*}
given by the blue arrows.
In particular from $a^*$, the path $(a^*, b^*, b) \circ P_{b, v} \subset G_{k_0 + 3, a^*}$ so that
\begin{equation*}
    \hat{\delta}(a^*, v) \leq \delta(a, v) + 5
\end{equation*}

We next argue that our algorithm has an improved running time.
Suppose we use naive matrix multiplication. 
Then, the time exponent to compute a $\monotoneMinPlus$ combinatorially is,

\begin{align*}
    \left( 1 - \frac{\beta-2}{\beta+2} x \right) +  \left( 1 - x \right) + \left( 1 - \frac{\beta-4}{\beta+2} x \right) &= 3 - \frac{3 \beta - 4}{\beta + 2} \cdot x
\end{align*}

Then, balancing this with the $2 + \frac{2x}{\beta + 2}$ exponent required to execute $\sparseAPASP$,
\begin{equation*}
    3 - \frac{3\beta - 4}{\beta + 2} \cdot x = 2 + \frac{2 x}{\beta + 2}
\end{equation*}

so that,
\begin{equation*}
    x = \frac{\beta + 2}{3 \beta - 2}
\end{equation*}

and obtain running time exponent $2 + \frac{2}{\beta + 2} \frac{\beta + 2}{3 \beta - 2} = 2 + \frac{2}{3 \beta - 2}$, thus recovering the approximation algorithm of Dor, Halperin, and Zwick \cite{dor2000apasp}.
In particular, by using fast matrix multiplication instead, we obtain a faster running time for all values of $\beta$.

\paragraph{Algorithm}

We are now ready to describe our algorithm.
We begin our algorithm by initializing the distance estimate matrix as the adjacency matrix of the graph $G$, as well as setting parameters $x, k$.
Then, we decompose our graph into $\floor{k/3} + 3$ levels according to the degree thresholds $s_{k_0}, s_{k_0 + 1}, \dotsc, s_{k - 1}$.

In Phase 1, we further decompose the top level into roughly $(1 - x) \log n$ levels.
At each level, we execute $\bfs$ from every vertex $w$ in the dominating set $C_i$ on the graph $(V, F_i)$.
We construct two matrices $A_i, B_i$ encoding the computed distances between $D_{k_0 + 2}, C_i$ and $C_i, D_{k_0 + 3}$ respectively.

In order to apply the $\monotoneMinPlus$ algorithm of \cite{chi2022monotone}, we define $B_i'$ to be the Euler tour distance matrix of $C_i$ on $G_i$ (see Definition \ref{def:euler-tour-matrix}) observing that this is a row-bounded difference matrix.
Then, we modify this matrix into a row-monotone matrix by adding $j$ to each column $j$ and taking a sub-matrix of the columns corresponding to $D_{k_0 + 3}$ to obtain a monotone matrix, as discussed in Section \ref{sec:prelims:fast-min-plus}.
In the pseudocode, this is contained in the invocation of $\bdToMonotone$.
Then, we compute the \minplus product of $A_i, B_i$ using $\monotoneMinPlus$.

In Phase 2, we execute $\dijkstra$ from each $D_i$ for $i \geq k_0 + 1$.
In each iteration, we search the graph $G_{j, w}$ consisting of $E_j \cup E^* \cup (w \times V)$.

We now present the algorithm and proof of Theorem \ref{thm:gen-additive-apasp}.

\IncMargin{1em}
\begin{algorithm}[ht]
\SetKwInOut{Input}{Input}\SetKwInOut{Output}{Output}
\Input{Unweighted, undirected Graph $G = (V, E)$ with $n$ vertices; approximation parameter $\beta$}
\Output{Distance estimate $\hat{\delta}: U \times V \rightarrow \Z$ such that $\delta(u, v) \leq \hat{\delta}(u, v) \leq \delta(u, v) + \beta$ for all $u, v \in V$}

\BlankLine

\textcolor{blue}{Phase 0: Set up and Decompose Graph}

$\hat{\delta}(u, v) \gets \begin{cases}
    1 & (u, v) \in E \\
    \infty & \otherwise
\end{cases}$

$k \gets \frac{3 \beta - 2}{2}$

$x \gets $ solution to $\omega\left(1 - \frac{\beta - 2}{\beta + 2} x, 1 - x, 1 - \frac{\beta - 4}{\beta + 2} x\right) = 1 + \frac{4  + 2 \beta}{\beta + 2} x$

$s_{k - i} \gets n^{i \cdot \frac{x}{\floor{k/3} + 2}}$ for all $1 \leq i \leq \floor{k/3} + 2$

$k_0 \gets k - (\floor{k/3} + 2) = k - \floor{k/3} - 2$

$(D_{k_0}, D_{k_0 + 1}, \dotsc, D_k), (E_{k_0}, E_{k_0 + 1}, \dotsc, E_k), E^* \gets \decompose(G, (s_{k_0}, s_{k_0 + 1}, \dotsc, s_{k - 1}))$

\textcolor{blue}{Phase 1: Estimate Distances on High-Degree Paths}

$t_{i} \gets \frac{n}{2^i}$ for all $1 \leq i \leq l - 1 = \ceil{(1 - x) \log n}$

$(C_{1}, C_{2}, \dotsc, C_{l}), (F_{1}, F_{2}, \dotsc, F_{l}), F^* \gets \decompose(G, (t_{1}, t_{2}, \dotsc, t_{l - 1}))$

\For{$0 \leq i \leq l$}{
    \For{$w \in C_i$}{
        $X_i(w) \gets \bfs(w)$ on $G_i = \left(V, F_{i}\right)$
        
        $\hat{\delta}(w, v) \gets \min(\hat{\delta}, X_i(w, v))$ for all $v \in V$
    }
    Construct $|D_{k_0 + 2}| \times |C_i|$ matrix $A_i$ where $A_i(v, w) = X_i(w, v)$
    
    Construct $|C_i| \times (2n - 1)$ Euler tour distance matrix $B_i'$ of $C_i$ on $G_i$
    
    $B_i \gets \bdToMonotone(B_i')[C_i, D_{k_0 + 3}]$
    
    $\hat{\delta} \gets \min(\hat{\delta}(w, v), \monotoneMinPlus(A_i, B_i))$
    \label{line:additive-apasp:monotone-min-plus}
}

\textcolor{blue}{Phase 2: Estimate Distances on Low-Degree Paths}

\For{$k_0 + 1 \leq j \leq k$}{
    \For{$w \in D_{j}$}{
        $G_{j, w} \gets (V, E_j \cup (u \times V) \cup E^*)$
        
        $\hat{\delta} \gets \dijkstra(G_{j, w}, w, \hat{\delta})$
    }
}

\caption{$\additiveAPASP(G, \beta)$} 
\label{alg:k-additive-apasp}
\end{algorithm}
\DecMargin{1em}

\subsubsection*{Warm-Up: \texorpdfstring{$\beta = 4$}{beta = 4}}

We begin by showing a proof of correctness for the case of $\beta = 4$.
The somewhat simplified argument in this case will be generalized to arbitrary $\beta \geq 6$ below.

\begin{lemma}
    \label{lemma:monotone-4-approximation}
    Algorithm $\additiveAPASP(G, 4)$ returns an estimate $\hat{\delta}$ such that $\delta(u, v) \leq \hat{\delta}(u, v) \leq \delta(u, v) + 4$ for all $u, v$.
\end{lemma}

\begin{proof}
    Let $P$ be a shortest path between $u, v$.
    Note that $k = 5, k_0 = 2$.
    
    Suppose $P \subset E_3 = E_{k_0 + 1}$. 
    Then, $L(P, 5) \subset \set{3, 4}$ and $|L(P, 5)| \leq 2$.
    By Lemma \ref{lemma:sparse-apasp-approx}, since $D_5 = D_k = V$, we obtain a $+4$ approximation in Phase 2.

    Thus, $P$ has some vertex with degree at least $s_2 = n^{x}$.
    Let $y$ denote the vertex in $P$ with degree at least $s_4$ closest to $v$.
    Let $z$ be a vertex of maximum degree on $P$ and $i^*$ be the integer $1 \leq i \leq l$ such that $\frac{n}{2^{i^*}} \leq \deg(z) < \frac{n}{2^{i^* - 1}}$.
    Since $P \subset F_{i^*}$, we have in the $i^*$-th iteration of Phase 1 that $X_{i^*}(z^*, u) \leq \delta(z, u) + 1$ and $X_{i^*}(z^*, y^*) \leq \delta(z, y) + 2$ where $z^* = r(z, C_{i^*})$, $y^* = r(y, D_{4})$.
    Therefore, after Line \ref{line:additive-apasp:monotone-min-plus},
    \begin{equation*}
        \hat{\delta}_{i^*}(u, y^*) \leq X_{i^*}(u, z^*) + X_{i^*}(z^*, y^*) \leq \delta(u, z) + \delta(z, y) + 3 = \delta(u, y) + 3
    \end{equation*}
    
    Now, in the $5$-th iteration in Phase 2, when we execute Dijkstra from $u$, we have an edge $(u, y^*)$ with weight at most $\delta(u, y) + 3$, an edge $(y^*, y) \in E$ and the remaining edges in $E_5 = E_k$, so that,
    \begin{equation*}
        \hat{\delta}(u, v) \leq \hat{\delta}_{i^*}(u, y^*) + 1 + \delta(y, v) \leq \delta(u, y) + \delta(y, v) + 4 \leq \delta(u, v) + 4
    \end{equation*}
\end{proof}

\subsubsection*{General Case: \texorpdfstring{$\beta \geq 6$}{b >= 6}}

We now present the proof of Theorem \ref{thm:gen-additive-apasp}.

\begin{proof}
    (Correctness)
    Let $u, v \in V$ be a pair of vertices and $P$ be a shortest path between $u, v$.
    Let $\deg(P) = \max_{v \in P} \deg(v)$ be the maximum degree of any vertex on $P$.
    
    Note by the definition of $k \gets \frac{3 \beta - 2}{2}$ and $\beta$ is an even integer,
    \begin{equation*}
        2 \left( \floor{k/3} + 1 \right) = 2 \left( \floor{\frac{\beta}{2} - \frac{1}{3}} + 1 \right) \leq \beta
    \end{equation*}
    By Lemma \ref{lemma:additive-feasible-estimates}, we only need to show $\hat{\delta}(u, v) \leq \delta(u, v) + 2(\floor{k/3} + 1)$.
    
    Recall that $V_{s} = \set{v \in V \given \deg(v) \geq s}$ from Definition \ref{def:degree-threshold}. 
    
    \paragraph{Case 1: $\deg(P) < s_{k_0}$}

    In this case, we have $P \subset E_{k_0 + 1}$, and we may follow the proof of $\sparseAPASP$ by Dor, Halperin, and Zwick \cite{dor2000apasp}.
    By assumption $L(P, k) \subset \set{k_0 + 1, k_0 + 2, \dotsc, k - 1}$ so that $|L(P, k)| \leq \floor{k/3} + 1$.
    From Lemma \ref{lemma:sparse-apasp-approx}, it immediately follows that $\hat{\delta}(u, v) \leq 2 (\floor{k/3} + 1)$.

    \paragraph{Case 2: $\deg(P) \geq s_{k_0}$}

    By assumption, there is a vertex $w \in P$ with $\deg(w) \geq s_{k_0 + 3}$.
    Then, with Lemma \ref{lemma:k0+3-approx-monotone}, we have that for all $x \in P$,
    \begin{equation*}
        \hat{\delta}_{k_0 + 3}(w^*, x) \leq \delta(w, x) + 5
    \end{equation*}
    
    Then, following the same argument as the inductive step of Lemma \ref{lemma:block-approx-error}, consider any vertex $w \in P$ such that $\level(w) \geq j$ for $j > k_0 + 3$.
    We claim the following for all $j$,
    \begin{equation*}
        \hat{\delta}_{j}(w^*, x) \leq \delta(w, x) + 2 (j - (k_0 + 1)) + 1
    \end{equation*}
    where above we have shown the base case for $j = k_0 + 3$.
    We now proceed by induction for $j > k_0 + 3$.
    Consider an execution of $\dijkstra$ from $w^*$ in $G_{j, w}$.
    For any vertex $x \in P$, let $y$ be the last vertex on the sub-path $P_{w, x}$ with $\level(y) \leq j - 1$.
    
    If no such vertex exists, then $P_{w, x} \subset E_j$ and we compute an exact distance from $w^*$ so that $\hat{\delta}(w^*, x) \leq \delta(w, x) + 1$.
    
    In particular, $\deg(y) \geq s_{j - 1}$ so let $y^* = r(y, D_{j - 1})$.
    We take the edges $(w^*, y^*)$, $(y^*, y) \in E^*$, and the remaining edges in $E_j$.
    By induction, we have,
    \begin{align*}
        \hat{\delta}_{j}(w^*, x) &\leq \hat{\delta}_{j - 1}(w^*, y^*) + 1 + \delta(y, x) \\
        &\leq \delta(w, y) +  2 ((j - 1) - (k_0 + 1)) + 3 + \delta(y, x) \\
        &\leq \delta(w, x) + 2 (j - (k_0 + 1)) + 1
    \end{align*}
    
    Then, on the final iteration, let $w$ be the closest vertex to $v$ on $P$ such that $\deg(w) \geq s_{k - 1}$.
    From $u$, we take the edge $(u, w^*)$, followed by $(w, w^*) \in E^*$ and the remaining path in $E_k$, so that,
    \begin{align*}
        \hat{\delta}(u, v) &\leq \hat{\delta}_{k - 1}(u, w^*) + 1 + \delta(w, v) \\
        &\leq \delta(u, w) + 2 (k - 1 - (k_0 + 1)) + 2 + \delta(w, v) \\
        &\leq \delta(u, v) + 2 (\floor{k/3} + 1)
    \end{align*}
    completing the proof.
\end{proof}

\begin{proof}
    (Time Complexity)

    Recall that $k = O(1)$ is a constant.
    
    Phase 0 requires $O(m)$ time, as $\decompose$ requires $O(k(m + n))$-time.

    We examine Phase 1. 
    Constructing the Euler Tours and $\bfs$ trees require $\tO{n^2}$-time over all iterations, as $|C_i||F_i| = \tO{n^2}$.
    The expensive step therefore is the \minplus product, which we compute with an invocation to $\monotoneMinPlus$.
    By Theorem \ref{thm:rect-monotone-min-plus}, this requires,
    \begin{equation*}
        \bigtO{n^{\frac{(a + b + \mu) + \omega(a, b, c)}{2}}}
    \end{equation*}

    where we apply
    \begin{align*}
        a &= 1 - \frac{\floor{k/3}}{\floor{k/3} + 2} x \\
        b &= 1 - x \\
        c &= 1 - \frac{\floor{k/3} - 1}{\floor{k/3} + 2} x \\
        \mu &= 1
    \end{align*}
    since all distances in $G$ are bounded by $n$.

    In Phase 2, each $|D_{j}| = \bigtO{n^{1 - (k - j) \frac{x}{\floor{k/3} + 2}}}$ and $|E_j| \leq \bigO{n^{1 + (k - j + 1) \frac{x}{\floor{k/3} + 2}}}$.
    Therefore, each invocation of $\dijkstra$ in Phase 2 requires time $\bigtO{n^{2 + \frac{x}{\floor{k/3} + 2}}}$.
    
    Finally, we balance terms to optimize,
    \begin{equation*}
        \frac{\left(3 - \frac{2 \floor{k/3} + 2}{\floor{k/3} + 2} x\right) + \omega\left(1 - \frac{\floor{k/3}}{\floor{k/3} + 2} x, 1 - x, 1 - \frac{\floor{k/3} - 1}{\floor{k/3} + 2} x\right)}{2} = 2 + \frac{x}{\floor{k/3} + 2}
    \end{equation*}

    To get the exact expression, we plug in $k = \frac{3 \beta - 2}{2}$ and in particular $\floor{k/3} = \frac{\beta}{2} - 1$.

    \begin{align*}
        \frac{\left(3 - \frac{2 \beta}{\beta + 2} x\right) + \omega\left(1 - \frac{\beta - 2}{\beta + 2} x, 1 - x, 1 - \frac{\beta - 4}{\beta + 2} x\right)}{2} &= 2 + \frac{2x}{\beta + 2} \\
        \left(3 - \frac{2 \beta}{\beta + 2} x\right) + \omega\left(1 - \frac{\beta - 2}{\beta + 2} x, 1 - x, 1 - \frac{\beta - 4}{\beta + 2} x\right) &= 4 + \frac{4x}{\beta + 2} \\
         \omega\left(1 - \frac{\beta - 2}{\beta + 2} x, 1 - x, 1 - \frac{\beta - 4}{\beta + 2} x\right) &= 1 + \frac{4 + 2 \beta}{\beta + 2}x
    \end{align*}
    \Cref{tbl:k-additive-apasp} exhibits the running times for a few values of $\beta$ utilizing \cite{Complexity}.
\end{proof}

In the following lemma, we show that every distance estimate produced is feasible.
That is, each distance estimate can be attained by some path in $G$.

\begin{lemma}
    \label{lemma:additive-feasible-estimates}
    Algorithm \ref{alg:k-additive-apasp} returns $\hat{\delta}$ such that $\delta(u, v) \leq \hat{\delta}(u, v)$ for all $u, v \in V$.
\end{lemma}

\begin{proof}
    This holds simply by observing that every distance estimate is produced by some path in the original graph $G$.

    Every estimate in Phase 1 found by concatenating two paths in $G_i$, a subgraph of $G$.

    In Phase 2, the only edges added are of the form $u \times V$.
    However, the weights of these edges, if finite, are computed by a previous step as the length of some path from $u$ to the relevant vertex $w \in V$.
\end{proof}

We present the key lemma for correctness below.
In particular, we prove that we have a good additive approximation after the $(k_0 + 3)$-th iteration.

\begin{lemma}
    \label{lemma:k0+3-approx-monotone}
    Let $P$ be a shortest path between $u, v$.
    In Phase 2 of Algorithm \ref{alg:k-additive-apasp}, let $w$ be a vertex of $P$ with $\deg(w) \geq s_{k_0 + 3}$.
    Then, for every vertex $x \in P$, the $(k_0 + 3)$-th iteration of Phase 2 finds a path of length $\delta(w, x) + 5$ from $w^*$ to $x$ where $w^* = r(w, D_{k_0 + 3})$.
\end{lemma}

\begin{proof}
    Let $P_{w, x}$ denote the sub-path of $P$ between $w$ and $x$.

    \paragraph{Case 1: $\deg(P_{w, x}) \leq s_{k_0}$}

    Consider the path $P_{w, x}$.
    By Lemma \ref{lemma:sparse-apasp-approx}, since $\deg(w) \geq s_{k_0 + 3}$, we have that, $L(P_{w, x}) \subset \set{k_0 + 1, k_0 + 2}$ so that,
    \begin{equation*}
        \hat{\delta}_{k_0 + 3}(w^*, x) \leq \delta(w, x) + 5
    \end{equation*}

    \paragraph{Case 2: $\deg(P_{w, x}) > s_{k_0}$}

    We now arrive at the key case depending on Phase 1.
    Let $y$ denote the vertex in $P_{w, x}$ closest to $x$ with degree $\deg(y) \geq s_{k_0 + 2}$ and $z$ be a vertex of maximum degree on $P_{w, x}$.
    Then, let $i^*$ be the integer $1 \leq i \leq l$ such that $\frac{n}{2^{i^*}} \leq \deg(z) < \frac{n}{2^{i^* - 1}}$.
    Let $y^* = r(y, D_{k_0 + 2})$ and $z^* = r(z, C_{i^*})$.
    
    In particular, since $P \subset F_{i^*}$ we have in the $i^*$ iteration of Phase 1 that $X_{i^*}(z^*, w^*) \leq \delta(z, w) + 2$ and $X_{i^*}(z^*, y^*) \leq \delta(z, y) + 2$.
    Now, since $B_i'$ is an Euler tour distance matrix, it is therefore a $1$-bounded difference matrix by the triangle inequality.
    Then, the procedure $\bdToMonotone$ creates a row-monotone matrix keeping values bounded in $O(n)$.
    Finally, $B_i$ is a row-monotone matrix as the sub-matrix of a row-monotone matrix.
    Thus, we can use $\monotoneMinPlus(A_{i^*}, B_{i^*})$ to compute,
    \begin{equation*}
        \hat{\delta}_{i^*}(w^*, y^*) \leq X_{i^*}(w^*, z^*) + X_{i^*}(z^*, y^*) \leq \delta(w, z) + \delta(z, y) + 4 = \delta(w, y) + 4
    \end{equation*}
    
    Now, in the $(k_0 + 3)$-th iteration in Phase 2, when we execute $\dijkstra$ from $w^*$, we have an edge to $y^*$ with weight at most $\delta(w, y) + 4$.
    If $\deg(x) > s_{k_0 + 2}$, and therefore $y = x$, we take an extra edge $(x^*, x) \in E^*$ to find a path of length at most $\delta(w, y) + 4 + 1 = \delta(w, x) + 5$.
    Otherwise, after using the edge $(y^*, y) \in E^*$, the remaining edges of $P_{w, x}$ are in $E_{k_0 + 3}$ and we can conclude,
    \begin{equation*}
        \hat{\delta}_{k_0 + 3}(w^*, x) \leq \hat{\delta}_{i^*}(w^*, y^*) + 1 + \delta(y, x) \leq \delta(w, y) + \delta(y, x) + 5 = \delta(w, x) + 5
    \end{equation*}
\end{proof}

\subsection{Additive Approximations for Constant Length Paths}
\label{sec:additive-approx-constant}

In this section, we give a simpler algorithm that uses fast matrix multiplication to obtain fast algorithms for general additive approximations on small distances.
We employ this result in the proof of Theorem \ref{thm:2-approx-apsp}.

\begin{theorem}
    \label{thm:gen-bounded-additive-apasp}
    Let $\beta \geq 4$ be an even integer and $C > 0$ be a constant.
    Let $G$ be an undirected, unweighted graph with $n$ vertices.
    Algorithm \ref{alg:bounded-additive-apasp} computes $\hat{\delta}$ in time,
    \begin{equation*}
        \bigtO{n^{2 + \frac{2x}{\beta+2}}}
    \end{equation*}
    $x$ is the solution to,
    \begin{equation*}
        \omega \left( 1 - \frac{\beta-2}{\beta+2} x, 1 - x, 1 - \frac{\beta-4}{\beta+2} x \right) = 2 + \frac{2x}{\beta+2}
    \end{equation*}
    such that $\delta(u, v) \leq \hat{\delta}(u, v)$ for all $u, v$ and whenever $\delta(u, v) \leq C$, $\hat{\delta}(u, v) \leq \delta(u, v) + \beta$.
\end{theorem}

The above result also implies a more efficient algorithm for approximating distances up to small polynomials.

\begin{corollary}
    Let $\eps > 0$.
    Then, there is an algorithm computing $\hat{\delta}$ in time
    \begin{equation*}
        \bigtO{n^{2 + \frac{2x}{\beta+2}}}
    \end{equation*}
    $x$ is the solution to,
    \begin{equation*}
        \omega \left( 1 - \frac{\beta-2}{\beta+2} x, 1 - x, 1 - \frac{\beta-4}{\beta+2} x \right) + \eps = 2 + \frac{2x}{\beta+2}
    \end{equation*}
    such that $\delta(u, v) \leq \hat{\delta}(u, v)$ for all $u, v$ and whenever $\delta(u, v) \leq n^{\eps}$, $\hat{\delta}(u, v) \leq \delta(u, v) + \beta$.
\end{corollary}

\begin{proof}
    The only difference with the proof of Theorem \ref{thm:gen-bounded-additive-apasp} is the setting of $x$ and the complexity analysis.
    For bounded $C$, the computation of the \minplus product requires time,
    \begin{equation*}
        C n^{\omega \left( 1 - \frac{\beta-2}{\beta+2} x, 1 - x, 1 - \frac{\beta-4}{\beta+2} x \right)}
    \end{equation*}

    Now, since the entries are bounded instead by $n^{\eps}$, we bound the computation of the \minplus product by,
    \begin{equation*}
        n^{\omega \left( 1 - \frac{\beta-2}{\beta+2} x, 1 - x, 1 - \frac{\beta-4}{\beta+2} x \right) + \eps}
    \end{equation*}
\end{proof}

We now present a high level overview of Theorem \ref{thm:gen-bounded-additive-apasp}.

\paragraph{High Level Overview}
The algorithm is essentially a simplification of the previous Algorithm \ref{alg:k-additive-apasp}.
Indeed, we can invoke $\boundedMinPlus$ in place of $\monotoneMinPlus$, which is a faster algorithm and does not require the transformation $\bdToMonotone$.

\paragraph{Algorithm}

Algorithm \ref{alg:bounded-additive-apasp} is identical to Algorithm \ref{alg:k-additive-apasp}.
The only modification is the invocation of $\boundedMinPlus$ in place of $\additiveAPASP$, as the entries to the matrices are now bounded by constant $C$.

Below, we exhibit some running times for specific choices of $\beta$, using \cite{Complexity}.

\begin{corollary}
    \label{cor:bounded-additive-examples}
    We obtain the following running times for an additive approximation on distances $\delta(u, v) \leq C$.
    \begin{center}
        \begin{tabular}{ |c|c| } 
             \hline
             $\beta$ & time \\ 
             \hline
             $4$ & $n^{2.09314841}$ \\
             $6$ & $n^{2.05794292}$ \\
             $8$ & $n^{2.04220679}$ \\
             $10$ & $n^{2.03322582}$ \\
             \hline
        \end{tabular}
    \end{center}
\end{corollary}

We now present the algorithm and proof of Theorem \ref{thm:gen-bounded-additive-apasp}.

\subsubsection*{Warm-Up: \texorpdfstring{$\beta = 4$}{b = 4}}

We begin by showing a proof of correctness for the case of $\beta = 4$.
The proof is almost identical to Lemma \ref{lemma:monotone-4-approximation}.

\begin{lemma}
    \label{lemma:bounded-4-approximation}
    Algorithm $\boundedAdditiveAPASP(G, 4, C)$ returns an estimate $\hat{\delta}$ satisfying:
    \begin{enumerate}
        \item $\delta(u, v) \leq \hat{\delta}(u, v)$ for all $u, v$
        \item $\hat{\delta}(u, v) \leq \delta(u, v) + 4$ for all $\delta(u, v) \leq C$
    \end{enumerate}
\end{lemma}

\begin{proof}
    Let $P$ be a shortest path between $u, v$.
    Note that $k = 5, k_0 = 2$.
    
    Suppose $P \subset E_3 = E_{k_0 + 1}$. 
    Then, $L(P, 5) \subset \set{3, 4}$ and $|L(P, 5)| \leq 2$.
    By Lemma \ref{lemma:sparse-apasp-approx}, since $D_5 = D_k = V$, we obtain a $+4$ approximation in Phase 2.

    Thus, $P$ has some vertex with degree at least $s_2 = n^{x}$.
    Let $y$ denote the vertex in $P$ with degree at least $s_4$ closest to $v$.
    Let $z$ be a vertex of maximum degree on $P$ and $i^*$ be the integer $1 \leq i \leq l$ such that $\frac{n}{2^{i^*}} \leq \deg(z) < \frac{n}{2^{i^* - 1}}$.
    Since $P \subset F_{i^*}$, we have in the $i^*$-th iteration of Phase 1 that $X_{i^*}(z^*, u) \leq \delta(z, u) + 1 \leq C + 1$ and $X_{i^*}(z^*, y^*) \leq \delta(z, y) + 2 \leq C + 2$ where $z^* = r(z, C_{i^*})$, $y^* = r(y, D_{4})$.
    Therefore, since the entries in $A_{i^*}, B_{i^*}$ are finite, we compute in Line \ref{line:additive-apasp:bounded-min-plus},
    \begin{equation*}
        \hat{\delta}_{i^*}(u, y^*) \leq X_{i^*}(u, z^*) + X_{i^*}(z^*, y^*) \leq \delta(u, z) + \delta(z, y) + 3 = \delta(u, y) + 3
    \end{equation*}
    
    Now, in the $5$-th iteration in Phase 2, when we execute Dijkstra from $u$, we have an edge $(u, y^*)$ with weight at most $\delta(u, y) + 3$, an edge $(y^*, y) \in E$ and the remaining edges in $E_5 = E_k$, so that,
    \begin{equation*}
        \hat{\delta}(u, v) \leq \hat{\delta}_{i^*}(u, y^*) + 1 + \delta(y, v) \leq \delta(u, y) + \delta(y, v) + 4 \leq \delta(u, v) + 4
    \end{equation*}
\end{proof}

\IncMargin{1em}
\begin{algorithm}[H]
\SetKwInOut{Input}{Input}
\SetKwInOut{Output}{Output}
\Input{Unweighted, undirected Graph $G = (V, E)$ with $n$ vertices; approximation parameter $\beta$; distance bound $C$}
\Output{Distance estimate $\hat{\delta}: U \times V \rightarrow \Z$ such that $\delta(u, v) \leq \hat{\delta}(u, v)$ for all $u, v \in V$ and $\hat{\delta}(u, v) \leq \delta(u, v) + \beta$ whenever $\delta(u, v) \leq C$.}

\BlankLine

\textcolor{blue}{Phase 0: Set up and Decompose Graph}

$\hat{\delta}(u, v) \gets \begin{cases}
    1 & (u, v) \in E \\
    \infty & \otherwise
\end{cases}$

$k \gets \frac{3 \beta - 2}{2}$

$x \gets$ solution to $\omega \left( 1 - \frac{\floor{k/3}}{\floor{k/3}+2} x, 1 - x, 1 - \frac{\floor{k/3} - 1}{\floor{k/3}+2} x \right) = 2 + \frac{x}{\floor{k/3} + 2}$

$s_{k - i} \gets n^{i \cdot \frac{x}{\floor{k/3} + 2}}$ for all $1 \leq i \leq \floor{k/3} + 2$

$k_0 \gets k - (\floor{k/3} + 2) = k - \floor{k/3} - 2$

$(D_{k_0}, D_{k_0 + 1}, \dotsc, D_k), (E_{k_0}, E_{k_0 + 1}, \dotsc, E_k), E^* \gets \decompose(G, (s_{k_0}, s_{k_0 + 1}, \dotsc, s_{k - 1}))$

\textcolor{blue}{Phase 1: Estimate Distances on High-Degree Paths}

$t_{i} \gets \frac{n}{2^i}$ for all $1 \leq i \leq l - 1 = \ceil{(1 - x) \log n}$

$(C_{1}, C_{2}, \dotsc, C_{l}), (F_{1}, F_{2}, \dotsc, F_{l}), F^* \gets \decompose(G, (t_{1}, t_{2}, \dotsc, t_{l - 1}))$

\For{$0 \leq i \leq l$}{
    \For{$w \in C_i$}{
        $X_i(w) \gets \bfs(w)$ on $G_i = \left(V, F_{i}\right)$
        
        $\hat{\delta}(w, v) \gets \min(\hat{\delta}, X_i(w, v))$ for all $v \in V$
    }
    
    Construct $|D_{k_0 + 2}| \times |C_i|$ matrix $A_i$ where $A_i(v, w) = \begin{cases}
        X_i(w, v) & X_i(w, v) \leq C + 2 \\ 
        \infty & \otherwise
    \end{cases}$
    
    Construct $|C_i| \times |D_{k_0 + 3}|$ matrix $B_i$ where $B_i(w, v) = \begin{cases}
        X_i(w, v) & X_i(w, v) \leq C + 2 \\ 
        \infty & \otherwise
    \end{cases}$
    
    $\hat{\delta} \gets \min(\hat{\delta}, \boundedMinPlus(A_i, B_i, C+2))$
    \label{line:additive-apasp:bounded-min-plus}
}

\textcolor{blue}{Phase 2: Estimate Distances on Low-Degree Paths}

\For{$k_0 + 1 \leq j \leq k$}{
    \For{$w \in D_{j}$}{
        $G_{j, w} = (V, E_j \cup (w \times V) \cup E^*)$
        
        $\hat{\delta} \gets \dijkstra(G_{j, w}, w, \hat{\delta})$
    }
}

\caption{$\boundedAdditiveAPASP (G, \beta, C)$} 
\label{alg:bounded-additive-apasp}
\end{algorithm}
\DecMargin{1em}

\subsubsection*{General Case: \texorpdfstring{$\beta \geq 6$}{b >= 6}}

We now present the proof of Theorem \ref{thm:gen-bounded-additive-apasp}.

\begin{proof}
    (Correctness)
    We will ignore cases that are too similar to Theorem \ref{thm:gen-additive-apasp}.
    Note by the definition of $k \gets \frac{3 \beta - 2}{2}$ and $\beta$ is an even integer,
    \begin{equation*}
        2 \left( \floor{k/3} + 1 \right) = 2 \left( \floor{\frac{\beta}{2} - \frac{1}{3}} + 1 \right) \leq \beta
    \end{equation*}
    
    By Lemma \ref{lemma:bounded-additive-feasible-estimates}, it suffices to check that for $\delta(u, v) \leq C$, we have $\hat{\delta}(u, v) \leq \delta(u, v) + 2(\floor{k/3} + 1)$.
    
    Let $u, v \in V$ be a pair of vertices such that $\delta(u, v) \leq C$.
    Let $P$ be a shortest path between $u, v$.
    Let $\deg(P) = \max_{v \in P} \deg(v)$ be the maximum degree of any vertex on $P$.

    \paragraph{Case 1: $\deg(P) < s_{k_0}$}

    Since $L(P, k) \subset \set{k_0 + 1, \dotsc, k - 1}$, correctness holds as in Theorem \ref{thm:gen-additive-apasp}.

    \paragraph{Case 2: $\deg(P) \geq s_{k_0}$}
    
    By assumption, there is a vertex $w \in P$ with $\deg(w) \geq s_{k_0 + 3}$.
    Then, with Lemma \ref{lemma:k0+3-approx-distance}, we have that for all $x \in P$,
    \begin{equation*}
        \hat{\delta}_{k_0 + 3}(w^*, x) \leq \delta(w, x) + 5
    \end{equation*}
    
    The remaining proof follows as in Theorem \ref{thm:gen-additive-apasp}.
\end{proof}

\begin{proof}
    (Time Complexity)

    As the only modifications are in Phase 1, we examine only the complexity of this phase.
    Each $|C_i| = \tO{\frac{n}{t_i}} = \tO{2^i}$
    Each graph $F_i$ has $O(n t_i)$ edges.
    Thus, each $\bfs$ requires $\tO{n^2}$-time.
    Over $k$ iterations, this is again $\tO{n^2}$-time.
    
    Now, $|D_{k_0 + 2}| = \bigtO{\frac{n}{s_{k_0 + 2}}} = \bigtO{\frac{n}{s_{k - \floor{k/3}}}} = \bigtO{n^{1 - \frac{\floor{k/3}}{\floor{k/3}+2} x}}$.
    Similarly, we can upper bound the size of $D_{k_0 + 3}$ as $|D_{k_0 + 3}| = \bigtO{\frac{n}{s_{k - (\floor{k/3} - 1)}}} = \bigtO{n^{1 - \frac{\floor{k/3 - 1}}{\floor{k/3}+2} x}}$.
    Finally, $|C_i| = \tO{2^l} = \tO{n^{1 - x}}$.
    Thus, each $\boundedMinPlus$ requires time,
    \begin{equation*}
        \omega \left( 1 - \frac{\floor{k/3}}{\floor{k/3}+2} x, 1 - x, 1 - \frac{\floor{k/3} - 1}{\floor{k/3}+2} x \right)
    \end{equation*}

    Recall that Phase 2 requires time $\bigtO{n^{2 + \frac{x}{\floor{k/3} + 2}}}$.
    Finally, we balance terms to optimize,
    \begin{equation*}
        \omega \left( 1 - \frac{\floor{k/3}}{\floor{k/3}+2} x, 1 - x, 1 - \frac{\floor{k/3} - 1}{\floor{k/3}+2} x \right) = 2 + \frac{x}{\floor{k/3} + 2}
    \end{equation*}
    To get the exact expression, we plug in $k = \frac{3 \beta - 2}{2}$ and in particular $\floor{k/3} = \frac{\beta}{2} - 1$.
\end{proof}

In the following lemma, we show that every distance estimate produced is feasible.
That is, each distance estimate can be attained by some path in $G$.

\begin{lemma}
    \label{lemma:bounded-additive-feasible-estimates}
    Algorithm \ref{alg:bounded-additive-apasp} returns $\hat{\delta}$ such that $\delta(u, v) \leq \hat{\delta}(u, v)$ for all $u, v \in V$.
\end{lemma}

\begin{proof}
    This holds simply by observing that every distance estimate is produced by some path in the original graph $G$, as argued in Lemma \ref{lemma:additive-feasible-estimates}.
\end{proof}

We present the key lemma for correctness below.
In particular, we prove that we have a good additive approximation after the $(k_0 + 3)$-th iteration.

\begin{lemma}
    \label{lemma:k0+3-approx-distance}
    Let $P$ be a shortest path between $u, v$.
    In Phase 2 of Algorithm \ref{alg:bounded-additive-apasp}, let $w$ be a vertex of $P$ in $V_{k_0 + 3}$. 
    Then, for every vertex $x \in P$, the $(k_0 + 3)$-th iteration of Phase 2 finds a path of length $\delta(w, x) + 5$ from $w^*$ to $x$.
\end{lemma}

\begin{proof}
    Let $P_{w, x}$ denote the sub-path of $P$ between $w$ and $x$.

    \paragraph{Case 1: $\deg(P_{w, x}) \leq s_{k_0}$}

    This is identical to Case 1 of Lemma \ref{lemma:k0+3-approx-monotone}.

    \paragraph{Case 2: $\deg(P_{w, x}) > s_{k_0}$}
    
    We now arrive at the key case depending on Phase 1.
    Let $y$ denote the vertex in $P_{w, x}$ closest to $x$ with degree $\deg(y) \geq s_{k_0 + 2}$ and $z$ be a vertex of maximum degree on $P_{w, x}$.
    Then, let $i^*$ be the integer $1 \leq i \leq l$ such that $\frac{n}{2^{i^*}} \leq \deg(z) < \frac{n}{2^{i^* - 1}}$.
    Let $y^* = r(y, D_{k_0 + 2})$ and $z^* = r(z, C_{i^*})$.
    
    Since $P \subset F_{i^*}$, in the $i^*$ iteration of Phase 1 $X_{i^*}(z^*, w^*) \leq \delta(z, w) + 2 \leq C + 2$ and $X_{i^*}(z^*, y^*) \leq \delta(z, y) + 2 \leq C + 2$.
    Since the entries are finite, we use $\boundedMinPlus(A_{i^*}, B_{i^*})$ to compute,
    \begin{equation*}
        \hat{\delta}_{i^*}(w^*, y^*) \leq X_{i^*}(w^*, z^*) + X_{i^*}(z^*, y^*) \leq \delta(w, z) + \delta(z, y) + 4 = \delta(w, y) + 4
    \end{equation*}
    
    Now, in the $(k_0 + 3)$-th iteration in Phase 2, when we execute Dijkstra from $w^*$, we have an edge to $y^*$ with weight at most $\delta(w, y) + 4$.
    If $\deg(x) > s_{k_0 + 2}$, and therefore $y = x$, we take an extra edge $(x^*, x) \in E^*$ to find a path of length at most $\delta(w, y) + 4 + 1 = \delta(w, x) + 5$.
    Otherwise, after using the edge $(y^*, y) \in E^*$, the remaining edges of $P_{w, x}$ are in $E_{k_0 + 3}$ and we can conclude,
    \begin{equation*}
        \hat{\delta}_{k_0 + 3}(w^*, x) \leq \hat{\delta}_{i^*}(w^*, y^*) + 1 + \delta(y, x) \leq \delta(w, y) + \delta(y, x) + 5 = \delta(w, x) + 5
    \end{equation*}
\end{proof}

\section{Weighted Approximation via Approximate \texorpdfstring{\minplus}{(min, +)} Product}
\label{sec:weighted-additive-apasp}

In this section, we generalize our results to graphs with arbitrary non-negative weights.
Let $\wt(e)$ denote the weight of an edge $e \in E$.
For any subset of edges $S \subset E$, let $\wt(S) = \sum_{e \in S} \wt(e)$ denote the total weight of edges in $S$.
Throughout, we assume all weights are non-negative.

\paragraph*{Preliminaries}
We first require some preliminary results analogous to unweighted graphs for computing dominating sets, and degree decompositions. Note that this decomposition is slightly different than what is used in previous sections.

\begin{restatable}{lemma}{weightdecompositionlemma}
    \label{lemma:weight-decomposition}
    Let $G$ be a weighted, undirected graph on $n$ vertices.
    For a given edge $e$ incident to vertex $v$, let $\rank_v(e)$ be the order of $e$ among all edges incident to $v$ when ranked in increasing order of weight.
    For example, if $e$ is the lightest edge incident to $v$, $\rank_v(e) = 1$.
    For a neighbor $w \in N(v)$, we also denote $\rank_v(w) = \rank_v(v, w)$ to be the rank of the edge $(v, w)$.
    
    Given thresholds $s_1 > s_2 > \dotsc > s_{k - 1}$, there is an algorithm $\wDecompose$ that outputs edge sets $\set{E_i}_{i = 1}^{k}$, edge set $E^*$, and vertex sets $\set{D_i}_{i = 1}^{k}$ satisfying,
    \begin{enumerate}
        \item $E_i = \set{(u, v) \in E \given \min(\rank_u(u, v), \rank_v(u, v)) < s_{i - 1}}$.
        \item $D_i$ dominates $V_{s_i} = \set{v \in V \given \deg(v) \geq s_i}$ and $|D_i| = \bigtO{\frac{n}{s_i}}$
        \item $E^* = \bigcup_{i = 1}^{k} E_i^{*}$ where each $E_i^* \subset E$ has for every $v \in V_{s_i}$ at least one edge $(v, w) \in E_i^{*}$ where $w \in D_i$ and $\rank_v(v, w) < s_i$.
    \end{enumerate}
    Furthermore, $\wDecompose$ runs in $\tO{k n^2}$-time.

    We define the level of a vertex and edge analogously as in the unweighted case.
    For a given vertex $v \in V$, define the {\bf level of $v$}, denoted $\level(v)$, as the integer $i$ such that $s_i \leq \deg(v) < s_{i - 1}$.
    For a given edge $e \in E$, define the {\bf level of $e$}, denoted $\level(e)$, as the integer $i$ such that $e \in E_i \setminus E_{i + 1}$.
\end{restatable}

\subsection{\texorpdfstring{$(1 + \eps, 2 w)$}{(1+e, 2w)} Approximation}

For a given path $P$, let $\heavy(P)$ denote the weight of the heaviest edge in $P$.
For a pair of vertices $u, v$, let $w_{u, v}$ denote the minimum $h(P)$ over all shortest paths $P$ between $u, v$.

For any integer $k \geq 1$, let $\heavy(P, k)$ denote the weight of the $k$-th heaviest edge in $P$ and $\kHeavy(P, k)$ denote the sum of the weights of the $k$ heaviest edges in $P$.
We analogously define $w_{u, v}(k) = \min_{P} \kHeavy(P, k)$ over shortest paths $P$ between $u, v$.

Using combinatorial methods, Cohen and Zwick \cite{cohen2001smallstretch} gave a $O(n^{7/3})$ algorithm for $+2 w_{u, v}$ approximation.
They in fact prove a stronger result and gain a $(7/3, 0)$ multiplicative approximation.
With Fast Matrix Multiplication, Berman and Kasiviswanathan \cite{berman2007approxapsp} gave a $(1+\eps, 2 w_{u, v})$ algorithm in time $O(n^{2.25})$.
In this work, we use our techniques to improve upon their algorithm and design a $(1+\eps, 2 w_{u, v})$ algorithm with running time roughly $\tO{n^{2.152}}$.

\begin{theorem}
    \label{thm:weighted-2-apasp}
    Let $\eps > 0$.
    Let $G$ be an undirected, weighted graph with $n$ vertices and weight function $w: E \rightarrow \R^+$.
    Algorithm \ref{alg:2-weighted-additive-apasp} computes $\hat{\delta}$ in time,
    \begin{equation*}
        \bigtO{\frac{n^{2.15195331}}{\eps}}
    \end{equation*}
    such that $\delta(u, v) \leq \hat{\delta}(u, v) \leq (1 + \eps) \delta(u, v) +  2 w_{u, v}$.
\end{theorem}

Throughout, we use the approximate \minplus product of Bringmann et al. \cite{bringmann2019noscaling}.

\begin{restatable*}{theorem}{approxminplusthm}
    \label{thm:approximate-min-plus}
    Let $A, B$ be two integer $n \times n$ matrices with non-negative entries.
    Let $\eps > 0$.
    There is an algorithm $\approximateMinPlus$ that returns $C$ in time $\bigtO{n^{\omega}/\eps}$ such that for all $1 \leq i, j \leq n$,
    \begin{equation*}
        (A * B)_{ij} \leq C_{ij} \leq (1 + \eps) (A * B)_{ij}
    \end{equation*}
\end{restatable*}

We provide the proof in \Cref{sec:prelims:fast-min-plus} for completeness.

\paragraph{High Level Overview}

Let $x$ be a parameter to be set later.
Let $s_1 = n^x$ and $s_2 = n^{x/2}$ and $D_1 \subset D_2 \subset D_3 = V$ and $E_3 \subset E_2 \subset E_1 = E$ be the output of $\wDecompose$.
There are 3 cases to consider.
Let $P$ be a shortest path between $u, v$.
If $P \subset E_3$, then we compute an exact distance from $u$ in the graph $G \subset G_{3, u}$.

Suppose $P \subset E_2$ but $P \not\subset E_3$.
Let $e = (a, b)$ be the closest edge to $u$ such that $e \notin E_3$ and let $a$ be the vertex closer to $u$ so that $P_{a, u} \subset E_3$.
Then, $\deg(a) \geq s_2$ so there is $a^* = r(a, D_2)$ such that $\rank_{a}(a^*) < s_2$.
Since $P \subset E_2$, $\hat{\delta}(a^*, v) \leq \wt(a^*, a) + \delta(a, v)$.
Then, $(v, a^*, a) \circ P_{a, u} \subset G_{3, v}$ so $\dijkstra$ computes an estimate,
\begin{equation*}
    \hat{\delta}(v, u) \leq \hat{\delta}(v, a^*) + \wt(a^*, a) + \delta(a, u) \leq \delta(v, u) + 2 \wt(a^*, a) \leq \delta(v, u) + 2 \wt(a, b) \leq \delta(v, u) + 2 w_{u, v}
\end{equation*}
as $\rank_a(a^*) < s_2 \leq \rank_a(b)$.

Finally, suppose $P \not\subset E_2$.
Let $P \subset F_i$ and $P \not\subset F_{i + 1}$.
Let $e = (a, b) \in F_i \setminus F_{i + 1}$.
Since $P \subset F_i$,
\begin{align*}
    \hat{\delta}(a^*, u) &\leq \wt(a^*, a) + \delta(a, u) \\
    \hat{\delta}(a^*, v) &\leq \wt(a^*, a) + \delta(a, v)
\end{align*}

Then, from the $\approximateMinPlus$ product,

\begin{equation*}
    \hat{\delta}(u, v) \leq \left( 1 + \frac{\eps}{3} \right) \left( \delta(u, v) + \wt(a^*, a) \right) \leq \left( 1 + \frac{\eps}{3} \right) \left( \delta(u, v) + 2 \wt(a, b) \right) \leq (1 + \eps) \delta(u, v) + 2 w_{u, v}
\end{equation*}
where $\rank_a(a^*) < t_j < t_{j - 1} \leq \rank_a(b)$.

\begin{algorithm}[H]
\SetKwInOut{Input}{Input}\SetKwInOut{Output}{Output}
\Input{Weighted, undirected Graph $G = (V, E)$ with $n$ vertices}
\Output{Distance estimate $\hat{\delta}: U \times V \rightarrow \Z$ such that $\delta(u, v) \leq \hat{\delta}(u, v) \leq \delta(u, v) + 2 w_{u, v}$ for all $u, v \in V$}

\BlankLine

\textcolor{blue}{Phase 0: Set up and Decompose Graph}

$\hat{\delta}(u, v) \gets \begin{cases}
    1 & (u, v) \in E \\
    \infty & \otherwise
\end{cases}$

$x \gets $ solution to $\omega\left(1, 1 - x, 1 \right) = 2 + \frac{x}{2}$

$s_1, s_2 \gets n^{x}, n^{\frac{x}{2}}$

$(D_{1}, D_{2}, D_{3}), (E_{1}, E_{2}, E_{3}), E^* \gets \wDecompose(G, (s_1, s_2))$

\textcolor{blue}{Phase 1: Estimate Distances on High-Degree Paths}

$t_{i} \gets \frac{n}{2^i}$ for all $1 \leq i \leq l - 1 = \ceil{(1 - x) \log n}$

$(C_{1}, C_{2}, \dotsc, C_{l}), (F_{1}, F_{2}, \dotsc, F_{l}), F^* \gets \wDecompose(G, (t_{1}, t_{2}, \dotsc, t_{l - 1}))$

\For{$0 \leq i \leq l$}{
    \For{$w \in C_i$}{
        $X_i(w) \gets \dijkstra(w)$ on $G_i = \left(V, F_{i}\right)$
        
        $\hat{\delta}(w, v) \gets \min(\hat{\delta}, X_i(w, v))$ for all $v \in V$
    }

    Construct $V \times C_i$ matrix $A_i$ where $A_i(v, w) = X_i(w, v)$
    
    $\hat{\delta} \gets \min\left(\hat{\delta}, \approximateMinPlus\left(A_i, A_i^T, 1 + \frac{\eps}{3}\right)\right)$
}

\textcolor{blue}{Phase 2: Estimate Distances on Low-Degree Paths}

\For{$2 \leq j \leq 3$}{
    \For{$w \in D_{j}$}{
        $G_{j, w} = (V, E_j \cup (u \times V) \cup E^*)$
        
        $\hat{\delta} \gets \dijkstra(G_{j, w}, w, \hat{\delta})$
    }
}

\caption{$\alg{Weighted2AdditiveAPASP}(G)$} 
\label{alg:2-weighted-additive-apasp}
\end{algorithm}

\begin{proof}
    (Correctness)
    We follow a similar proof to that of Cohen and Zwick \cite{cohen2001smallstretch}.
    Fix a pair of vertices $u, v$ and let $P$ be a shortest path between $u, v$.

    \paragraph{Case 1: $P \not\subset E_2$}

    Let $e_0 \in P$ be edge such that $e_0 \notin E_2$.
    Let $1 \leq j \leq l - 1$ such that $P \subset F_j$ but $P \cap (F_j \setminus F_{j + 1}) = \emptyset$. 
    Such a $j$ must exist as $F_1 = E$ and $F_l$ contains edges with rank at most $t_{l - 1} = \frac{n}{2^{l - 1}} \leq n^x$ so $e_0 \notin F_l$.
    
    In particular, we now choose a specific $e = (a, b)$ such that $e \in F_j \setminus F_{j + 1}$.
    By the correctness of $\wDecompose$, $\rank_a(e), \rank_b(e) \geq t_j$.
    Since both $a, b$ have degree at least $t_j$, let $x \in D_{j}$ such that $(x, a) \in E_{j}^*$ so that $\rank_a(x, a) < t_j \leq \rank_a(e)$.
    Now, since $(x, a), P \subset F_{j}$, when we execute $\dijkstra$ from $x$, we have
    \begin{align*}
        \hat{\delta}(x, u) &\leq \delta(a, u) + \wt(x, a) \leq \delta(a, u) + \wt(a, b) \\
        \hat{\delta}(x, v) &\leq \delta(a, v) + \wt(x, a) \leq \delta(a, v) + \wt(a, b)
    \end{align*}

    Then, from the $\approximateMinPlus$ product,
    \begin{equation*}
        \hat{\delta}(u, v) \leq \left( 1 + \frac{\eps}{3} \right) \left(\hat{\delta}(x, u) + \hat{\delta}(x, v) + 2 h(P) \right) \leq (1 + \eps) \delta(u, v) + 2 h(P)
    \end{equation*}
    where the second inequality follows as $\frac{2}{3} w_{u, v} \leq \frac{2}{3} \delta(u, v)$.

    \paragraph{Case 2: $P \subset E_2$ and $P \not\subset E_3$}
    
    Let $e = (a, b) \in E_2 \setminus E_3$ be the closest such edge to $v$ so that $\rank_a(e), \rank_a(b) \geq s_2$.
    Let $b$ be the closest vertex to $v$.
    Then, there is $x \in D_2$ such that $(x, b) \in E_2^*$ and $\rank_b(x, b) < s_2 \leq \rank_b(e)$.
    When we execute $\dijkstra$ from $x$, since $P \subset E_2$, we find,
    \begin{equation*}
        \hat{\delta}(x, u) \leq \delta(b, u) + \wt(x, b) \leq \delta(b, u) + \wt(a, b)
    \end{equation*}

    Then, when we execute $\dijkstra$ from $u$, we use the edge $(u, x)$ and obtain,
    \begin{equation*}
        \hat{\delta}(u, v) \leq \hat{\delta}(x, u) + \wt(x, b) + \delta(b, v) \leq \delta(u, v) + 2 \wt(a, b) \leq \delta(u, v) + 2 h(P)
    \end{equation*}
    
    \paragraph{Case 3: $P \subset E_3$}

    Note $D_3 = V$.
    In this case, $P \subset E_3$ so that when we execute $\dijkstra$ from $u \in D_3$ we find exactly the path $P$ and therefore $\hat{\delta}(u, v) \leq \delta(u, v)$.

    In all cases, we obtain an approximation $\hat{\delta}(u, v) \leq (1 + \eps) \delta(u, v) + 2 h(P)$ for some shortest path $P$.
    We complete the proof by taking the minimum over all shortest paths $P$.
\end{proof}

\begin{proof}
    (Time Complexity)

    The analysis of time complexity is extremely similar to that of Deng et al. \cite{deng2022apasp}.

    From Lemma \ref{lemma:weight-decomposition}, both calls to $\wDecompose$ require time $\tO{n^2}$.
    In Phase 1, for a fixed $i \in [l]$, we can bound $|C_i| = \tO{t_i}$ and each graph $F_i$ has at most $\bigO{\frac{n^2}{t_i}}$ edges so that all invocations of $\dijkstra$ require $\tO{n^2}$ time over all $l$.
    Since $|C_i| \leq \tO{t_i} = \tO{n^{1-x}}$, the invocations of $\approximateMinPlus$ requires $\bigtO{\frac{n^{\omega(1, 1-x, 1)}}{\eps}}$ time.

    In Phase 2, $|D_2| = \tO{n^{1-x/2}}, |E_2| = O(n^{1 + x})$ and $|D_3| = \tO{n^{1 - x}}, |E_3| = O(n^{1 + x/2})$ so that both invocations of $\dijkstra$ require $\tO{n^{2 + x/2}}$ time.

    Thus, we solve for $x$ as the solution to,
    \begin{equation*}
        \omega(1, 1 - x, 1) = 2 + \frac{x}{2}
    \end{equation*}

    Using \cite{Complexity}, we choose $x = 0.30390661$ and obtain the running time $\tO{n^{2.15195331}}$.
\end{proof}

\subsection{\texorpdfstring{$(1 + \eps, \beta w)$}{(1+e, beta w)} Approximation for \texorpdfstring{$\beta \geq 4$}{beta geq 4}}

We now extend Algorithm \ref{alg:k-additive-apasp} to handle weighted graphs, proving the following analogue of Theorem \ref{thm:gen-additive-apasp}.

\weightedgenadditiveapasp

\paragraph{High Level Overview}

The weighted approximation differs from $\additiveAPASP$ in two key points.
First, we use a weighted decomposition of the graph computed by $\wDecompose$ in place of $\decompose$.
Second, we compute an {\it approximate} \minplus product with in place of an {\it exact} \minplus product.

Whereas $\decompose$ returns edge sets $E_i$ which require that at least one endpoint is of low degree, $\wDecompose$ returns edge sets $E_i$ which contain the $s_{i - 1}$ edges of least weight incident to each vertex.
In particular, for a given vertex $v$ of level $j$, there is an edge $(v, v^*)$ to $v^* \in D_j$ that is among the $s_{j - 1}$ lightest edges incident to $v$.
This ensures that when we encounter a blocking vertex, the edge that we take from its dominating vertex is lighter than the edge on the path we would have traversed, ensuring that the additive error incurred is at most twice the weight of some edge in the path.
Since the blocking vertices are distinct, this allows us to bound our overall error in terms of the weights of the $k$ heaviest edges on the path $P$.

However, since weights can be large and not necessarily bounded by $O(n)$, in order to compute a fast \minplus product, we must make a sacrifice in the accuracy of the \minplus product and settle for an approximate product.
By choosing an appropriate error parameter $\eps$, we can compute an efficient approximate \minplus product while incurring only a $(1 + \eps)$-multiplicative factor in our approximation.

\paragraph{Analysis}

Following \Cref{thm:gen-additive-apasp}, we exhibit an $(1 + \eps, 2 w_{u, v}(2))$ as an example, and note that the general algorithm for $\beta \geq 2$ follows similarly.

For $\beta = 2$, we set $k = 5$ and $k_0 = 2$.
If $P \subset E_3$, then the $\sparseAPASP$ algorithm obtains a $+2 w_{u, v}(2)$ approximation following \Cref{lemma:weighted-sparse-apasp}.
Suppose therefore $P \not\subset E_3$.
As in the unweighted case, we decompose $G$ into $\log n$ levels with degree thresholds $t_i = \frac{n}{2_i}$ and compute dominating sets $C_i$ of size $\tO{n/t_i}$ and edge sets $F_i$ of size $O(n t_{i - 1})$ so that all invocations of $\dijkstra$ require $\tO{n^2}$ time.

Let $P \subset F_j$ but $P \not\subset F_{j + 1}$ and let $e = (a, b) \in F_j \setminus F_{j + 1}$.
Then, $a^* = r(a, C_j)$ and $\rank_a(a^*) \leq t_j$.
Let $e' = (c, d)$ be the edge closest to $u$ such that $e' \notin E_5$ and $c$ the endpoint closest to $u$.
Let $c^* = r(c, D_4)$ so $\rank_c(c^*) < s_4$.
In the overview, we will assume $e' \neq e$ as otherwise we obtain a $+2 w_{u, v}$ approximation.
\begin{align*}
    \hat{\delta}(a^*, c^*) &\leq \wt(a^*, a) + \delta(a, c) + \wt(c, c^*) \\ 
    \hat{\delta}(a^*, v) &\leq \wt(a^*, a) + \delta(a, v)
\end{align*}
Then, the approximate \minplus product computes,
\begin{equation*}
    \hat{\delta}(v, c^*) \leq \left(1 + \frac{\eps}{3} \right) \left( \delta(v, c) + 2 \wt(a^*, a) + \wt(c, c^*) \right)
\end{equation*}
Then, $(v, c^*, c) \circ P_{c, u} \subset G_{5, u}$ so that $\dijkstra$ computes,
\begin{align*}
    \hat{\delta}(v, u) &\leq \hat{\delta}(v, c^*) + \wt(c^*, c) + \delta(c, u) \\
    &\leq \left(1 + \frac{\eps}{3} \right) \left( \delta(v, c) + 2 \wt(a^*, a) + \wt(c, c^*) \right) + \wt(c^*, c) + \delta(c, u) \\
    &\leq \left(1 + \frac{\eps}{3} \right) \left( \delta(v, u) + 2 \wt(a^*, a) + 2 \wt(c, c^*) \right) \\
    &\leq \left(1 + \frac{\eps}{3} \right) \left( \delta(v, u) + 2 \wt(e) + 2 \wt(e') \right) \\
    &\leq (1 + \eps) \delta(u, v) + 2 w_{u, v}(2)
\end{align*}

We note that the general argument follows from a similar adaptation to \Cref{thm:gen-additive-apasp}.

\paragraph{Algorithm}

Algorithm \ref{alg:k-weighted-additive-apasp} closely follows Algorithm \ref{alg:k-additive-apasp}.
Observe that the only differences are the usage of $\wDecompose$ in place of $\decompose$ and the usage of $\approximateMinPlus$ in place of $\monotoneMinPlus$.

We present the algorithm pseudocode and detailed analysis below.

\IncMargin{1em}
\begin{algorithm}[H]
\SetKwInOut{Input}{Input}\SetKwInOut{Output}{Output}
\Input{Unweighted, undirected Graph $G = (V, E)$ with $n$ vertices; approximation parameter $\beta$}
\Output{Distance estimate $\hat{\delta}: U \times V \rightarrow \Z$ such that $\delta(u, v) \leq \hat{\delta}(u, v) \leq \delta(u, v) + \beta$ for all $u, v \in V$}

\BlankLine

\textcolor{blue}{Phase 0: Set up and Decompose Graph}

$\hat{\delta}(u, v) \gets \begin{cases}
    1 & (u, v) \in E \\
    \infty & \otherwise
\end{cases}$

$k \gets 3 \beta - 1$

$x \gets$ solution to $\omega \left( 1 - \frac{\floor{k/3}}{\floor{k/3}+2} x, 1 - x, 1 - \frac{\floor{k/3} - 1}{\floor{k/3}+2} x \right) = 2 + \frac{x}{\floor{k/3} + 2}$

$s_{k - i} \gets n^{i \cdot \frac{x}{\floor{k/3} + 2}}$ for all $1 \leq i \leq \floor{k/3} + 2$

$k_0 \gets k - (\floor{k/3} + 2) = k - \floor{k/3} - 2$

$(D_{k_0}, D_{k_0 + 1}, \dotsc, D_k), (E_{k_0}, E_{k_0 + 1}, \dotsc, E_k), E^* \gets \wDecompose(G, (s_{k_0}, s_{k_0 + 1}, \dotsc, s_{k - 1}))$

\textcolor{blue}{Phase 1: Estimate Distances on High-Degree Paths}

$t_{i} \gets \frac{n}{2^i}$ for all $1 \leq i \leq l - 1 = \ceil{(1 - x) \log n}$

$(C_{1}, C_{2}, \dotsc, C_{l}), (F_{1}, F_{2}, \dotsc, F_{l}), F^* \gets \wDecompose(G, (t_{1}, t_{2}, \dotsc, t_{l - 1}))$

\For{$0 \leq i \leq l$}{
    \For{$w \in C_i$}{
        $X_i(w) \gets \dijkstra(w)$ on $G_i = \left(V, F_{i}\right)$
        
        $\hat{\delta}(w, v) \gets \min(\hat{\delta}, X_i(w, v))$ for all $v \in V$
    }
    Construct $|D_{k_0 + 2}| \times |C_i|$ matrix $A_i$ where $A_i(v, w) = X_i(w, v)$
    
    Construct $|C_i| \times |V|$ matrix $B_i$ where $B_i(w, v) = X_i(w, v)$
    
    $\hat{\delta} \gets \min\left(\hat{\delta}, \approximateMinPlus\left(A_i, B_i, 1 + \frac{\eps}{3}\right)\right)$
}

\textcolor{blue}{Phase 2: Estimate Distances on Low-Degree Paths}

\For{$k_0 + 1 \leq j \leq k$}{
    \For{$w \in D_{j}$}{
        $G_{j, w} = (V, E_j \cup (u \times V) \cup E^*)$
        
        $\hat{\delta} \gets \dijkstra(G_{j, w}, w, \hat{\delta})$
    }
}

\caption{$\alg{WeightedAdditiveAPASP}(G, \beta)$} 
\label{alg:k-weighted-additive-apasp}
\end{algorithm}
\DecMargin{1em}

To verify the correctness of our algorithm, we will require the following analog of Lemma \ref{lemma:sparse-apasp-approx}.
Note that this is equivalent to Lemma 5.2 in Cohen and Zwick \cite{cohen2001smallstretch}.

\begin{lemma}
    \label{lemma:weighted-sparse-apasp}
    Let $G$ be an undirected, weighted graph with $n$ vertices and $m$ edges and weight function $w: E \rightarrow \R$.
    Let $n = s_{k_0} > s_{k_0 + 1} > \dotsc > s_{k - 1} > s_{k} = 0$ be degree thresholds.

    For $k_0 + 1 \leq j \leq i$, let $\delta_j$ denote the estimate after the $j$-th iteration.
    Let $P$ be a path between $u \in D_j$ and $v \in V$ such that $P \subset E_{k_0 + 1}$.
    Define the set $L(P, j) = \set{k_0 + 1 \leq i \leq j - 1 \given P \cap (E_i \setminus E_{i + 1}) \neq \emptyset}$.
    Then, there is a set $\edgeset_j \subset P$ such that $|\edgeset_j| \leq |L(P, j)|$, $\edgeset_j \cap E_j = \emptyset$, and $\delta_j(u, v) \leq \wt(P) + 2 \wt(\edgeset_j)$.
\end{lemma}

\begin{proof}
    We proceed by induction.
    In the base case, let $|L(P, j)| = 0$ so that $P \subset E_{j}$.
    In particular, since $u \in D_j$, when we execute $\dijkstra$ in the $j$-th iteration, we find $P$ so that $\hat{\delta}_j(u, v) = \wt(P)$.

    Now, let $|L(P, j)| > 0$.
    Let $e_0$ be the closest edge to $v$ not in $E_j$.
    Then, $e_0 \in E_l \setminus E_{l + 1}$ for some $l < j$.
    Let $w$ be the endpoint of $e_0$ closest to $v$.
    Since $\deg(w) \geq s_{l}$, there is some $w^* = r(w, D_{l}) \in D_{l}$ such that $\rank_w(w^*, w) < s_{l}$ so that $(w^*, w) \in E_{l + 1} \subset E_{l}$.
    In particular, $\rank_w(w^*, w) < s_{l} < \rank_w(e_0)$.
    If we consider the sub-path $P_{u, w} \subset P$, we have $L(P_{u, w}, l) \subset L(P, j) \cap \set{k_0 + 1, k_0 + 2, \dotsc, l - 1}$.
    Now, applying the inductive hypothesis to the path $P_{u, w} \cup (w, w^*)$,
    \begin{equation*}
        \hat{\delta}_{l}(w^*, u) \leq w(P_{u, w}) + \wt(w^*, w) + 2 \wt(\edgeset_{l})
    \end{equation*}
    where $\edgeset_{l} \subset P_{u, w} \cup (w, w^*)$, $|\edgeset_{l}| \leq |L(P_{u, w} \cup (w, w^*), l)|$ and $\edgeset_{l} \cap E_{l} = \emptyset$.
    Note since $(w, w^*) \in E_l$, $L(P_{u, w}, l) = L(P_{u, w} \cup (w, w^*), l)$.
    
    Now, in the $j$-th iteration, we use the path $(u, w^*), (w^*, w)$ and the remaining sub-path $P_{w, v} \subset E_j$ so that,
    \begin{equation*}
        \hat{\delta}_{j}(w^*, u) \leq \wt(P_{u, w}) + 2 \wt (\edgeset_{l}) + 2 \wt(w^*, w) + \wt(P_{w, v}) \leq \wt(P) + 2 \wt(\edgeset_{l}) + 2 \wt(e_0)
    \end{equation*}

    Note $e_0 \subset E_{l}$ so that $e_0 \notin \edgeset_{l}$ and we can define $\edgeset_{j} = \edgeset_{l} \cup \set{e_0}$.
    Finally, it is easy to check $|\edgeset_{j}| = |\edgeset_{l}| + 1 \leq |L(P_{u, w}, l)| + 1 \leq |L(P, j)|$ as $l \in L(P, j) \setminus L(P_{u, w}, l)$.
\end{proof}

We are now ready to prove Theorem \ref{thm:weighted-additive-apasp}.

\begin{proof}
    (Correctness)
    
    Recall that $\beta \geq 2$.
    We prove the following for all $j \geq k_0 + 3$.
    Let $\hat{\delta}_j$ be the estimate $\hat{\delta}$ after the $j$-th iteration of Phase 2 in Algorithm \ref{alg:k-weighted-additive-apasp}.
    For $u \in D_j, v \in V$ and $P$ a shortest path from $u$ to $v$,
    \begin{equation*}
        \hat{\delta}_j(u, v) \leq \delta(u, v) + 2 \wt(\edgeset_j)
    \end{equation*}
    where $\edgeset_j \subset P$, $\edgeset_j \cap E_{j} = \emptyset$ and $|\edgeset_j| \leq j - (k_0 + 1)$.
    From Lemma \ref{lemma:k0+3-approx-weighted}, we have proved the base case for $j = k_0 + 3$.
    Let $j > k_0 + 3$ and $e = (a, b)$ be the closest edge in $P$ to $v$ not in $E_j$ and $b$ the endpoint closer to $v$.

    Then, $\rank_a(e), \rank_b(e) \geq s_{j - 1}$ and $j \geq k_0 + 3$.
    In particular, $\deg(b) \geq s_{j - 1}$ so there is a vertex $b^* = r(b, D_{j - 1})$ satisfying $\rank_b(b, b^*) < s_{j - 1}$ and therefore $\wt(b, b^*) \leq \wt(e)$ and $(b^*, b) \in E_j$,
    By induction, 
    \begin{equation*}
        \hat{\delta}_{j - 1}(b^*, u) \leq \delta(b^*, u) + 2 \wt(\edgeset_{j - 1})
    \end{equation*}

    Then,
    \begin{equation*}
        \hat{\delta}_{j}(u, v) \leq \hat{\delta}_{j - 1}(b^*, u) + \wt(b^*, b) + \delta(b, v) \leq \delta(u, b) + \delta(b, v) + 2 \wt(\edgeset_{j - 1}) + 2 \wt(b, b^*)
    \end{equation*}

    Thus, define $\edgeset_{j} = \edgeset_{j - 1} \cup \set{(b, b^*)}$ which satisfies the desired conditions since $(b, b^*) \in E_j \subset E_{j - 1}$ and is therefore not in $\edgeset_{j - 1}$.
    Furthermore, with induction, $|\edgeset_{j}| \leq j - (k_0 + 1)$.
    In particular, for $j = k$, since $D_k = V$,
    \begin{equation*}
        \hat{\delta}(u, v) \leq \delta(u, v) + 2 \wt (\edgeset_{k}) \leq \delta(u, v) + 2 \kHeavy(P, \floor{k/3} + 1) = \delta(u, v) + 2 \kHeavy(P, \beta)
    \end{equation*}

    as desired, since $k - (k_0 + 1) = \floor{k/3} + 1$.
\end{proof}

\begin{proof}
    (Time Complexity)

    Phase 0 requires $O(m)$ time, as $\decompose$ requires $O(k(m + n))$-time.

    We now examine Phase 1.
    Each $|C_i| = \bigtO{\frac{n}{t_i}} = \tO{2^i}$
    Each graph $F_i$ has $O(n t_i)$ edges.
    Thus, each $\dijkstra$ requires $\tO{n^2}$-time.
    Over $k$ iterations, this is again $\tO{n^2}$-time.
    
    Now, $|D_{k_0 + 2}| = \bigtO{\frac{n}{s_{k_0 + 2}}} = \bigtO{\frac{n}{s_{k - \floor{k/3}}}} = \bigtO{n^{1 - \frac{\floor{k/3}}{\floor{k/3}+2} x}}$.
    Similarly, we can upper bound the size of $D_{k_0 + 3}$ as $|D_{k_0 + 3}| = \bigtO{\frac{n}{s_{k - (\floor{k/3} - 1)}}} = \bigtO{n^{1 - \frac{\floor{k/3} - 1}{\floor{k/3}+2} x}}$.
    Finally, $|C_i| = \tO{2^l} = \tO{n^{1 - x}}$.
    Thus, each $\approximateMinPlus$ requires time,
    \begin{equation*}
        \frac{1}{\eps} n^{\omega \left( 1 - \frac{\floor{k/3}}{\floor{k/3}+2} x, 1 - x, 1 - \frac{\floor{k/3} - 1}{\floor{k/3}+2} x \right)}
    \end{equation*}

    In Phase 2, each $|D_{j}| = \bigtO{n^{1 - (k - j) \frac{x}{\floor{k/3} + 2}}}$ and $|E_j| \leq \bigO{n^{1 + (k - j + 1) \frac{x}{\floor{k/3} + 2}}}$.
    Therefore, each invocation of $\dijkstra$ in Phase 2 requires time $\bigtO{n^{2 + \frac{x}{\floor{k/3} + 2}}}$.
    Finally, we balance terms to optimize,
    \begin{equation*}
        \omega \left( 1 - \frac{\floor{k/3}}{\floor{k/3}+2} x, 1 - x, 1 - \frac{\floor{k/3} - 1}{\floor{k/3}+2} x \right) = 2 + \frac{x}{\floor{k/3} + 2}
    \end{equation*}
    To get the exact expression, we plug in $k = 3 \beta - 1$ and in particular $\floor{k/3} = \beta - 1$.
\end{proof}

\begin{lemma}
    \label{lemma:k0+3-approx-weighted}
    Let $u \in D_{k_0 + 3}$ and $v \in V$ and $P$ be a shortest path between $u, v$.
    Then, after the $(k_0 + 3)$-th iteration of Phase 2 of Algorithm \ref{alg:k-weighted-additive-apasp},
    \begin{equation*}
        \delta(u, v) \leq \hat{\delta}(u, v) \leq \delta(u, v) + 2 \wt(\edgeset_{k_0 + 3})
    \end{equation*}
    for some $\edgeset_{k_0 + 3} \subset P$ of size at most 2 and $\edgeset_{k_0 + 3} \cap E_{k_0 + 3} = \emptyset$.
\end{lemma}

\begin{proof}
    We proceed by case analysis.
    
    \paragraph{Case 1: $P \subset E_{k_0 + 1}$}

    By Lemma \ref{lemma:weighted-sparse-apasp} and we have $\hat{\delta}(u, v) \leq \delta_{k_0 + 3}(u, v) \leq \delta(u, v) + 2 \wt(\edgeset_{k_0 + 3})$ where $|\edgeset_{k_0 + 3}| \leq |L(P, k_0 + 3)| \leq 2$ and $\edgeset_{k_0 + 3} \cap E_{k_0 + 3} = \emptyset$.
    
    \paragraph{Case 2: $P \not\subset E_{k_0 + 1}$}

    Finally, suppose $P$ contains some edge not in $E_{k_0 + 1}$.
    Let $e_0$ be the closest edge in $P$ to $v$ not in $E_{k_0 + 3}$ and $w_0$ be the endpoint closer to $v$.
    Then, both endpoints of $e_0$ have degree at least $s_{k_0 + 2}$ so that there is $w_0^* = r(w_0, D_{k_0 + 2})$ such that $\rank_w(w_0, w) < s_{k_0 + 2}$ and therefore $\wt(w_0^*, w) \leq \wt(e_0)$.
    
    Note that there is some $1 \leq j \leq l - 1$ such that $P \subset F_{j} \setminus F_{j+1}$ as $t_l = \frac{n}{2^l} \leq n^x = s_{k_0}$.
    Let $e = (a, b) \in P$ such that $e \in F_{j} \setminus F_{j + 1}$.
    Then, $\rank_a(a, b), \rank_b(a, b) \geq t_{j}$.
    Since $\deg(a) > s_j$, there is vertex $a^* = r(a, D_{j})$ such that $\rank_a(a, a^*) < s_j$ and $(a, a^*) \in F_{j + 1} \subset F_{j}$. 
    Since $P \subset F_{j}$, in the $j$-th iteration and $w_0^* \in D_{k_0 + 2}$,
    \begin{align*}
        X_j(a^*, u) &\leq \wt(a^*, a) + \delta(a, u) \leq \delta(a, u) + \wt(a, b) \\
        X_j(a^*, w_0^*) &\leq \wt(a^*, a) + \delta(a, w_0) + \wt(w_0, w_0^*) \leq \delta(a, w_0) + \wt(a, b) + \wt(e_0)
    \end{align*}

    From the $\approximateMinPlus$, we have,
    \begin{equation*}
        \hat{\delta}(u, w_0^*) \leq \left(1 + \frac{\eps}{3}\right) ( \delta(a, u) + \delta(a, w_0) + 2 \wt(a, b) + \wt(e_0)) \leq (1 + \eps) \delta(u, w_0) + 2 \wt(a, b) + \wt(e_0)
    \end{equation*}

    If $e = e_0$, then we automatically have $X_j(w_0^*, u) \leq \delta(w_0, u) + \wt(w_0^*, w_0) \leq \delta(w_0, u) + \wt(e_0)$.
    Thus, we can assume $e \neq e_0$.

    Then, when we execute $\dijkstra$ from $u$ in the $(k_0 + 3)$-th iteration, we use the edges $(u, w_0^*)$, $(w_0^*, w_0)$ and the remaining edges in $E_{k_0 + 3}$, so that,
    \begin{equation*}
        \hat{\delta}(u, v) \leq \hat{\delta}(u, w_0^*) + \wt(e_0) + \delta(w_0, v) \leq (1 + \eps) \delta(u, v) + \leq (1 + \eps) \delta(u, v) + 2 (\wt(a, b) + \wt(e_0))
    \end{equation*}

    We claim $\set{(a, b), e_0}$ satisfies the desired conditions.
    Clearly $\edgeset_{k_0 + 3}$ has size 2.
    By definition, $e_0 \notin E_{k_0 + 3}$.
    Recall $\rank_a(e), \rank_b(e) \geq t_{j} \geq t_{l} \geq \frac{n^x}{4} > s_{k_0 + 2}$ so that $(a, b) \notin E_{k_0 + 3}$.
\end{proof}

\section{\texorpdfstring{$(\alpha, \beta)$}{(alpha, beta)}-Approximate APSP}
\label{sec:sub-2-mult-approx}

In Section \ref{sec:mult-approx-bounded}, we showed the best known algorithm for $(2, 0)$-multiplicative approximation on general graphs, as well as observed that $(\frac{7}{3}, 0)$-approximation can be done in quadratic time, building on a quadratic $(2, 1)$-approximation due to Baswana and Kavitha \cite{baswana2010fasterapasp} and Berman  and Kasiviswanathan \cite{berman2007approxapsp}.
Furthermore, it is well known that any $(2 - \eps, 0)$-approximation is as hard as Boolean Matrix Multiplication for any $\eps > 0$, and therefore requires $\Omega(n^{\omega})$ time.
However, it is possible to obtain $(\alpha, \beta)$ approximations for $\alpha < 2$ if we allow sufficient additive error $\beta$.
We will briefly describe the trade-off in this section.

\subsection{\texorpdfstring{$(\alpha, \beta)$}{(alpha, beta)}-Approximation for Unweighted Graphs}

Since our algorithms on unweighted graphs generally provide additive guarantees, the limiting factors are the pairs of vertices with short distances.
The following result shows the running time required to obtain an $(\alpha, \beta)$ approximation using our previous results for $\alpha < 2$.

\begin{theorem}
    \label{thm:sub-2-mult-add-approx}
    Let $\beta \geq 2$ and $1 < \alpha < 2$.
    Let $T(\gamma)$ denote the running time of Algorithm \ref{alg:bounded-additive-apasp} when given approximation parameter $\gamma$.
    
    Then, Algorithm \ref{alg:alpha-beta-apasp} produces a $(\alpha, \beta)$-approximate $\hat{\delta}$ such that for all $u, v$,
    \begin{equation*}
        \delta(u, v) \leq \hat{\delta}(u, v) \leq \alpha \delta(u, v) + \beta
    \end{equation*}
    in time $\bigtO{T(\gamma)}$ where $\gamma$ is the largest even integer satisfying,
    \begin{equation*}
        \gamma \leq 1 + \frac{\beta - 1}{2 - \alpha}
    \end{equation*}
\end{theorem}

Before discussing the algorithm and proof, we give a high level overview and discuss our parameter choices.

\paragraph{High Level Overview}

Our goal in this section is to use our previous algorithms to obtain good bi-criteria approximations for the APSP Problem.
The idea is that by increasing the tolerance for additive error, we can in fact obtain approximations where the multiplicative factor is less than 2.

The key component of our algorithm is running Algorithm \ref{alg:bounded-additive-apasp}.
We will ensure small error on short paths by invoking $\baswanaAPASP$ and ensure accuracy on long paths with $\additiveAPASP$ (or alternatively $\denseAPASP$) since we can choose an arbitrary constant $k$ such that executing $\additiveAPASP$ or $\denseAPASP$ is not the computational bottleneck.
The constant $C$ is chosen so that any path incurring the larger additive error $k$ is long enough to ensure that the multiplicative error component $\alpha$ is small.

Why do we require $\beta \geq 2$?
Let $\gamma$ be the approximation parameter used in the invocation of Algorithm \ref{alg:bounded-additive-apasp}.
Then for all paths of length $\gamma - 1 \leq \delta(u, v) \leq C$, we obtain a $+\gamma$ additive error.
Therefore, for a fixed $\beta$, the paths of length $\gamma - 1$ have the worst multiplicative error.
Let $P$ be a path of length $\gamma - 1$.
On such paths, we incur an error of $\gamma$ and thus $\hat{\delta}(u, v) \leq 2 \gamma - 1$.
If $\beta = 1$, then we must have $\alpha \geq 2$, whereas we are interested in $\alpha < 2$.
Thus, we require $\beta \geq 2$ in the input of our algorithm.

As a sub-routine, we will require the $(2, 1)$-approximation algorithm of Baswana and Kavitha \cite{baswana2010fasterapasp} and Berman and Kasiviswanathan \cite{berman2007approxapsp}.

\begin{lemma}
    \label{lemma:2-1-approx}
    There is an algorithm $\baswanaAPASP$ that runs in time $\tO{n^2}$ and outputs estimate $\hat{\delta}$ satisfying,
    \begin{equation*}
        \delta(u, v) \leq \hat{\delta}(u, v) \leq 2 \delta(u, v) + 1 
    \end{equation*}
    for all $u, v \in V$.
\end{lemma}

We will briefly remark that Algorithm $\baswanaAPASP$ in fact returns a $(2, 0)$-approximation for paths of even length, only returning a $(2, 1)$-approximation for paths of odd length.

\IncMargin{1em}
\begin{algorithm}[H]
\SetKwInOut{Input}{Input}\SetKwInOut{Output}{Output}
\Input{Unweighted, undirected Graph $G = (V, E)$ with $n$ vertices; approximation parameters $\alpha, \beta$}
\Output{Distance estimate $\hat{\delta}: U \times V \rightarrow \Z$ such that $\delta(u, v) \leq \hat{\delta}(u, v) \leq \alpha \delta(u, v) + \beta$ for all $u, v \in V$}

\BlankLine

$\gamma \gets$ maximum even $\gamma$ such that $\gamma \leq 1 + \frac{\beta - 1}{2 - \alpha}$

$k \gets $ minimum $k$ such that $\additiveAPASP(G, k)$ requires $O(T(\gamma))$ time

$C \gets \frac{k - \beta}{\alpha - 1}$

$\hat{\delta}(u, v) \gets \begin{cases}
    1 & (u, v) \in E \\
    \infty & \otherwise
\end{cases}$

$\hat{\delta} \gets \min(\hat{\delta}, \baswanaAPASP(G))$

$\hat{\delta} \gets \min(\hat{\delta}, \boundedAdditiveAPASP(G, \gamma, C))$

$\hat{\delta} \gets \min(\hat{\delta}, \additiveAPASP(G, k))$

\caption{$\alphabetaAPASP(G, \alpha, \beta)$} 
\label{alg:alpha-beta-apasp}
\end{algorithm}
\DecMargin{1em}

\begin{proof}
    (Correctness)
    Let $u, v \in V$.

    \paragraph{Case 1: $\delta(u, v) \geq C$}

    Then, from $\additiveAPASP$,
    \begin{equation*}
        \hat{\delta}(u, v) \leq \delta(u, v) + k = (\alpha - 1) C + \beta \leq \alpha \delta(u, v) + \beta
    \end{equation*}

    \paragraph{Case 2: $\gamma \leq \delta(u, v) \leq C$}

    Then, from $\boundedAdditiveAPASP(G, \gamma, C)$,
    \begin{align*}
        \hat{\delta}(u, v) &\leq \delta(u, v) + \gamma
    \end{align*}

    Note that if $\beta \geq \gamma$, we have in fact already obtained a $(1, \beta)$ approximation. 
    Thus, in the following, assume $\beta < \gamma$.
    
    For now, assume $\delta(u, v) = \gamma$.
    Then,
    \begin{align*}
        \hat{\delta}(u, v) &\leq 2 \delta(u, v) \\
        &= (2 \delta(u, v) - \beta) + \beta \\
        &= \frac{2 \delta(u, v) - \beta}{\delta(u, v)} \delta(u, v) + \beta \\
        &= \left( 2 - \frac{\beta}{\gamma} \right) \delta(u, v) + \beta
    \end{align*}

    noting that $2 - \frac{\beta}{\gamma} \leq 2 - \frac{\beta - 1}{\gamma - 1} \leq \alpha$ as $\beta < \gamma$.
    We conclude the proof by noting that for all larger paths, an additive error of $\gamma$ implies an $(\alpha, \beta)$ approximation, as desired.

    \paragraph{Case 3: $\delta(u, v) \leq \gamma - 1$}
    
    From $\baswanaAPASP$,
    \begin{equation*}
        \hat{\delta}(u, v) \leq 2 \delta(u, v) + 1 = 2 \delta(u, v) - (\beta - 1) + \beta \leq \left( 2 - \frac{\beta - 1}{\gamma - 1} \right) \delta(u, v) + \beta
    \end{equation*}

    By our choice of $\gamma$,
    \begin{align*}
        2 - \frac{\beta - 1}{\gamma - 1} \leq \alpha
    \end{align*}
\end{proof}

\begin{proof}
    (Time Complexity)
    Throughout, we use the fact that $\gamma, \beta$ are constants.
    Then, we chose a constant $k$ to satisfy $T'(k) = O(T(\gamma))$ where $T'(k)$ is the running time of $\additiveAPASP$ when required to produce additive error $k$.
    Since $C$ is a function of constants, $C$ is a constant as well.
    In particular, $\boundedAdditiveAPASP$ requires $O(T(\gamma))$ time, $\baswanaAPASP$ requires $\tO{n^2} = O(T(\gamma))$ time, and $\additiveAPASP$ requires $O(T(\gamma))$ time, proving the desired time complexity.
\end{proof}

We give some examples for $\gamma = 6$.

\begin{corollary}
    In time $T(6) = n^{2.06382}$ (Corollary \ref{cor:bounded-additive-examples}), we can obtain any of the following approximations.

    \begin{equation*}
        \left(\frac{9}{5}, 2\right), \left(\frac{8}{5}, 3\right), \left(\frac{7}{5}, 4\right), \left(\frac{6}{5}, 5\right)
    \end{equation*}
\end{corollary}

\subsection{Extension to Weighted Graphs}

In this section, we give the simple observation that a good additive approximation can imply a multiplicative approximation even for weighted graphs, if given sufficient additional additive error.

\begin{theorem}
    \label{thm:mult-add-weighted-approx}
    Let $1 \leq \alpha < 3$ and $\beta \geq 1$ be an integer.
    Let $T(\gamma)$ be the running time of Theorem \ref{thm:weighted-additive-apasp} when asked to produce a $(1 + \eps, 2 w_{u, v}(\gamma))$ approximation.
    Let $\gamma = \ceil{\frac{\alpha (\beta + 1) + (\beta - 1)}{2}} - 1$.
    Then, Algorithm \ref{alg:k-weighted-additive-apasp} produces an $(\alpha, 2 w_{u, v}(\beta))$-approximation in time $T(\gamma)$.
\end{theorem}

Suppose for some integer $\gamma \geq 1$ we have a $(1+\eps, 2 w_{u, v}(\gamma))$ approximation from Theorem \ref{thm:weighted-additive-apasp}.
Given an integer value $\beta < \gamma$, suppose we want a $(\alpha, 2 w_{u, v}(\beta))$ approximation for APSP on weighted graphs.
Consider the edges that are the $\beta + 1, \beta + 2, \dotsc, \gamma$-th heaviest edges in a given shortest path $P$.
Then,
\begin{equation*}
    \sum_{i = \beta + 1}^{\gamma} \heavy(P, i) \leq (\gamma - \beta) \heavy(P, \beta + 1) \leq \frac{\gamma - \beta}{\beta + 1} \wt(P)
\end{equation*}

as the weight of the $\beta + 1$-heaviest edge is at most a $\frac{1}{\beta + 1}$ fraction of the path's total weight.
In particular, if $\hat{\delta}$ is the distance estimate given by Theorem \ref{thm:weighted-additive-apasp},
\begin{align*}
    \hat{\delta}(u, v) &\leq (1 + \eps) \wt(P) + 2 \kHeavy(P, \gamma) \\
    &\leq (1 + \eps) \wt(P) + 2 \frac{\gamma - \beta}{\beta + 1} \wt(P) + 2 \kHeavy(P, \beta) \\
    &\leq \left( \frac{2 \gamma - \beta + 1}{\beta + 1} + \eps \right) \delta(u, v) + 2 \kHeavy(P, \beta)
\end{align*}

Then, let us express $\gamma$ in terms of $\alpha, \beta$.
In particular, we require $\gamma$ such that,

\begin{align*}
    \alpha &\geq \frac{2 \gamma - \beta + 1}{\beta + 1} + \eps \\
    \gamma &\leq \frac{(\alpha - \eps)(\beta + 1) + (\beta - 1)}{2}
\end{align*}

Since $\gamma$ is an integer, we conclude that to give any $(\alpha, 2 w_{u, v}(\beta))$ approximation, it suffices to find a $(1 + \eps, 2 w_{u, v}(\gamma))$ approximation for $\gamma = \floor{\frac{(\alpha - \eps)(\beta + 1) + (\beta - 1)}{2}}$.
In particular, since we are free to choose $\eps$, we can define,
\begin{equation*}
    \gamma = \begin{cases}
        \frac{\alpha (\beta + 1) + (\beta - 1)}{2} - 1 & \frac{\alpha (\beta + 1) + (\beta - 1)}{2} \in \Z \\
        \floor{\frac{\alpha (\beta + 1) + (\beta - 1)}{2}} & \otherwise
    \end{cases}
\end{equation*}

or equivalently, 
\begin{equation*}
    \gamma = \ceil{\frac{\alpha (\beta + 1) + (\beta - 1)}{2}} - 1
\end{equation*}

As an example, suppose we want a $\left( \frac{4}{3} + \eps, 2 w_{u, v}(2) \right)$ approximation.
Then, define $\gamma = 3$ and apply Theorem \ref{thm:weighted-additive-apasp}.

\newpage
\bibliographystyle{alpha}
\bibliography{references}

\newcommand{\etalchar}[1]{$^{#1}$}
\begin{thebibliography}{ENWN16}

\bibitem[ACIM99]{aingworth1999fast}
Donald Aingworth, Chandra Chekuri, Piotr Indyk, and Rajeev Motwani.
\newblock Fast estimation of diameter and shortest paths (without matrix
  multiplication).
\newblock {\em SIAM Journal on Computing}, 28(4):1167--1181, 1999.

\bibitem[AG13]{agarwal2013stretch2oracle}
Rachit Agarwal and Philip~Brighten Godfrey.
\newblock Brief announcement: a simple stretch 2 distance oracle.
\newblock In {\em Proceedings of the 2013 ACM symposium on Principles of
  distributed computing}, pages 110--112, 2013.

\bibitem[AGM97]{alon1997exponent}
Noga Alon, Zvi Galil, and Oded Margalit.
\newblock On the exponent of the all pairs shortest path problem.
\newblock {\em Journal of Computer and System Sciences}, 54(2):255--262, 1997.

\bibitem[AR20]{akav2020almost2}
Maor Akav and Liam Roditty.
\newblock An almost 2-approximation for all-pairs of shortest paths in
  subquadratic time.
\newblock In {\em Proceedings of the Fourteenth Annual ACM-SIAM Symposium on
  Discrete Algorithms}, pages 1--11. SIAM, 2020.

\bibitem[BGS05]{baswana2005nearly2approxapasp}
Surender Baswana, Vishrut Goyal, and Sandeep Sen.
\newblock All-pairs nearly 2-approximate shortest-paths in o (n 2 polylog n)
  time.
\newblock In {\em STACS 2005: 22nd Annual Symposium on Theoretical Aspects of
  Computer Science, Stuttgart, Germany, February 24-26, 2005. Proceedings 22},
  pages 666--679. Springer, 2005.

\bibitem[BGSW19]{bringmann2019boundeddifference}
Karl Bringmann, Fabrizio Grandoni, Barna Saha, and Virginia~Vassilevska
  Williams.
\newblock Truly subcubic algorithms for language edit distance and rna folding
  via fast bounded-difference min-plus product.
\newblock {\em SIAM Journal on Computing}, 48(2):481--512, 2019.

\bibitem[BK07]{berman2007approxapsp}
Piotr Berman and Shiva~Prasad Kasiviswanathan.
\newblock Faster approximation of distances in graphs.
\newblock In {\em Algorithms and Data Structures: 10th International Workshop,
  WADS 2007, Halifax, Canada, August 15-17, 2007. Proceedings 10}, pages
  541--552. Springer, 2007.

\bibitem[BK10]{baswana2010fasterapasp}
Surender Baswana and Telikepalli Kavitha.
\newblock Faster algorithms for all-pairs approximate shortest paths in
  undirected graphs.
\newblock {\em SIAM Journal on Computing}, 39(7):2865--2896, 2010.

\bibitem[BKW19]{bringmann2019noscaling}
Karl Bringmann, Marvin K{\"u}nnemann, and Karol W{\k{e}}grzycki.
\newblock Approximating apsp without scaling: equivalence of approximate
  min-plus and exact min-max.
\newblock In {\em Proceedings of the 51st Annual ACM SIGACT Symposium on Theory
  of Computing}, pages 943--954, 2019.

\bibitem[Bra]{Complexity}
Jan van~den Brand.
\newblock Complexity term balancer.
\newblock \url{www.ocf.berkeley.edu/~vdbrand/complexity/}.
\newblock Tool to balance complexity terms depending on fast matrix
  multiplication.

\bibitem[CDX22]{chi2022boundeddifference}
Shucheng Chi, Ran Duan, and Tianle Xie.
\newblock Faster algorithms for bounded-difference min-plus product.
\newblock In {\em Proceedings of the 2022 Annual ACM-SIAM Symposium on Discrete
  Algorithms (SODA)}, pages 1435--1447. SIAM, 2022.

\bibitem[CDXZ22]{chi2022monotone}
Shucheng Chi, Ran Duan, Tianle Xie, and Tianyi Zhang.
\newblock Faster min-plus product for monotone instances.
\newblock In {\em Proceedings of the 54th Annual ACM SIGACT Symposium on Theory
  of Computing}, pages 1529--1542, 2022.

\bibitem[Cha21]{chan2021all}
Timothy~M Chan.
\newblock All-pairs shortest paths for real-weighted undirected graphs with
  small additive error.
\newblock In {\em 29th Annual European Symposium on Algorithms (ESA 2021)}.
  Schloss Dagstuhl-Leibniz-Zentrum f{\"u}r Informatik, 2021.

\bibitem[Che14]{chechik2014approximate}
Shiri Chechik.
\newblock Approximate distance oracles with constant query time.
\newblock In {\em Proceedings of the forty-sixth annual ACM symposium on Theory
  of computing}, pages 654--663, 2014.

\bibitem[Che15]{chechik2015oracles}
Shiri Chechik.
\newblock Approximate distance oracles with improved bounds.
\newblock In {\em Proceedings of the forty-seventh annual ACM symposium on
  Theory of Computing}, pages 1--10, 2015.

\bibitem[CZ01]{cohen2001smallstretch}
Edith Cohen and Uri Zwick.
\newblock All-pairs small-stretch paths.
\newblock {\em Journal of Algorithms}, 38(2):335--353, 2001.

\bibitem[CZ22]{chechik2022nearly2}
Shiri Chechik and Tianyi Zhang.
\newblock Nearly 2-approximate distance oracles in subquadratic time.
\newblock In {\em Proceedings of the 2022 Annual ACM-SIAM Symposium on Discrete
  Algorithms (SODA)}, pages 551--580. SIAM, 2022.

\bibitem[DFK{\etalchar{+}}23]{DBLP:journals/corr/DoryFKNWV}
Michal Dory, Sebastian Forster, Yael Kirkpatrick, Yasamin Nazari,
  Virginia~Vassilevska Williams, and Tijn de~Vos.
\newblock Fast 2-approximate all-pairs shortest paths.
\newblock {\em CoRR}, abs/2307.09258, 2023.

\bibitem[DHZ00]{dor2000apasp}
Dorit Dor, Shay Halperin, and Uri Zwick.
\newblock All-pairs almost shortest paths.
\newblock {\em SIAM Journal on Computing}, 29(5):1740--1759, 2000.

\bibitem[DKR{\etalchar{+}}22]{deng2022apasp}
Mingyang Deng, Yael Kirkpatrick, Victor Rong, Virginia Vassilevska~Williams,
  and Ziqian Zhong.
\newblock New additive approximations for shortest paths and cycles.
\newblock In {\em 49th International Colloquium on Automata, Languages, and
  Programming (ICALP 2022)}, volume 229, 2022.

\bibitem[D{\"u}r23]{durr2023rect_monotone}
Anita D{\"u}rr.
\newblock Improved bounds for rectangular monotone min-plus product and
  applications.
\newblock {\em Information Processing Letters}, page 106358, 2023.

\bibitem[DWZ22]{duan2022matrixmult}
Ran Duan, Hongxun Wu, and Renfei Zhou.
\newblock Faster matrix multiplication via asymmetric hashing.
\newblock {\em arXiv preprint arXiv:2210.10173}, 2022.

\bibitem[Elk05]{elkin2005asp}
Michael Elkin.
\newblock Computing almost shortest paths.
\newblock {\em ACM Transactions on Algorithms (TALG)}, 1(2):283--323, 2005.

\bibitem[ENWN16]{elkin2016spaceoracle}
Michael Elkin, Ofer Neiman, and Christian Wulff-Nilsen.
\newblock Space-efficient path-reporting approximate distance oracles.
\newblock {\em Theoretical Computer Science}, 651:1--10, 2016.

\bibitem[EP16]{elkin2016linearoracle}
Michael Elkin and Seth Pettie.
\newblock A linear-size logarithmic stretch path-reporting distance oracle for
  general graphs.
\newblock {\em ACM Transactions on Algorithms (TALG)}, 12(4):1--31, 2016.

\bibitem[Gal12]{le2012rectangular}
Fran{\c{c}}ois~Le Gall.
\newblock Faster algorithms for rectangular matrix multiplication.
\newblock In {\em 2012 IEEE 53rd annual symposium on foundations of computer
  science}, pages 514--523. IEEE, 2012.

\bibitem[GU18]{gall2018rectangularmm}
Fran{\c{c}}ois~Le Gall and Florent Urrutia.
\newblock Improved rectangular matrix multiplication using powers of the
  coppersmith-winograd tensor.
\newblock In {\em Proceedings of the Twenty-Ninth Annual ACM-SIAM Symposium on
  Discrete Algorithms}, pages 1029--1046. SIAM, 2018.

\bibitem[HP98]{huang1998pan}
Xiaohan Huang and Victor~Y Pan.
\newblock Fast rectangular matrix multiplication and applications.
\newblock {\em Journal of complexity}, 14(2):257--299, 1998.

\bibitem[LR83]{lotti1983rectangularmm}
Grazia Lotti and Francesco Romani.
\newblock On the asymptotic complexity of rectangular matrix multiplication.
\newblock {\em Theoretical Computer Science}, 23(2):171--185, 1983.

\bibitem[MN07]{mendel2007ramsey}
Manor Mendel and Assaf Naor.
\newblock Ramsey partitions and proximity data structures.
\newblock {\em Journal of the European Mathematical Society}, 9(2):253--275,
  2007.

\bibitem[PR10]{patrascu2010distanceoracles}
Mihai Patrascu and Liam Roditty.
\newblock Distance oracles beyond the thorup-zwick bound.
\newblock In {\em 2010 IEEE 51st Annual Symposium on Foundations of Computer
  Science}, pages 815--823. IEEE, 2010.

\bibitem[Rod23]{roditty2023newapasp}
Liam Roditty.
\newblock New algorithms for all pairs approximate shortest paths.
\newblock In {\em Proceedings of the 55th Annual ACM Symposium on Theory of
  Computing}, pages 309--320, 2023.

\bibitem[Sei95]{seidel1995all}
Raimund Seidel.
\newblock On the all-pairs-shortest-path problem in unweighted undirected
  graphs.
\newblock {\em Journal of computer and system sciences}, 51(3):400--403, 1995.

\bibitem[Som16]{sommer2016distanceoracle}
Christian Sommer.
\newblock All-pairs approximate shortest paths and distance oracle
  preprocessing.
\newblock In {\em 43rd International Colloquium on Automata, Languages, and
  Programming (ICALP 2016)}. Schloss Dagstuhl-Leibniz-Zentrum fuer Informatik,
  2016.

\bibitem[SZ99]{shoshan_zwick}
A.~Shoshan and U.~Zwick.
\newblock All pairs shortest paths in undirected graphs with integer weights.
\newblock In {\em 40th Annual Symposium on Foundations of Computer Science
  (Cat. No.99CB37039)}, pages 605--614, 1999.

\bibitem[TZ05]{thorup2005approximate}
Mikkel Thorup and Uri Zwick.
\newblock Approximate distance oracles.
\newblock {\em Journal of the ACM (JACM)}, 52(1):1--24, 2005.

\bibitem[Wil14]{williams2014faster}
Ryan Williams.
\newblock Faster all-pairs shortest paths via circuit complexity.
\newblock In {\em Proceedings of the forty-sixth annual ACM symposium on Theory
  of computing}, pages 664--673, 2014.

\bibitem[WN12]{wulff2012approximateoracles}
Christian Wulff-Nilsen.
\newblock Approximate distance oracles with improved preprocessing time.
\newblock In {\em Proceedings of the twenty-third annual ACM-SIAM symposium on
  Discrete Algorithms}, pages 202--208. SIAM, 2012.

\bibitem[WW10]{williams2010subcubic}
Virginia~Vassilevska Williams and Ryan Williams.
\newblock Subcubic equivalences between path, matrix and triangle problems.
\newblock In {\em 2010 IEEE 51st Annual Symposium on Foundations of Computer
  Science}, pages 645--654. IEEE, 2010.

\bibitem[WX20]{williams2020truly}
Virginia~Vassilevska Williams and Yinzhan Xu.
\newblock Truly subcubic min-plus product for less structured matrices, with
  applications.
\newblock In {\em Proceedings of the Fourteenth Annual ACM-SIAM Symposium on
  Discrete Algorithms}, pages 12--29. SIAM, 2020.

\bibitem[Yus12]{yuster2012approximate}
Raphael Yuster.
\newblock Approximate shortest paths in weighted graphs.
\newblock {\em Journal of Computer and System Sciences}, 78(2):632--637, 2012.

\bibitem[Zwi98]{zwick_apsp_weighted_directed}
Uri Zwick.
\newblock All pairs shortest paths in weighted directed graphs-exact and almost
  exact algorithms.
\newblock In {\em Proceedings 39th Annual Symposium on Foundations of Computer
  Science (Cat. No. 98CB36280)}, pages 310--319. IEEE, 1998.

\bibitem[Zwi02]{zwick2002bridgingsets}
Uri Zwick.
\newblock All pairs shortest paths using bridging sets and rectangular matrix
  multiplication.
\newblock {\em Journal of the ACM (JACM)}, 49(3):289--317, 2002.

\end{thebibliography}

\appendix
\section{Additive Approximation Algorithms of \texorpdfstring{\cite{dor2000apasp}}{[DHZ00]} with a Different Analysis}

In this appendix, we state the approximation algorithms of Dor, Halperin and Zwick \cite{dor2000apasp}, proving some useful properties about the algorithms that we require in our algorithm design.

\dhzapasp*

The pseudocode for $\denseAPASP$ and $\sparseAPASP$ are provided in Algorithms \ref{alg:dhz-dense-apasp} and \ref{alg:dhz-sparse-apasp} respectively.
Although we never explicitly call $\sparseAPASP$, we implicitly run it as a sub-routine in many of our algorithms.

\IncMargin{1em}
\begin{algorithm}[H]
\SetKwInOut{Input}{Input}\SetKwInOut{Output}{Output}
\Input{Unweighted, undirected Graph $G = (V, E)$ with $n$ vertices and even integer $\beta$}
\Output{$\hat{\delta}: U \times V \rightarrow \Z$ such that $\delta(u, v) \leq \hat{\delta}(u, v) \leq \delta(u, v) + \beta$ for all $u, v \in V$}
\BlankLine
$\hat{\delta} \gets \begin{cases}
        1 & (u, v) \in E \\
        \infty & \otherwise
    \end{cases}$
    
$k \gets \frac{3 \beta - 2}{2}$

$s_i \gets n^{1 - \frac{i}{k}}$ for all $1 \leq i \leq k - 1$

$(D_{1}, D_{2}, \dotsc, D_k), (E_{1}, E_{2}, \dotsc, E_k), E^* \gets \decompose(G, (s_{1}, s_{2}, \dotsc, s_{k - 1}))$

\For{$1 \leq i \leq k$}{
    \For{$w \in D_i$}{
        $G_{i, w} \gets \left(V, E_i \cup (u \times V) \cup E^* \cup \left(\bigcup_{i + j_1 + j_2 \leq 2k + 1} D_{j_1} \times D_{j_2} \right)\right)$
        
        $\hat{\delta} \gets \dijkstra(G_{i, w}, w, \hat{\delta})$
    }
}

\caption{$\denseAPASP(G, \beta)$}
\label{alg:dhz-dense-apasp}
\end{algorithm}
\DecMargin{1em}

\IncMargin{1em}
\begin{algorithm}[H]
\SetKwInOut{Input}{Input}\SetKwInOut{Output}{Output}
\Input{Unweighted, undirected Graph $G = (V, E)$ with $n$ vertices, $m$ edges, and even integer $\beta$}
\Output{$\hat{\delta}: U \times V \rightarrow \Z$ such that $\delta(u, v) \leq \hat{\delta}(u, v) \leq \delta(u, v) + \beta$ for all $u, v \in V$}
\BlankLine
$\hat{\delta} \gets \begin{cases}
        1 & (u, v) \in E \\
        \infty & \otherwise
    \end{cases}$
    
$k \gets \frac{\beta + 2}{2}$

$s_i \gets (m/n)^{1 - \frac{i}{k}}$ for all $1 \leq i \leq k - 1$

$(D_{1}, D_{2}, \dotsc, D_k), (E_{1}, E_{2}, \dotsc, E_k), E^* \gets \decompose(G, (s_{1}, s_{2}, \dotsc, s_{k - 1}))$

\For{$1 \leq i \leq k$}{
    \For{$w \in D_i$}{
        $\hat{\delta} \gets \dijkstra(G_{i, w}, w, \hat{\delta})$ where $G_{i, w} \gets \left(V, E_i \cup (u \times V) \cup E^* \right)$
    }
}

\caption{$\sparseAPASP(G, \beta)$}
\label{alg:dhz-sparse-apasp}
\end{algorithm}
\DecMargin{1em}

\subsection{Closer Analysis of \texorpdfstring{$\sparseAPASP$}{SparseAPASP}}
\label{sec:sparse-apasp-analysis}

Instead of bounding the approximation error simply by using the number of thresholds used in the degree decomposition, we perform a more fine-grained analysis to bound the approximation error based on the structure and degrees of a given path $P$.
Recall the definition of blocking vertices and blocking levels.

\blockingvertices*

We begin with some simple properties to note about the blocking vertices and blocking levels of a given path $P$.

\begin{lemma}
    \label{lemma:prop-blocking}
    Let $P$ be a shortest path and $B(P) = \set{x_0, x_1, \dotsc, x_t}$ be the set of blocking vertices. 
    Then,
    \begin{enumerate}
        \item $k \geq \level(x_1) > \level(x_2) > \dotsc > \level(x_t) \geq 1$.
        \item For all $2 \leq i \leq t$, $x_{i - 1} \notin P_{x_i, x_{i - 2}}$.
        \item $B(P_{x_i, x_{i + 1}}) \subset B(P)$ and $L_B(P_{x_i, x_{i + 1}}) \subset L_B(P)$ for $1 \leq i \leq t - 1$.
        \item $P_{x_{t - 1}, x_t} \subset P_{x_{t - 2}, x_{t - 1}} \subset \dotsc \subset P_{x_0, x_1} = P$
    \end{enumerate}
    
\end{lemma}

\begin{proof}
    Property 1 follows as for any edge $(a, b)$, $\level(a, b) = \min(\level(a), \level(b))$.

    Property 2 follows as for all $i \geq 2$, $x_i \in P_{x_{i - 1}, x_{i - 2}}$.

    $B(P_{x_i, x_{i + 1}}) \subset B(P)$ follows by the construction of the blocking vertices, noting that $x_i, x_{i + 1}$ determine the remaining blocking vertices.
    Then, $L_B(P_{x_i, x_{i + 1}}) \subset L_B(P)$ follows as the blocking levels are constructed from the blocking vertices.

    Property 4 follows as $x_i$ is chosen as a vertex on the path $P_{x_{i - 2}, x_{i - 1}}$.
\end{proof}

We now give the deferred proof that the number of blocking levels directly determines the additive error accumulated on a given path when executing $\sparseAPASP$ (\Cref{alg:dhz-sparse-apasp}).

\begin{restatable}{lemma}{blockapproxerror}
    \label{lemma:block-approx-error}
    Let $u, v$ be vertices in a graph $G$ and $P$ be a shortest path between $u, v$.
    Let $B(P)$ be the blocking vertices of $P$.
    Then, Algorithm \ref{alg:dhz-sparse-apasp} in the $\level(v)$-th iteration computes distance estimate $\hat{\delta}$ such that,
    \begin{equation*}
        \hat{\delta}(v^*, u) \leq \hat{\delta}_{\level(v)}(v^*, u) \leq \delta(v, u) + 2 |L_B(P)| + 1
    \end{equation*}

    For any level $j \geq \level(v)$, if $v \in D_j$,
    \begin{equation*}
        \hat{\delta}(v, u) \leq \hat{\delta}_{j}(v, u) \leq \delta(v, u) + 2 |L_B(P, j)|
    \end{equation*}
\end{restatable}

\begin{proof}
    We proceed by induction.
    Suppose $|L_B(P)| = 0$.
    Then, $B(P) = \set{x_0, x_1} = \set{u, v}$.
    In particular, $P \subset E_{\level(v)}$ so that Algorithm \ref{alg:dhz-sparse-apasp} finds the exact distance in the $\level(v)$-th iteration.

    Therefore, assume $|L_B(P)| > 0$ and denote $B(P) = \set{x_0, x_1, \dotsc, x_t}$.
    Since $\level(v) \leq \level(u)$, $x_2 \neq u$.
    Consider the sub-path $P_{v, x_2} = P_{x_1, x_2}$. 
    Then, the set $L_B(P_{v, x_2}) \subset L_B(P)$ by Lemma \ref{lemma:prop-blocking} and $\level(x_2) \in L_B(P) \setminus L_B(P_{v, x_2})$ so that $L_B(P)$ is strictly larger.
    By the inductive hypothesis, in the $\level(v)$-th iteration of Algorithm \ref{alg:dhz-sparse-apasp},
    \begin{align*}
        \hat{\delta}_{\level(v)}(v^*, u) &\leq (\hat{\delta}_{\level(x_2)}(x_2^*, v) + 1) + \delta(x_2, u) + 1 \\
        &\leq \delta(x_2, v) + 2 |L_B(P_{v, x_2})| + \delta(x_2, u) + 3 \\
        &\leq \delta(u, v) + 2 |L_B(P)| + 1
    \end{align*}

    When $v \in D_{\level(v)}$,
    \begin{align*}
        \hat{\delta}(v, u) &\leq \hat{\delta}_{\level(x_2)}(x_2^*, v) + 1 + \delta(x_2, u) \\
        &\leq \delta(x_2, v) + 2 |L_B(P_{x_1, x_2})| + \delta(x_2, u) + 2 \\
        &\leq \delta(v, u) + 2 |L_B(P)|
    \end{align*}

    Now, let $j \geq \level(v)$ and $v \in D_j$. 
    If $|L_B(P, j)| = 0$, then $P \subset E_j$ and we compute an exact distance in the $j$-th iteration.
    Otherwise, let $b = b(v, P)$ and $b^* \in D_{\level(b)}$ where $\level(b) < j$ so that,
    \begin{align*}
        \hat{\delta}_j(v, u) &\leq \hat{\delta}_{\level(b)}(v, b^*) + 1 + \delta(b, u) \\
        &\leq \delta(v, b) + 2 |L_B(P_{b, v})| + 2 + \delta(b, u) \\
        &\leq \delta(u, v) + 2 |L_B(P, j)|
    \end{align*}

    where we have used the inductive hypothesis and $L_B(P_{b, v}) \subset L_B(P, j)$ and $\level(b) \in L_B(P, j) \setminus L_B(P_{b, v})$.
\end{proof}

In many cases for our algorithms, it suffices to use the following weaker lemma.
Here, we bound the additive error simply by the number of levels occurring in the shortest path, whereas above the notion of blocking vertices takes into account the specific positions where each level occurs.

\begin{restatable}{lemma}{sparseapaspapprox}
    \label{lemma:sparse-apasp-approx}
    Let $u, v$ be vertices in a graph $G$ and $P$ be a shortest path between $u, v$.
    Suppose $\level(v) \leq \level(u)$.

    For $j \geq \level(v)$, let $L(P, j) = \set{\level(w) \given w \in P \andT \level(w) < j}$ and let $L(P) = L(P, \level(v)) = \set{\level(w) \given w \in P \andT \level(w) < \level(v)}$.
    Then, Algorithm \ref{alg:dhz-sparse-apasp} in the $\level(v)$-th iteration computes distance estimate $\hat{\delta}$ such that,
    \begin{equation*}
        \delta(v^*, u) \leq \hat{\delta}_{\level(v)}(v^*, u) \leq \delta(v, u) + 2 |L(P)| + 1
    \end{equation*}

    Furthermore, for $j \geq \level(v)$ such that $v \in D_j$,
    \begin{equation*}
        \hat{\delta}_{j}(v, u) \leq \delta(v, u) + 2 |L(P, j)|
    \end{equation*}
\end{restatable}

\begin{proof}
    This follows from Lemma \ref{lemma:block-approx-error} as $L_B(P, j) \subset L(P, j)$ for all paths $P$ and iterations $j$.
\end{proof}

\section{Dominating Sets and Degree Decomposition}
\label{app:degree-decomposition}

In this section, we give the proofs of Lemmas \ref{lemma:hitting-set}, \ref{lemma:dominating-set}, \ref{lemma:degree-decomposition} and \ref{lemma:weight-decomposition}.

\hittingsetlemma*

\IncMargin{1em}
\begin{algorithm}[H]
\SetKwInOut{Input}{Input}\SetKwInOut{Output}{Output}
\Input{Universe $U$ of $n$ elements and a collection $\family$ of $n$ subsets of size at least $s$}
\Output{Hitting set $X$ of size $\bigO{n \log n/s}$}

\BlankLine

$X \gets \emptyset$

$c(j) \gets |\set{S_i \in \family \given j \in S_i}|$ for all $j \in U$

\While{$\family \neq \emptyset$}{
    $v \gets \arg \max_{j \in U} c(j)$
    
    $X \gets X \cup \set{v}$

    $\family \gets \set{S_i \in \family \given v \notin S_i}$
    
    $c(j) \gets |\set{S_i \in \family \given j \in S_i}|$ for all $j \in U$ 
}

\caption{$\hittingSet(U, \family, s)$} 
\label{alg:hitting-set-greedy}
\end{algorithm}
\DecMargin{1em}

\begin{proof}
    We begin with the deterministic algorithm $\hittingSet$ with pseudocode provided in \Cref{alg:hitting-set-greedy}.
    Without loss of generality, assume all $|S_i| = s$.
    Otherwise, drop all but the first $s$ elements of each subset $S_i$ and observe that any hitting set of the resulting collection must also be a hitting set for the original collection.

    Now, initially we have $\sum_{j \in U} c(j) = \sum_{S_i \in \family} |S_i| = n s$.
    When the algorithm terminates, we have an empty collection $\family$ so that $\sum_{j \in U} c(j) = 0$.
    In each iteration, the most expensive operation is to update $\family$ and $c(j)$.
    By keeping $\set{S_i \in \family \given j \in S_i}$ in addition to the count $c(j)$ alone, we can update $\family$ and $c(j)$ in $\tO{1}$ time for each decrement of $c(j)$.
    Thus, the overall running time can be bounded by $\tO{ns}$.

    We now bound the size of output set $X$.
    Let $T_j$ denote the number of sets remaining in $\family$ after $j$ elements have been added to $X$.
    Then, $T_0 = |\family| = n$, $|X| = \min_{j \geq 0 \given T_j = 0} j$, and $T_j = T_{j - 1} - c(u_j)$ where $u_j$ is the $j$-th element added to $X$.
    By our choice of $c(u_j)$, we have $c(u_j) \geq \frac{T_{j-1} s}{n - j + 1}$ since there are $T_{j - 1}$ sets of size $s$ and $n - (j - 1) = n - j + 1$ elements not in $X$.
    Then,
    \begin{equation*}
        T_j \leq \left( 1 - \frac{s}{n - j + 1} \right) T_{j - 1} \leq T_0 \prod_{l = 0}^{j - 1} \left( 1 - \frac{s}{n - l} \right) < n \left( 1 - \frac{s}{n} \right)^{j} \leq n e^{- \frac{s j}{n}}
    \end{equation*}

    When $j = \frac{n \log n}{s}$ we have $T_j < 1$ so $T_j = 0$.
    In particular, $|X| \leq \frac{n \log n}{s}$.
\end{proof}

\begin{proof}
    The randomized algorithm $\rHittingSet$ simply samples each element with probability $p$.
    Clearly this algorithm runs in time $O(n)$.

    First, we bound the probability that $X$ is not a hitting set.
    For each subset of size $s$, the probability that no element is sampled can be bounded as,
    \begin{equation*}
        \left( 1 - p \right)^{s} \leq e^{- p s}
    \end{equation*}
    By the union bound, the probability some subset has no element sampled can be bounded as,
    \begin{equation*}
        n e^{- p s}
    \end{equation*}
    so if $p = c \log n/s$, we upper bound this as $n^{1 - c}$ for some constant $c > 1$.

    Then, to upper bound the size of $X$, we apply a standard Chernoff bound to see that $|X| = \bigO{n \log n/s}$ with high probability.
\end{proof}

\dominatingsetlemma*

\begin{proof}
    Given the previous result, we simply apply $\hittingSet$ or $\rHittingSet$ on the subsets $N(v) \cup \set{v}$ for all $v \in V_s$.
    Since $|V_s| \leq n$, note that our previous result is also correct for all collections $\family$ of at most $n$ subsets as we can augment $\family$ with $n - |\family|$ subsets and any hitting set of the augmented collection is also a hitting set for the original collection.
\end{proof}

\degreedecompositionlemma*

We first discuss the deterministic algorithm $\decompose$.

\begin{proof}
    By Lemma \ref{lemma:dominating-set}, we can compute each $D_i$ with $\dominate$ in $\tO{n s_i} = \tO{n^2}$ time, so that overall we can compute all $D_i$ in $\tO{k n^2}$ time.
    Furthermore, we can ensure $D_1 \subset D_2 \subset \dotsc \subset D_k = V$ by augmenting the $D_i$ as necessary.
    Of course, when constructing $D_i$, we can easily construct $E_i^*$ by choosing the appropriate edges when we add a vertex to $v$ to $D_i$.
    Finally, we can construct $E_i$ in $O(m)$ time by examining each edge for the degrees of its endpoints.
\end{proof}

Next, we discuss the randomized algorithm $\rDecompose$.

\begin{proof}
    Using the randomized construction $\rDominate$, we can likewise choose $p_i = \frac{c \log n}{s_i}$ for large enough $c$ to ensure correctness by the union bound.
\end{proof}

\weightdecompositionlemma*

\begin{proof}
    As before, constructing $E_i$ can easily be done in $O(m) = O(n^2)$ time.
    To construct $D_i$, for every vertex with degree at least $s_i$, we compute a hitting set for its neighborhood restricted to the $s_i - 1$ lightest neighbors.
    In particular, we compute a hitting set for the sets $\set{v} \cup \set{w \in N(v) \given \rank_v(w) < s_i}$ of size $s_i$ for each $v \in V_{s_i}$.
    Then, this guarantees that the edge $(v, w) \in E_i^*$ satisfies the additional constraint $\rank_v(v, w) < s_i$ for $v \in V_{s_i}$.
    By our previous arguments, $D_i$ dominates $V_{s_i}$ and has size $\bigtO{\frac{n}{s_i}}$.
    Note that as before computing $D_i$ requires $\tO{n^2}$ time.
\end{proof}

\subsection{Useful Results on Fast Matrix Multiplication}

We will extensively use various algorithms for efficiently computing matrix products.
We lay out some of the key results used below.
The matrix multiplication exponent $\omega(a, b, c)$ denotes infimum over all constants $t$ such that we can multiply two matrices of dimension $n^a \times n^b$ and $n^b \times n^c$ in $O(n^t)$ time.
Denote $\omega(1, 1, 1)$ as the constant $\omega$.
As the results given in this section are quite involved, we omit them for brevity and instead provide references to the relevant proofs.

\begin{theorem}
    \label{thm:duan_wu_zhou}
    (\cite{duan2022matrixmult}) The matrix multiplication exponent $\omega \leq 2.37188$.
\end{theorem}

See Duan, Wu, and Zhou \cite{duan2022matrixmult} for the appropriate proof.
The dual matrix multiplication exponent $\alpha$ is the supremum over all constants $k$ such that $\omega(1, k, 1) \leq 2$.

\begin{theorem}
    \label{thm:gall_urrutia}
    (\cite{gall2018rectangularmm})
    The dual matrix multiplication exponent $\alpha \geq 0.31389$.
\end{theorem}

\begin{lemma}
    \label{lemma:gall_urrutia_curve}
    (\cite{le2012rectangular})
    For any fixed value $0 \leq k_0 < 1$ and $k_0 \leq k \leq 1$,
    \begin{equation*}
        \omega(1, k, 1) \leq \omega(1, k_0, 1) + (\omega - \omega(1, k_0, 1)) \frac{k - k_0}{1 - k_0}
    \end{equation*}
\end{lemma}

We refer the reader to Le Gall and Urrutia \cite{gall2018rectangularmm} and Le Gall \cite{le2012rectangular} for the proofs of the above two results.
We will also use some useful properties of $\omega$, in particular that it is convex, symmetric and linear.

\begin{lemma}
    \label{lemma:lotti-convexity}
    (Corollary 5.1 of \cite{lotti1983rectangularmm})
    Suppose $\omega(1, a, b) = \omega(1, a', b') = c$.
    Then, for $t \in [0, 1]$,
    \begin{equation*}
        \omega(1, t*a + (1 - t)*a', t*b + (1  - t)*b') \leq c
    \end{equation*}
\end{lemma}

\begin{lemma}
    \label{lemma:omega-symmetry-linearity}
    Let $t, a, b, c \geq 0$.
    Then,
    \begin{equation*}
        \omega(a, b, c) = \omega(a, c, b) = \omega(b, a, c) = \omega(b, c, a) = \omega(c, a, b) = \omega(c, b, a)
    \end{equation*}
    and
    \begin{equation*}
        \omega(t a, t b, t c) = t \omega(a, b, c)
    \end{equation*}
\end{lemma}

See, for example, Lotti and Romani \cite{lotti1983rectangularmm} for the former result and Huang and Pan \cite{huang1998pan} for the latter.
The following theorem of Huang and Pan \cite{huang1998pan} gives improved bounds on general rectangular matrix multiplication for matrix products where $a, b, c$ are distinct.

\begin{theorem}
    \label{thm:huang_pan}
    (Theorem 8.1 of \cite{huang1998pan})
    Suppose $r > 1 > t > 0$.
    Then,
    \begin{equation*}
        \omega(t, 1, r) \leq \begin{cases}
            r + 1 & t \leq \alpha \\ 
            \frac{r(1 - \alpha) + (1 - t) + (\omega - 1)(t - \alpha)}{1 - \alpha} & \alpha < t \leq 1
        \end{cases}
    \end{equation*}
\end{theorem}

See Huang and Pan \cite{huang1998pan} for the proof of the above result.
We will use the following corollary as a generalization,

\begin{restatable}{corollary}{abcrectmatrix}
    \label{cor:a-b-c-rect-matrix}
    Suppose $0 < a < b < c$.
    Then,
    \begin{equation*}
        \omega(a, b, c) \leq \begin{cases}
            b + c & \frac{a}{b} \leq \alpha \\
            \frac{a(\omega - 2) + b(1 + \alpha - \alpha \cdot \omega) + c (1 - \alpha)}{1 - \alpha} & \alpha < \frac{a}{b} \leq 1
        \end{cases}
    \end{equation*}
\end{restatable}

\begin{proof}
    By linearity (Lemma \ref{lemma:omega-symmetry-linearity}), we have,
    \begin{align*}
        \omega(a, b, c) &= b \cdot \omega\left(\frac{a}{b}, 1, \frac{c}{b} \right) \\
        &\leq \begin{cases}
            b \cdot \left(\frac{c}{b} + 1\right) & \frac{a}{b} \leq \alpha \\
            b \cdot \frac{\frac{c}{b}(1 - \alpha) + \left(1 - \frac{a}{b}\right) + (\omega - 1)\left( \frac{a}{b} - \alpha \right)}{1 - \alpha}
        \end{cases}\\
        &= \begin{cases}
            b + c & \frac{a}{b} \leq \alpha \\
            \frac{a(\omega - 2) + b(1+\alpha-\alpha\omega) + c(1 - \alpha)}{1 - \alpha}
        \end{cases}
    \end{align*}
\end{proof}

\subsection{Fast \texorpdfstring{\minplus}{(min, +)} Matrix Multiplication}
\label{sec:prelims:fast-min-plus}

Recall that for two matrices $A, B$, their \minplus product, denoted $A * B$, is defined as,
\begin{equation*}
    (A * B)_{ij} = \min_k A_{ik} + B_{kj}
\end{equation*}

While computing \minplus product in general is difficult, there are several structured instances in which efficient algorithms are known beyond the trivial $\tO{n^3}$ algorithm.
We describe a few that we will require here.
First, we begin with the case that all entries are bounded.

\begin{restatable}{theorem}{boundedminplusthm}
    \label{thm:bounded-min-plus}
    (\cite{alon1997exponent})
    Let $A, B$ be two $n \times n$ integer matrices with entries bounded in absolute value by $M$.
    Then, there is an algorithm $\boundedMinPlus$ computing the \minplus product $C = A * B$ in time $\tO{M n^{\omega}}$.
\end{restatable}

\begin{proof}
    Without loss of generality, we can assume the entries of $A, B$ are in fact non-negative and therefore in the range $[0, M]$.
    Otherwise, define non-negative matrices $A'_{ik} = A_{ik} + M$ and $B'_{kj} = B_{kj} + M$ with values in the range $[0, 2M]$.
    Given the product $A' * B'$, we can easily recover $A * B$ as,
    \begin{equation*}
        (A' * B')_{ij} - 2 M = \min_{k} A'_{ik} + B'_{kj} - 2M = \min_{k} A_{ik} + B_{kj} = (A * B)_{ij}
    \end{equation*}

    Given $A$, we define the polynomial matrices $A^p_{ik} = x^{A_{ik}}$ if $A_{ik}$ is finite and $0$ otherwise.
    We define $B^p$ analogously, and compute $C^p = A^p B^p$ using fast matrix multiplication in $\tO{M n^{\omega}}$ time.
    Then, for each entry in $C^p$, if $C^p_{ij} = 0$, we set $C_{ij} = \infty$ and otherwise, we compute the lowest degree term $x^{c}$ and set $C_{ij} = c$.
    Now, we prove that that $C = A * B$.
    Recall that $(A*B)_{ij} = \min_{k} A_{ik} + B_{kj}$.
    $(A * B)_{ij}$ is infinite if and only if all $A_{ik} + B_{kj}$ are infinite.
    If either $A_{ik}, B_{kj}$ is infinite, then the corresponding entry in $A^p, B^p$ is zero.
    Thus, if all $A_{ik} + B_{kj}$ are infinite, $C^p_{ij} = 0$ so $C_{ij} = \infty$.

    Suppose $(A * B)_{ij}$ is finite.
    Then $C^p_{ij} = \sum_{k} x^{A_{ik} + B_{kj}}$ where the sum ranges over all $k$ where $A_{ik} + B_{kj}$ is finite.
    Taking the minimum degree, we have $(A * B)_{ij}$, proving correctness.

    Let us now analyze the performance of the above algorithm.
    Constructing $A^p$ from $A$, $B^p$ from $B$, and $C$ from $C^p$ requires time at most $\tO{M n^2}$.
    Computing $C^p = A^p B^p$ requires $\tO{M n^{\omega}}$ as each multiplication of at most $M$ degree polynomials requires $\tO{M}$ time.
\end{proof}

More generally, there is also a sub-cubic algorithm when the entries have bounded differences \cite{bringmann2019boundeddifference}.

\begin{definition}
    \label{def:bounded-difference-matrix}
    An $n \times n$ matrix is called {\bf $b$-bounded difference} if for all $i, j$, $|A_{i, j} - A_{i, j + 1}| \leq b$ and $|A_{i, j} - A_{i + 1}, j| \leq b$.
    If only the former holds, the matrix is {\bf row bounded difference}.
    If only the latter holds, the matrix is {\bf column bounded difference}.
\end{definition}

A more general requirement may be that the entries are monotone. 
That is, the values of each row (or column) are monotonically increasing.

\begin{definition}
    \label{def:monotone-matrix}
    An $n \times n$ matrix is called {\bf row-monotone} if all entries are non-negative integers bounded by $O(n)$ and each row of this matrix is non-decreasing.
    A {\bf column-monotone} matrix is defined similarly. 
\end{definition}

\begin{theorem}
    \label{thm:monotone-min-plus}
    (\cite{chi2022monotone})
    Let $A, B$ be two $n \times n$ integer matrices with entries $A_{ij}, B_{ij} = O(n)$.
    Suppose furthermore that $B$ is either row-monotone or column-monotone.
    Then, there is an algorithm $\monotoneMinPlus$ computing the \minplus product $C = A * B$ in time $\bigtO{n^{\frac{3 + \omega}{2}}}$
\end{theorem}

See Chi, Duan, Xie, and Zhang \cite{chi2022monotone} for a detailed proof.
We also provide a proof overview in the proof of \Cref{thm:rect-monotone-min-plus}.
Given a $b$-bounded difference matrix, it is easy to construct a row-monotone matrix, by adding the $b \cdot j$ to the $j$-th column.
To recover the \minplus product, simply subtract the same constant from each column in the result.
We describe this procedure as $\bdToMonotone$.
By a similar argument, we can convert row-bounded difference matrices to row-monotone matrices, and column-bounded difference matrices to column-monotone matrices.
We will also use the simple observation that the sub-matrix of a monotone matrix (chosen by selecting a subset of rows or columns) is a monotone matrix.

We also require the approximate \minplus product of Zwick \cite{zwick2002bridgingsets}.

\approxminplusthm

\begin{proof}
    Suppose the entries of $A, B$ are in the range $[0, M]$.
    Let the resolution $R$ be a parameter to be set later.
    Initialize $C_{ij} \gets \infty$ for all $i, j$.
    For every $\floor{\log R} \leq r \leq \ceil{\log M}$, define,
    \begin{align*}
        A_{ik}^{(r)} &\gets \begin{cases}
            \ceil{R A_{ik}/2^r} & 0 \leq A_{ik} \leq 2^r \\
            \infty & \otherwise
        \end{cases} \\
        B_{kj}^{(r)} &\gets \begin{cases}
            \ceil{R B_{kj}/2^r} & 0 \leq B_{kj} \leq 2^r \\
            \infty & \otherwise
        \end{cases}
    \end{align*}
    and compute $C^{(r)} \gets \boundedMinPlus(A^{(r)}, B^{(r)}, R)$.
    Update $C_{ij} \gets \min\left(C_{ij}, \frac{2^r}{R} C^{(r)}\right)$.
    The inequality $(A * B)_{ij} \leq C_{ij}$ follows as $A^{(r)}, B^{(r)}$ always round its entries up.
    Suppose $(A * B)_{ij} = A_{iq} + B_{qj}$ and without loss of generality $A_{iq} \leq B_{qj}$.
    Let $1 \leq s \leq \ceil{\log M}$ where $2^{s - 1} \leq B_{kj} \leq 2^s$. 
    If $s \leq \log R$, then the $C^{\log R}$ computes $(A * B)_{ij}$ exactly.
    Thus, assume $s \geq \log R$.
    Then,
    \begin{equation*}
        \frac{2^r}{R} C^{(r)}_{ij} \leq \frac{2^r}{R} \left( A^{(r)}_{iq} + B^{(r)}_{qj} \right) \leq A_{iq} + B_{qj} + \frac{2^{r + 1}}{R} \leq \left( 1 + \frac{4}{R} \right) (A * B)_{ij}
    \end{equation*}
    Then, if we set $R = \bigO{1/\eps}$, we obtain the approximation error $(1 + \eps)$ and note that each invocation of $\boundedMinPlus$ requires $\tO{n^{\omega}/\eps}$ time.
\end{proof}
\section{Missing Proofs}

\begin{lemma}
    \label{lemma:dominating-apasp-feasible}
    Let $G$ be an unweighted, undirected graph on $n$ vertices.
    Let $\hat{\delta}$ be the estimate output by Algorithm \ref{alg:dominating-set-apasp}.
    Then, $\delta(u, v) \leq \hat{\delta}(u, v)$ for all $u, v \in V$.
\end{lemma}
\begin{proof}
    Consider some $t = 2^j$.
    Then,
    \begin{equation*}
        \hat{\delta}(u, v) = \min_{w \in W} \hat{\delta}(u, w) + \hat{\delta}(w, v) \delta(u, v)
    \end{equation*}
    so that \Cref{alg:dominating-set-apasp} outputs $\delta(u, v) \leq \hat{\delta}(u, v)$ for all $u, v$.
\end{proof}

\begin{lemma}
    \label{lemma:2-approx-feasible}
    Let $G$ be an unweighted, undirected graph on $n$ vertices.
    Let $\hat{\delta}$ be the estimate output by Algorithm \ref{alg:2-approx-apasp}.
    Then, $\delta(u, v) \leq \hat{\delta}(u, v)$ for all $u, v \in V$.
\end{lemma}

\begin{proof}
    We will show for all vertex pairs $(u, v)$ that $\delta(u, v) \leq \hat{\delta}(u, v)$. 
    To do so, it suffices to observe that every estimated distance is witnessed by some path in the graph $G$.

    Initially, the only finite entries of $\hat{\delta}$ are edges $(u, v)$ in the graph $G$.
    Phase 1 calls $\dominatingAPASP(G_t, V, V, C_t, C + 1)$ so that $\hat{\delta}(u, v) \geq \delta(G_t, u, v) \geq \delta(u, v)$ for all $u, v$ since $G_t \subset G$ is a sub-graph and \Cref{lemma:dominating-apasp-feasible}.

    Similarly, Phase 2.1 calls $\dominatingAPASP(G_{s_1}, V, D_{k - 1}, D_1, C + 2)$ so that after Phase 2.1, $\hat{\delta}(u, v) \geq \delta(G_{s_1}, u, v) \geq \delta(u, v)$ as $G_{s_1} \subset G$.

    In Phase 2.2, any path found by $\bfs$ in subgraphs $G_x, G_2 \subset G$ is naturally also a path in $G$.

    In Phase 2.3, since for any $u \in q(w, D_i)$, $(u, w) \in E$, by the triangle inequality,
    \begin{equation*}
        1 + \hat{\delta}(u, v) \geq 1 + \delta(u, v) \geq \delta(w, v)
    \end{equation*}
    Thus, any distance in $G_{i, w}$ is witnessed by some path in $G$.
    Therefore, executing $\dijkstra$ on $G_{i, w}$ with weight function $\hat{\delta}$ will not return any distance $\hat{\delta}$ smaller than the shortest distance between $u, v$.

    Finally, we conclude by noting the correctness of $\boundedAdditiveAPASP$ and $\denseAPASP$ implies $\delta(u, v) \leq \hat{\delta}(u, v)$.
\end{proof}

\begin{proposition}
    \label{prop:2-approx-min-plus-convexity}
    For any $x \in [0, 1]$ and any integer $k \geq 1$,
    \begin{equation*}
        \omega \left(1, 1 - \frac{(k - 1)}{k} x, 1 - \frac{x}{k} \right) \leq \omega(1, 1 - x, 1)
    \end{equation*}
\end{proposition}

\begin{proof}
    Recall that by symmetry of $\omega$, $\omega(1, 1 - x, 1) = \omega(1, 1, 1 - x)$.
    The claim follows from applying Lemma \ref{lemma:lotti-convexity} with $t = \frac{1}{k}$.
\end{proof}

\begin{proposition}
    \label{prop:mult-approx-bk-fmm-examples}
    We obtain the running times in \Cref{tbl:2-approx-geq-beta} for $(2, 0)$ approximation on paths of length $\delta(u, v) \geq k$ from \Cref{thm:mult-approx-bk-fmm}.
\end{proposition}

\begin{proof}
    Recall from \Cref{thm:mult-approx-bk-fmm} we wish to minimize the following,

    \begin{equation*}
        \max \left( \omega\left( 1 - \frac{k - 2}{2(l - k + 2)} x, 1 - x, 1 - \frac{k - 4}{2(l - k + 2)} x \right), 2 + \frac{x}{l - k + 2}, 1.5 + x \right)
    \end{equation*}

    In all cases, we can set $M$ to be some arbitrary constant without affecting the asymptotic running time.

    \paragraph{$k = 6$}
    We minimize,
    \begin{equation*}
        \max \left( \omega\left( 1 - \frac{2 x}{l - 4}, 1 - x, 1 - \frac{x}{l - 4} \right), 2 + \frac{x}{l - 4}, 1.5 + x \right)
    \end{equation*}

    Using \cite{Complexity}, we set $l = 51, x = 0.50980335$ and obtain the running time $\tO{n^{2.01084688}}$.

    \paragraph{$k = 8$}
    We minimize,
    \begin{equation*}
        \max \left( \omega\left( 1 - \frac{3 x}{l - 6}, 1 - x, 1 - \frac{2x}{l - 6} \right), 2 + \frac{x}{l - 6}, 1.5 + x \right)
    \end{equation*}

    Using \cite{Complexity}, we set $l = 74, x = 0.50716126$ and obtain the running time $\tO{n^{2.00745825}}$.

    \paragraph{$k = 10$}
    We minimize,
    \begin{equation*}
        \max \left( \omega\left( 1 - \frac{4 x}{l - 8}, 1 - x, 1 - \frac{3x}{l - 8} \right), 2 + \frac{x}{l - 8}, 1.5 + x \right)
    \end{equation*}

    Using \cite{Complexity}, we set $l = 96, x = 0.50496404$ and obtain the running time $\tO{n^{2.00573823}}$.

    \paragraph{$k = 12$}
    We minimize,
    \begin{equation*}
        \max \left( \omega\left( 1 - \frac{5 x}{l - 10}, 1 - x, 1 - \frac{4 x}{l - 10} \right), 2 + \frac{x}{l - 10}, 1.5 + x \right)
    \end{equation*}

    Using \cite{Complexity}, we set $l = 119, x = 0.50432041$ and obtain the running time $\tO{n^{2.00462679}}$.
\end{proof}

\section{A Quadratic \texorpdfstring{$(\frac{7}{3}, 0)$}{(7/3,0)}-Approximate APSP on Unweighted Graphs}
\label{sec:7/3-det-approx}

In the section, we briefly observe that a $\frac{7}{3}$-approximate APSP solution can be obtained in $\tO{n^2}$-time.
This follows simply from a $\tO{n^2}$-time $(2, 1)$-approximate APSP algorithm due to Baswana and Kavitha \cite{baswana2010fasterapasp} and Berman and Kasiviswanathan \cite{berman2007approxapsp} along with the result of \cite{dor2000apasp}

\begin{lemma}
    \label{lemma:distance-2-apsp}
    Given an undirected, unweighted graph, \Cref{alg:distance-2-apasp} outputs a $(2, 0)$ approximation for paths of length $\delta(u, v) \leq 2$.
    in $\tO{n^2}$-time.
\end{lemma}

\IncMargin{1em}
\begin{algorithm}[H]

\SetKwInOut{Input}{Input}\SetKwInOut{Output}{Output}
\Input{Unweighted, undirected Graph $G = (V, E)$ with $n$ vertices}
\Output{Distance estimate $\hat{\delta}: U \times V \rightarrow \Z$ such that $\delta(u, v) \leq \hat{\delta}(u, v)$ for all $u, v$ and $\hat{\delta}(u, v) \leq 4$ for all $\delta(u, v) = 2$}

\BlankLine

$\hat{d}(u, v) \gets \begin{cases}
    1 & (u, v) \in E \\ 
    \infty & \otherwise
\end{cases}$

$s_i \gets \frac{n}{2^i}$ for all $1 \leq i \leq k - 1 = \floor{\log n}$

$(D_1, D_2, \dotsc, D_k), (E_1, E_2, \dotsc, E_k), E^* \gets \decompose(G, (s_1, s_2, \dotsc, s_{k - 1}))$

\For{$1 \leq i \leq k = \floor{\log n} + 1$}{
    \For{$w \in D_i$}{
        $X_i(w) \gets \dijkstra(w)$ on graph $\left(V, E_i \cup (w \times V) \cup E^*, \hat{\delta} \right)$
        
        $\hat{\delta}(w, v) \gets \min(\hat{\delta}(w, v), X_i(w, v))$
    }
}

$\hat{\delta}(u, v) \gets \min(\hat{\delta}(u, v), \hat{\delta}(v, u))$

\caption{$\twoedgeAPASP(G)$}
\label{alg:distance-2-apasp}
\end{algorithm}
\DecMargin{1em}

\begin{proof}
    First we analyze the running time.
    Note $k = \tO{1}$ so that $\decompose$ requires $\tO{n^2}$ time.
    Fix an iteration $i$.
    Then $|D_i| = \bigtO{\frac{n}{s_i}}$, $E_{i}$ has $n s_{i - 1}$ edges, while $E^* \cup (w \times V)$ have size $\tO{n}$.
    Thus, each iteration $i$ requires $\tO{n^2}$ time.
    Over $\log n$ iterations, Algorithm \ref{alg:distance-2-apasp} requires $\tO{n^2}$ time.

    Consider now $\delta(u, v) = 2$ and a path $P = (u, w, v)$.
    Without loss of generality, assume $\deg(u) \geq \deg(v)$ so that $P \subset E_{\level(u)}$.
    Let $u^* \in D_{\level(u)}$ be its representative $r(u, D_{\level(u)})$.
    Then, the execution of $\dijkstra$ from $u^*$ returns $\hat{\delta}(u^*, v) \leq 3$ via the path $(u^*, u, w, v)$.
    Now, in the final iteration, we have $\hat{\delta}(v, u) \leq 4$ via the path $(v, u^*, u)$ when executing $\dijkstra$ from vertex $v$.

    Finally, it is easy to observe that $\hat{\delta}(u, v) \geq \delta(u, v)$ for all $u, v$ (without any assumption on the distance) since all traversed paths exist in the original graph $G$.
\end{proof}

Combining Lemma \ref{lemma:distance-2-apsp} with the algorithm of Baswana et al. \cite{baswana2010fasterapasp} or Berman et al. \cite{berman2007approxapsp} gives the desired result.

\begin{corollary}
    \label{thm:7/3-approx-apasp}
    There is an $\tO{n^{2}}$-time algorithm that computes a $\left( \frac{7}{3}, 0 \right)$-approximate APSP solution in undirected, unweighted graphs.
\end{corollary}
\begin{proof}
    Of course length $1$ paths can be directly deduced from the adjacency matrix in time $O(n^2)$. From Lemma~\ref{lemma:distance-2-apsp}, all distance $2$ paths can be approximated within $(2,0)$-approximation in $\bigtO{n^2}$ time. For paths of length $3$ or more, a $(2,1)$-approximation is at least as good as a $(\frac{7}{3},0)$ approximation.
\end{proof}

\section{Faster Combinatorial \texorpdfstring{$(2,0)$}{(2,0)}-approximation for path lengths \texorpdfstring{$\geq 4$}{>=4}.}
\label{sec:application-mult-approx-long-path}

In this section, we give an improved combinatorial algorithm for finding $(2, 0)$ approximations for paths with length $k \geq 4$. We obtain a $+\beta$-additive approximations for paths of length at most $\beta+1$ for every even $\beta$ on dense graphs. In particular, we apply the algorithm with $\beta = 4$ to compute a $+4$ approximation on paths with length at most 5.
This implies a $(2, 0)$ approximation for paths of length at least $4$.

\begin{restatable}{theorem}{shortadditiveapprox}
    Let $G$ be an undirected, unweighted graph with $n$ vertices.
    Let $\beta \geq 4$ be an even odd integer.
    \Cref{alg:short-additive-apasp} computes in expected time $\bigtO{n^{2 + \frac{2}{3 \beta + 2}}}$ a distance estimate $\hat{\delta}$ such that $\delta(u, v) \leq \hat{\delta}(u, v)$ for all $u, v \in V$ and $\hat{\delta}(u, v) \leq \delta(u, v) + \beta$ for all $\delta(u, v) \leq \beta + 1$.
    \label{thm:short-additive-approx}
\end{restatable}

\paragraph{Notation and Definitions}

We define some notation that will be useful in this section.
Let $G$ be an undirected, unweighted graph.
Let $\beta$ be an even integer, $k = \frac{3 \beta + 2}{2}$ and $k_0 = \beta - 2$.
Let $s_i, s_2, \dotsc, s_{k - 1}$ be degree thresholds.
For $k_0 + 1 \leq k \leq k_0 + \frac{\beta}{2} - 1$ define $r_i = \frac{\beta}{2} - (i - (k_0 + 1))$ and let $N(u, r_i) = \set{v \in V \given \delta(u, v) \leq r_i}$ be the $r_i$-neighborhood of $u$.

\begin{definition}
    \label{def:nearest-vertex-set}
    For a vertex $u \in V$ and set $S \subset V$, let $p(u, S) = \arg \min_{v \in S} \delta(u, v)$ be the closest vertex in $S$ to $u$.
\end{definition}

We note that $p(u, S)$ can be efficiently computed by executing $\dijkstra$ on the graph $G$ augmented with a dummy vertex $w$ with 0 weight edges to every vertex $v \in S$.
For a given $i$, $X(u, i)$ is an arbitrary subset of $N(u, r_i)$ if $|N(u, r_i)| \geq s_i$ and $C(u_i) = N(u, r_i) \cap D_{i + 1}$ if $|N(u, r_i) \cap D_{i + 1}| \leq 12 s_{k - 1} \log n$.
Roughly speaking, $X(u, i)$ is the first $s_i$ vertices encountered when conducting a $\bfs$ from $u$.
For details on how $X(u, i), C(u, i)$ are computed see \Cref{lemma:x-c-computation}. 

\paragraph{High Level Overview}

To illustrate our algorithm, we walk through an overview of how to compute a $+4$ approximation for paths of length at most 5.
For a $+4$ approximation, the $\denseAPASP$ algorithm decomposes graph $G$ into $5$ levels.
Instead, we will show that we can decompose the graph $G$ into $7$ levels.
Thus, define thresholds $s_1 = n^{6/7}, s_2 = n^{5/7}, \dotsc, s_6 = n^{1/7}$.
As length 5 paths are the most challenging to handle, let $u, v$ be vertices and $P$ a shortest path of length at exactly $5$.
If $P$ has a vertex of degree at least $s_{2}$, since $2 + 6 + 7 = 15 = 2 * 7 + 1$, we may include the edge sets $\bigcup_{i = 1}^{2} D_{i} \times D_{6} = \tO{|E_9|}$. we can include these edges and obtain a $+4$ approximation following similar arguments to $\denseAPASP$ (\Cref{fig:+6-min-plus-approx}).

Thus, let us assume $P \subset E_3$.
If $P \subset E_5$, then the blocking levels $L(P, 7) \subset \set{5, 6}$ (\Cref{def:blocking-vertices}) so that we can obtain a $+4$ approximation via \Cref{lemma:sparse-apasp-approx}.
Then, assume $P \not\subset E_5$, so $P$ has some edge in $E_j \setminus E_{j + 1}$ for $j \in \set{3, 4}$.
Denote the length 5 path $P = (u, u_2, u_3, v_3, v_2, v)$.

We will argue that if $\level(u_2) \leq 4$, then we obtain a good additive approximation.
Suppose $|N(u, 2)| \geq s_3$, so that from $p(u, D_3) \in N(u, 2)$ (\Cref{def:nearest-vertex-set}) we compute an exact distance $\hat{\delta}(p(u, D_3), v) = \delta(p(u, D_3), v) \leq \delta(u, v) + 2$.
When we add $2$, we obtain a $+4$ approximation.
\Cref{fig:x(u)-approx} gives an illustration\\

\begin{figure}[ht]
    \centering
    \includegraphics[width=0.6\textwidth]{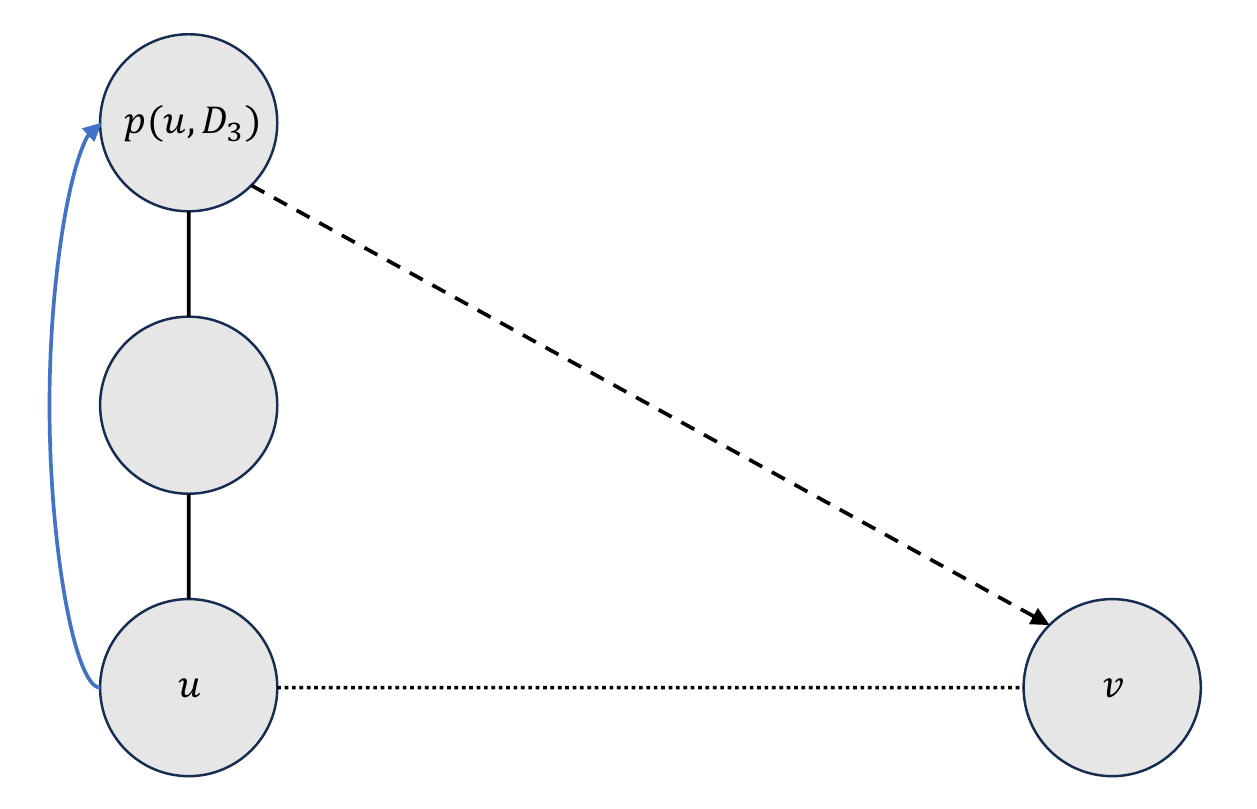}
    \caption{Estimate when $|N(u, 2)| \geq s_3$.
    Solid lines denote edges and dotted lines denote paths in $G$.
    Black dashed arrows denote computed distances estimates.
    Blue solid arrows denote a constant additive term.}
    \label{fig:x(u)-approx}
\end{figure}

$\shortAdditiveAPASP$ (\Cref{alg:short-additive-apasp}) samples each dominating set $D_i$ randomly with the algorithm $\rDominate$, rather than following the greedy deterministic construction of $\dominate$.
With high probability, if $|N(u, 2)| \leq s_3$, then $|C(u, 3)| = |N(u, 2) \cap D_{4}| = \tO{s_{6}}$ and if $\level(u_2) \leq 4$, then \Cref{lemma:sparse-apasp-approx} guarantees that $\hat{\delta}(u_2^*, v) \leq \delta(u_2, v) + 3 = \delta(u, v) + 2$.
We efficiently iterate over $C(u, 3)$ and update $\hat{\delta}(u, v) \leq 2 + \min_{w \in C(u, 3)} \hat{\delta}(w, v)$ to obtain a $+4$ approximation.
\Cref{fig:degree-restriction-intro} gives an illustration.

\begin{figure}[ht]
\centering
\includegraphics[width=0.9\textwidth]{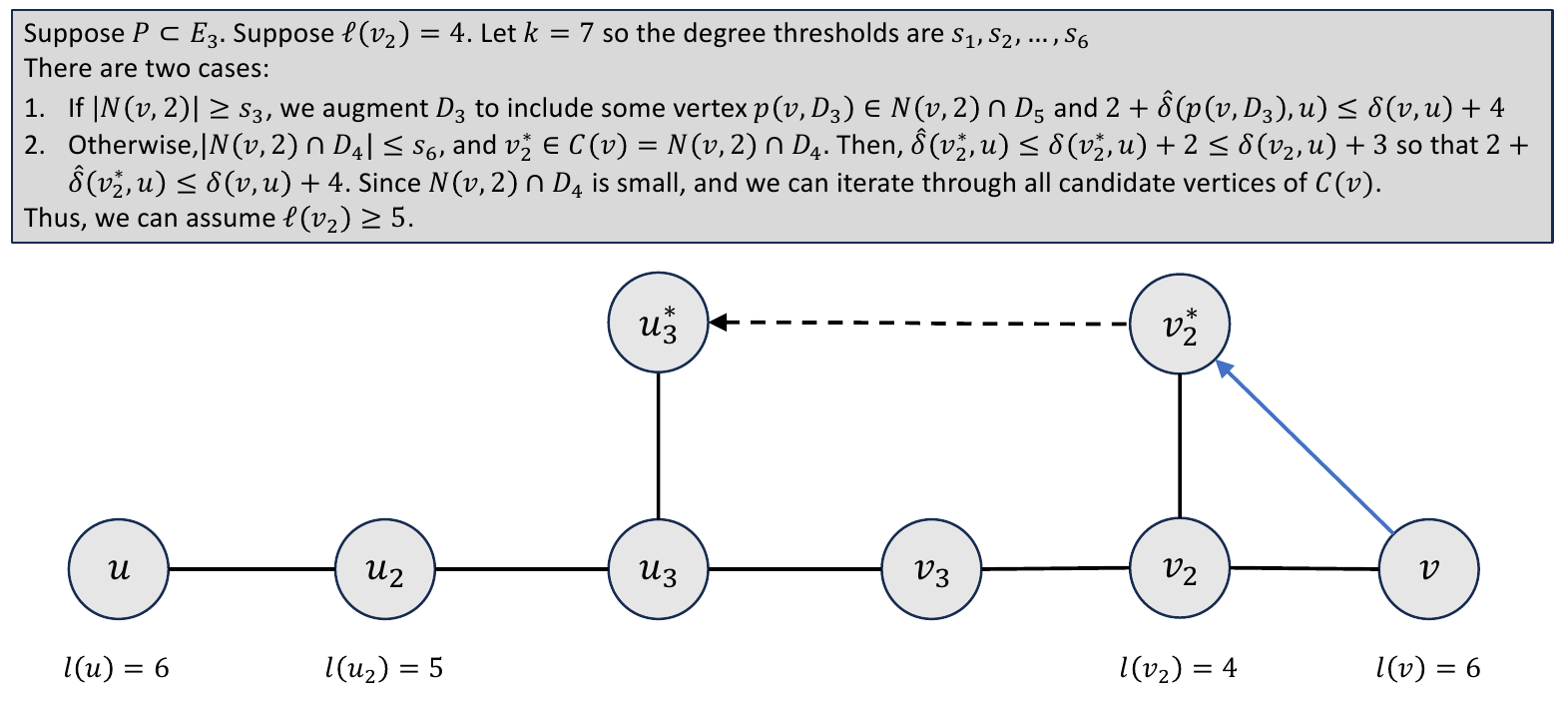}
\caption{Imposing the degree restriction $\level(v_2) \geq 5$}
\label{fig:degree-restriction-intro}
\end{figure}

We can compute the sets $X(u, i), C(u, i)$ efficiently as we will terminate our $\bfs$ whenever $s_i$ distinct vertices are encountered, which requires time at most $s_i^2 \leq s_{3}^2 = \tO{n^{8/7}}$.

Then, we can assume $\deg(u_2), \deg(v_2) < s_4$.
In particular, only one edge $(u_3, v_3)$ is not in $E_5$.
We now execute $\dijkstra$ from every vertex $w \in D_j$ on graph $G_{i, w}$ containing (at least) the edges $E_j \cup E^* \cup (w \times V)$.

Since only one edge is not in $E_5$, the blocking levels of $P$ (\Cref{def:blocking-vertices}), $L_B(P)$ cannot contain both $\set{3, 4}$.
If the blocking levels $L_B(P)$ do not contain every level from $\set{5, 6}$, then we obtain a good additive approximation from \Cref{lemma:block-approx-error} as there are at most 2 blocking levels.
On the other hand, suppose $\set{5, 6} \subset L_B(P)$.
Note $\level(u_3), \level(v_3) < 5$, otherwise we obtain a $+4$ approximation as $L_B(P) = \set{5, 6}$.
Assume without loss of generality that $\level(v) \leq \level(u)$ as in \Cref{def:blocking-vertices}.
If $\level(v) = 5$, then since $P \subset E_{\level(u_3)}$, $\hat{\delta}(u_3^*, v) \leq \delta(u_3, v)$.
Since the remaining edges are in $E_5$, $G_{5, v^*}$

Otherwise, $\level(v) \geq 6$.
If $\level(v) = 7$, then $\level(u) = \level(v) = 7$ and we obtain a $+4$ approximation following the same argument as \Cref{thm:2-approx-apsp}.
Therefore, assume $\level(v) = 6$.
If $\level(u_2) \geq 6$, then the blocking vertex $b(v, P) = u_3$ and we obtain a $+4$ approximation as $B(P) = \set{u, v, u_3}$.

We can therefore assume $\level(u_2) = 5$.
However, since this vertex has degree at most $s_4$, we use a depth 2 $\bfs$ in \Cref{line:short-approx:depth-2-bfs} to bypass the middle edge entirely, therefore obtaining a $+4$ approximation, as illustrated in \Cref{fig:bfs-2-intro}.
The depth 2 $\bfs$ requires time $s_2 s_4 = \tO{n^{8/7}}$.

\begin{figure}[ht]
\centering
\includegraphics[width=0.9\textwidth]{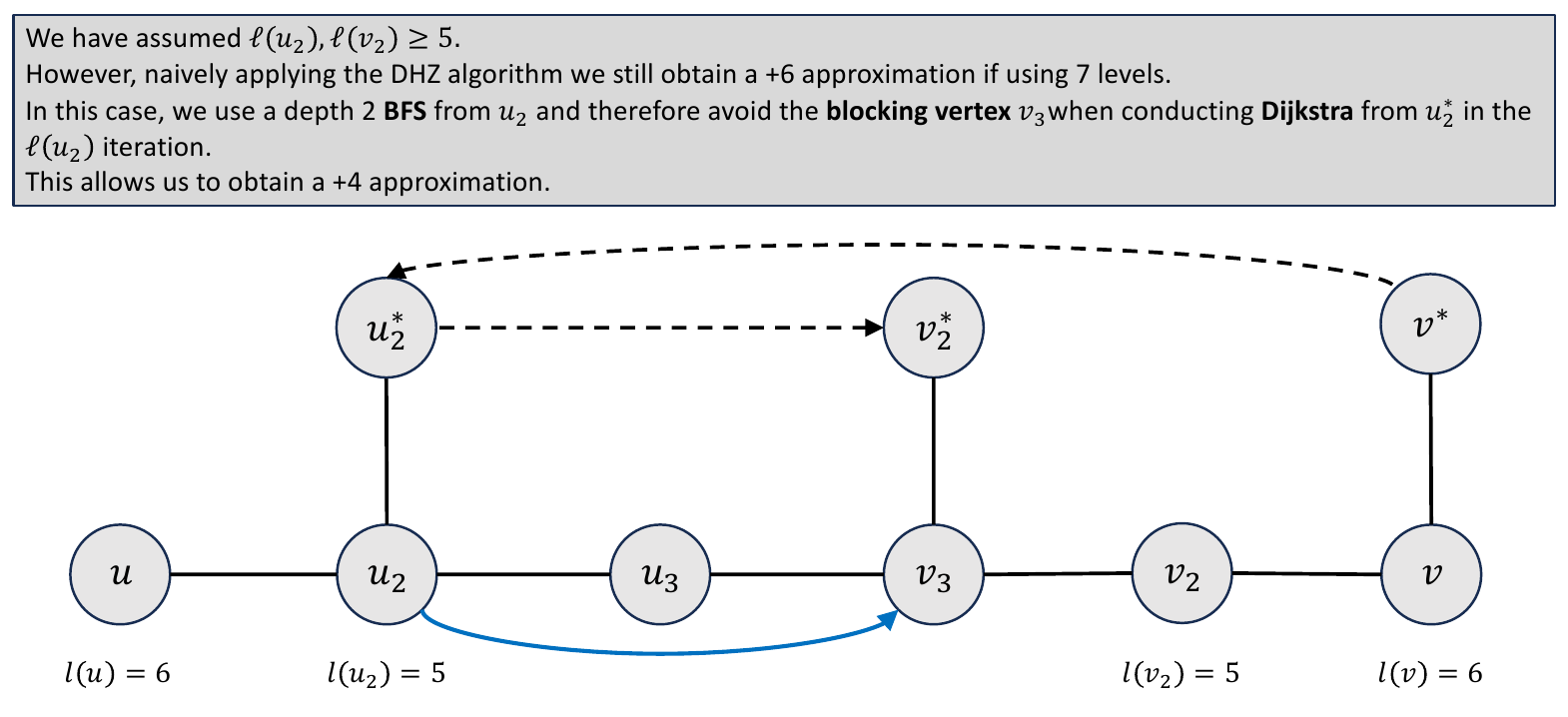}
\caption{Using depth 2 $\bfs$ to bypass the central edge.}
\label{fig:bfs-2-intro}
\end{figure}

\paragraph{Algorithm}

In Phase 0, we initialize the distance estimates to the adjacency matrix, set parameters $k, k_0$ and decompose the graph $G$ into $k$ degree thresholds.
To compute the dominating sets, we use a random construction $\rDecompose$ which samples all vertices with identical probability.
Finally, we check that the dominating sets are computed correctly and are not too large.
If there is an error in the sample dominating sets, we return an exact $\APSP$ solution.

In Phase 1, we compute the sets $X(u, i)$, $C(u, i)$.
For a given $i$, if $|N(u, r_i)| \geq s_i$ is large, we let $X(u, i)$ be an arbitrary subset of size $s_i$ within this neighborhood.
On the other hand, if $|N(u, r_i)| < s_i$ is small, we let $C(u, i) = N(u, r_i) \cap D_{i + 1}$ given $|N(u, r_i) \cap D_{i + 1}|$ is small.
If we are unable to compute either $X(u, i)$ or $C(u, i)$, we return an exact $\APSP$ computation.

Then, for vertices $u$ with degree at most $s_{k_0 + 2}$, we compute a depth 2 $\bfs$ in the graph $G_{k_0 + 1}$.
Finally, we augment $D_i$ with a hitting set for $X(u, i)$ for all $u$ with large $N(u, r_i)$.
Note that when we augment $D_i$, this does not affect the sets $r(u, D_i)$ for $\deg(u) \geq s_i$ or $q(w, D_i)$ for $w \in D_i$.
These relationships are fixed at the computation of $\rDecompose$.

In Phase 2, we begin by iteratively executing $\dijkstra$ from each vertex $w \in D_j$ on the graph $G_{j, w}$ consisting of edges $E_j, E^*$ as well as $w \times V$ and $D_{j_1} \times D_{j_2}$ for all indices $j + j_1 + j_2 \leq 2 k + 1$.
Before the execution of $\dijkstra$, each vertex $w \in D_j$ examines its constituency (Definition \ref{def:dominating-set}) and updates its distance estimates where an improvement can be found.

Finally, each vertex $u$ examines $p(u, D_i)$ and $C(u, i)$ where appropriate to update its distance estimates for all other vertices $v \in V$.

\subsubsection*{Preliminaries for Proving Theorem \ref{thm:short-additive-approx}}

A key ingredient of our algorithm is computing small neighborhoods of each vertex.
In the following lemma, we describe how to compute $X(u, i), C(u, i)$ efficiently.

\begin{lemma}
    \label{lemma:x-c-computation}
    Suppose $\beta \geq 4$.
    Let $u \in V$ and $k_0 + 1 \leq i \leq k_0 + \frac{\beta}{2} - 1$.
    Algorithm \ref{alg:short-additive-apasp} computes $X(u, i), C(u, i)$ in time $\bigO{n^{2 + \frac{1}{k}}}$ such that,
    \begin{enumerate}
        \item $X(u, i) \subset N(u, r_i)$ and $|X(u, i)| \geq s_i$ if $|N(u, r_i)| \geq s_i$.
        \item $C(u, i) = N(u, r_i) \cap D_{i + 1}$ if $|N(u, r_i) \cap D_{i + 1}| < 12 s_{k - 1} \log n$.
    \end{enumerate}
    
\end{lemma}

\begin{proof}
    Fix some iteration $i$ and vertex $u$.
    We conduct a $\bfs$ on the graph $G$ from $u$ until we either encounter a vertex at depth $r_i + 1$ or $s_i$ distinct vertices.
    If we encounter $s_i$ distinct vertices first, then $|N(u, r_i)| \geq s_i$, so define $X(u, i)$ to be the distinct vertices encountered thus far.
    Otherwise, if we encounter a vertex at depth $r_i + 1$ first, then $|N(u, r_i)| < s_i$ and we can efficiently compute $C(u, i) = N(u, r_i) \cap D_{i + 1}$ by iterating over the distinct vertices found and querying their membership in $D_{i + 1}$.
    
    To bound the time required to compute these sets, note that finding $s_i$ distinct vertices requires time at most $s_i^2$.
    In particular, we can bound the time consumed by computing the $\bfs$ as,
    \begin{equation*}
        s_i^2 \leq s_{k_0 + 1}^2 = n^{2 - \frac{2 (k_0 + 1)}{k}}
    \end{equation*}
    Note $k_0 = \beta - 2 = \frac{2}{3} (k - 1) - 2 = \frac{2}{3} k - \frac{8}{3}$.
    Summing over all vertices $u$ and iterations $i$, the total time consumed is,
    \begin{equation*}
        \bigO{n^{3 - \frac{2 (k_0 + 1)}{k}}} = \bigO{n^{3 - \frac{4k/3 - 10/3}{k}}} =  \bigO{n^{\frac{5}{3} + \frac{10}{3k}}} = \bigO{n^{2 + \frac{1}{k}}} 
    \end{equation*}
    The last inequality holds whenever $k \geq 7$ or alternatively $\beta = \frac{2}{3} (k - 1) \geq 4$.
    Computing $C(u, i)$ requires only $s_i$ time as we only need to iterate over the vertices found in $N(u, r_i)$.
\end{proof}

\IncMargin{1em}
\begin{algorithm}[H]
\SetKwInOut{Input}{Input}
\SetKwInOut{Output}{Output}
\Input{Unweighted, undirected graph $G = (V, E)$ with $n$ vertices; approximation parameter $\beta$}
\Output{Distance estimate $\hat{\delta}: U \times V \rightarrow \Z$ such that $\delta(u, v) \leq \hat{\delta}(u, v)$ for all $u, v \in V$ and $\hat{\delta}(u, v) \leq \delta(u, v) + \beta$ whenever $\delta(u, v) \leq \beta + 1$}

\BlankLine

\textcolor{blue}{Phase 0: Set up and Decompose Graph}

$\hat{\delta}(u, v) \gets \begin{cases}
    1 & (u, v) \in E \\
    \infty & \otherwise
\end{cases}$

$k \gets \frac{3 \beta + 2}{2}$ and $k_0 \gets \beta - 2$.

$s_{i} \gets n^{1 - \frac{i}{k}}$ for all $1 \leq i \leq k - 1$

$(D_{1}, D_{2}, \dotsc, D_{k}), (E_{1}, E_{2}, \dotsc, E_{k}), E^* \gets \rDecompose(G, (s_{1}, s_{2}, \dotsc, s_{k-1}))$

\For{$1 \leq i \leq k$}{
    \If{$D_i$ does not dominate all neighborhoods $|N(v)| \geq s_i$ or $|D_i| \geq \frac{12 n \log n}{s_i}$}{
        \Return $\APSP(G)$
    }
}

\textcolor{blue}{Phase 1: Compute $X(u), C(u)$ for all $u \in V$}

\For{$u \in V$}{
    \For{$k_0 + 1 \leq i \leq k_0 + \frac{\beta}{2} - 1$}{
        $r_i \gets \frac{\beta}{2} - (i - (k_0 + 1))$
        
        \lIf{$|N(u, r_i)| \geq s_{i}$}{compute $X(u, i)$ (see \Cref{lemma:x-c-computation})}
        
        \lElseIf{$|N(u, r_i) \cap D_{i + 1}| \leq 12 s_{k - 1} \log n$}{compute $C(u, i) \gets N(u, r_i) \cap D_{i + 1}$}
        
        \lElse{\Return $\APSP(G)$}
    }
    
    \lIf{$\deg(u) \leq s_{k_0 + 2}$}{$\hat{\delta}(u, v) \gets \min(\hat{\delta}(u, v), \bfs(G_{k_0 + 1}, u, 2))$ where $G_{k_0 + 1} = (V, E_{k_0 + 1})$} \label{line:short-approx:depth-2-bfs}
}

\lFor{$k_0 + 1 \leq i \leq k_0 + \frac{\beta}{2} - 1$}{$D_{i} \gets D_{i} \cup \hittingSet(\set{X(u, i)})$ for all $u$ such that $|X(u, i)| \geq s_{i}$}
\label{line:short-approx:di-augment}

\textcolor{blue}{Phase 2: Compute Distance Estimates}

\For{$1 \leq j \leq k$}{
    \For{$w \in D_{j}$}{
        $G_{j, w} \gets \left(V, E_j \cup \left( \bigcup_{j + j_1 + j_2 \leq 2k + 1} D_{j_1} \times D_{j_2} \right) \cup E^* \cup (w \times V) \right)$

        $\hat{\delta}(w, v) \gets \min(\hat{\delta}(w, v), \min_{u \in q(w, D_j)} 1 + \hat{\delta}(u, v))$ for all $v \in V$ \label{line:short-approx:truncated-bfs-hitting-set}
        
        $\hat{\delta} \gets \dijkstra(G_{j, w}, w, \hat{\delta})$ 
    }
}

\For{$u, v \in V$}{
    \For{$k_0 + 1 \leq i \leq k_0 + \frac{\beta}{2} - 1$}{
        \lIf{$\delta(u, p(u, D_{i})) \leq r_i$}{$\hat{\delta}(u, v) \gets \min(\hat{\delta}(u, v), r_i + \hat{\delta}(p(u, D_{i}), v))$}
        \label{line:short-approx:x(u)-approx}
        
        \lIf{$|C(u, i)| \leq 12 s_{k - 1} \log n$}{$\hat{\delta}(u, v) \gets \min(\hat{\delta}(u, v), \min_{w \in C(u, i)} r_i + \hat{\delta}(w, v))$}
        \label{line:short-approx:c(u)-approx}
    }
}

\caption{$\shortAdditiveAPASP(G, \beta)$} 
\label{alg:short-additive-apasp}
\end{algorithm}
\DecMargin{1em}

The following lemma bounds the probability that we must compute an exact $\APSP$ solution.

\begin{lemma}
    \label{lemma:short-approx-error-case}
    Let $u \in V$ and $k_0 + 1 \leq i \leq k_0 + \frac{\beta}{2} - 1$
    Define the following events for each $u, i$.
    \begin{align*}
        A(u, i) &= \set{|N(u, r_i)| \geq s_i} \\
        B(u, i) &= \set{|N(u, r_i) \cap D_{i + 1}| \geq 12 s_{k - 1} \log n}
    \end{align*}
    For all $u$, let $E(u, i)$ denote the error event $E(u, i) = \neg A(u, i) \wedge B(u, i)$.

    For each $1 \leq i \leq k$, let $E(i)$ denote the event that either $D_i$ fails to dominate some neighborhood of size at least $s_i$ or $|D_i| \geq \frac{12 n \log n}{s_i}$.
    
    Let $E = (\bigcup_{i = 1}^{k} E(i)) \cup (\bigcup_{i, u} E(u, i))$ denote the event that any error event occurs.
    
    Then, $\Pr(E) = \bigO{\frac{1}{n^2}}$
\end{lemma}

\begin{proof}
    Fix $u \in V$ and $k_0 + 1 \le qi \leq k_0 + \frac{\beta}{2} - 1$,
    Consider the event $E(u, i)$.
    \begin{equation*}
        \Pr(E(u, i)) = \Pr(\neg A(u, i) \wedge B(u, i) \leq \Pr(B(u, i) \mid \neg A(u, i))
    \end{equation*}
    Conditioned on $\neg A(u, i)$, define the random variable $X = |N(u, r_i) \cap D_{i + 1}|$.
    Since $D_{i + 1}$ is constructed by sampling each vertex probability $p = \frac{3 \log n}{s_{i + 1}}$, $X$ is the sum of at most $s_i$ Bernoulli random variables.
    Let $Y$ be the sum of $s_i$ Bernoulli random variables with the same probability $p$ so that $Y$ stochastically dominates $X$.
    Note that $\E[X] \leq \frac{3 \log n s_i}{s_{i + 1}} = 3 s_{k - 1} \log n$.
    By a standard Chernoff bound, 
    \begin{equation*}
        \Pr(X > 12 s_{k - 1} \log n) \leq \Pr(Y > 12 s_{k - 1} \log n) < \exp(-3 s_{k - 1} \log n)
    \end{equation*}

    Thus, we can easily union bound over $O(n)$ events $E(u, i)$ as each occurs with inverse exponential probability.
    Note that $C(u, i)$ is defined from the randomly chosen $D_i$, and these sets are constructed before $D_i$ is augmented by a deterministic hitting set of $X(u, i)$ in Line \ref{line:short-approx:di-augment}.

    From \Cref{lemma:degree-decomposition}, $\rDominate$ with high probability satifies the necessary requirements.
    To see this, consider some threshold $s_i$ and vertex $\deg(v) \geq s_i$.
    The probability $D_i$ fails to dominate $s_i$ with $P$ chosen above is at most,
    \begin{equation*}
        (1 - p)^{s_i} \leq \exp \left( - 3 \log n \right) \leq n^3
    \end{equation*}
    We union bound over all $u, i$ to obtain the probability $\tO{1/n^2}$.

    Second, the probability that $|D_i| \geq \frac{12 n \log n}{s_i}$ is small by an identical argument to the error bound for $E_i$.
    Since we bound over $O(1)$ events, these also occur with negligible probability.
    We conclude the proof by union bounding over all the error events.
\end{proof}

\begin{lemma}
    \label{lemma:short-approx-correctness}
    Suppose that the error event $E$ does not occur.
    Then, Algorithm \ref{alg:short-additive-apasp} outputs $\hat{\delta}$ such that,
    \begin{equation*}  
        \delta(u, v) \leq \hat{\delta}(u, v) \leq \delta(u, v) + \beta
    \end{equation*}
    for all $\delta(u, v) \leq \beta + 1$.
\end{lemma}

We delay the proof of Lemma \ref{lemma:short-approx-correctness}, instead beginning with a few useful lemmas that will be crucial to our proof.
In the following lemma, we prove that a reasonable additive approximation is obtained whenever $A(u, i)$ is true. 
That is, we assume $|N(u, r_i)| \geq s_i$ is large.

\begin{lemma}
    \label{lemma:r_i-neighborhood-large}
    Let $k_0 + 1 \leq i \leq k_0 + \frac{\beta}{2} - 1$.
    Let $u, v$ be a pair of vertices and $P \subset E_{k_0 + 1}$ be a shortest path.
    Suppose $|N(u, r_i)| \geq s_i$ where $r_i = \frac{\beta}{2} - (i - (k_0 + 1))$.
    Then, after Line \ref{line:short-approx:x(u)-approx} we have,
    \begin{equation*}
        \hat{\delta}(u, v) \leq \delta(u, v) + \beta
    \end{equation*}
\end{lemma}

\begin{proof}
    Suppose $|N(u, r_i)| \geq s_i$, then there is some vertex $w \in X(u, i) \cap D_i \subset N(u, r_i) \cap D_i$ so that $\delta(u, p(u, D_i)) \leq r_i$.
    In particular,
    \begin{equation*}
        \hat{\delta}(u, v) \leq r_i + \hat{\delta}(p(u, D_i), v)
    \end{equation*}

    Let us examine the quantity $\hat{\delta}(p(u, D_i), v)$.
    Since after Line \ref{line:short-approx:di-augment} the shortest path in $G$ between $u, p(u, D_i)$ is at most $r_i$, we begin by showing that the path between $u, p(u, D_i)$ exists in the graph $G_{i, u}$.
    If there is an edge missing from $E_i$, let $w$ be the vertex closest to $u$ such that $\deg(w) > s_{i - 1}$.
    Therefore, $N(w) \cap D_{i - 1} \neq \emptyset$ and $N(w) \cap D_{i} \neq \emptyset$.
    If $w$ is not the vertex immediately preceding $p(u, D_i)$, then this violates the property that $p(u, D_i)$ is the nearest vertex in $D_i$ to $u$.
    Otherwise, if $w$ is the vertex immediately preceding $p(u, D_i)$, then the edge $(w, p(u, D_i))$ is accessible in the graph $G_{i, p(u, D_i)}$.
    Thus, the shortest path of length $r_i$ between $p(u, D_i), u$ is available in the $i$-th iteration.

    Consider now the a path $P'$ from $p(u, D_i)$ to $v$ of length at most $|P| + r_i$.
    By our arguments above and $P \subset E_{k_0 + 1}$, the path $P'$ is also in $E_{k_0 + 1}$.
    From Lemma \ref{lemma:sparse-apasp-approx}, $L(P') \subset \set{k_0 + 1, \dotsc, i - 1}$ has size at most $i - (k_0 + 1)$, so we can conclude the total error is at most,
    \begin{equation*}
        \hat{\delta}_{i}(p(u, D_i), v) \leq \delta(u, v) + r_i + 2 (i - (k_0 + 1))
    \end{equation*}

    Then,
    \begin{align*}
        \hat{\delta}(u, v) &\leq r_i + \hat{\delta}(p(u, D_i), v) \\
        &\leq \delta(u, v) + 2 r_i + 2 (i - (k_0 + 1)) \\
        &= \delta(u, v) + 2 \left( r_i + (i - (k_0 + 1)) \right) \\
        &\leq \delta(u, v) + \beta
    \end{align*}
\end{proof}

Next, we prove correctness under the assumption that $B(u)$ is false.
In particular, assume the set $|C(u, i)| = |N(u, r_i) \cap D_{i + 1}| \leq 12 s_{k - 1} \log n$ is small.

\begin{lemma}
    \label{lemma:r_i-neighborhood-small}
    Suppose that the error event $E$ does not occur.

    Let $k_0 + 1 \leq i \leq k_0 + \frac{\beta}{2} - 1$.
    Let $u, v$ be a pair of vertices with a shortest path $P = (u, u_2, u_3, \dotsc, v_3, v_2, v)$ with all edges in $E_{k_0 + 1}$.
    Suppose $|N(u, r_i) \cap D_{i + 1}| \leq 12 s_{k - 1} \log n$.
    If $\level(u_{r_i}) \leq i + 1$ then, after Line \ref{line:short-approx:c(u)-approx} we have,
    \begin{equation*}
        \hat{\delta}(u, v) \leq \hat{\delta}(u, v) + \beta
    \end{equation*}
\end{lemma}

\begin{proof}
    Since $\level(u_{r_i}) \leq i + 1$, there is some vertex $z \in N(u, r_i) \cap D_{i + 1}$ such that $z = r(u_{r_i}, D_{i + 1})$ as $u_{r_i} \in N(u, r_i - 1)$.
    Since $P \subset E_{k_0 + 1}$, then after the $(i + 1)$-th iteration, 
    \begin{equation*}
        \hat{\delta}(z, v) \leq \delta(u_{r_i}, v) + 2 (i - k_0) + 1
    \end{equation*}
    where we apply Lemma \ref{lemma:sparse-apasp-approx} with $L(P) \subset \set{k_0 + 1, \dotsc, i}$ has size at most $i - k_0$.

    Then, after Line \ref{line:short-approx:c(u)-approx}, since $z \in N(u, r_i) \cap D_{i + 1} = C(u, i)$,
    \begin{equation*}
        \hat{\delta}(u, v) \leq r_i + \hat{\delta}(z, v) \leq \delta(u_{r_i}, v) + r_i + 2(i - k_0) + 1
    \end{equation*}

    Now, $\delta(u_{r_i}, v) = \delta(u, v) - r_i + 1$ so that,
    \begin{equation*}
        \hat{\delta}(u, v) \leq \delta(u, v) + 2 (i - k_0 + 1) \leq \delta(u, v) + \beta
    \end{equation*}
    as $i + 1 \leq k_0 + \frac{\beta}{2}$.
\end{proof}

From Lemmas \ref{lemma:r_i-neighborhood-large} and \ref{lemma:r_i-neighborhood-small} we can place an upper bound on the levels of many vertices of the shortest path.
The following lemma shows that we obtain a good additive approximation when these bounds hold.

\begin{lemma}
    Let $P = (u, u_2, \dotsc, v_2, v)$ be a shortest path of length at most $\beta + 1$ such that $P \subset E_{k_0 + 1}$.
    Suppose the degrees $\level(u_{r_i}), \level(v_{r_i}) > i + 1$ for all $k_0 + 1 \leq i \leq k_0 + \frac{\beta}{2} - 1$.

    Then, Algorithm \ref{alg:short-additive-apasp} obtains an approximation such that $\hat{\delta}(u, v) \leq \delta(u, v) + \beta$.
    \label{lemma:k-1-approx-distance}
\end{lemma}

\begin{proof}
    As before, let $\deg(P) = \max_{v \in P} \deg(v)$ be the maximum degree of any vertex in $P$ and $\level(P) = \min_{v \in P} \level(v)$ the minimum level of any vertex in $P$.
    In the following, we will frequently use (some variant of) the following equality,
    \begin{equation}
        k - k_0 = \left(\frac{3}{2} \beta + 1\right) - \left(\beta - 2\right) = \frac{\beta}{2} + 3
        \label{eq:comb-k-k0-diff}
    \end{equation}

    Let us begin by assuming $P$ has length exactly $\beta + 1$ so that we can label the vertices $P = (u, u_2, \dotsc, u_{\beta/2 + 1}, v_{\beta/2 + 1}, \dotsc, v_2, v)$.

    If $P \subset E_{k_0 + 3}$, then $L(P, k) \subset \set{k_0 + 3, \dotsc, k - 1}$ is a set of size at most $k - k_0 - 3 = \frac{\beta}{2}$, so we may apply Lemma \ref{lemma:sparse-apasp-approx} and obtain a $+\beta$-approximation.
    Thus, in the following, assume $P$ has an edge in $E_{k_0 + 1} \setminus E_{k_0 + 3}$.
    In particular, there are at least two vertices $x \in P$ with $\deg(x) > s_{k_0 + 2}$.

    Suppose $|L_B(P, k)| \leq \frac{\beta}{2}$.
    In the $k$-th iteration, as $D_k = V$, we have by \ref{lemma:block-approx-error}, $\hat{\delta}_{k}(u, v) \leq \delta(u, v) + \beta$.
    Thus, we can assume $|L_B(P, k)| > \frac{\beta}{2}$.

    By assumption, the levels of vertices $u_{r_i}, v_{r_i}$ are restricted.
    Recall that $r_i = \frac{\beta}{2} - (i - (k_0 + 1))$ so that,
    \begin{align*}
        r_{k_0 + 1} &= \frac{\beta}{2} - ((k_0 + 1) - (k_0 + 1)) = \frac{\beta}{2} \\
        r_{k_0 + \frac{\beta}{2} - 1} &= \frac{\beta}{2} - ((k_0 + \frac{\beta}{2} - 1) - (k_0 + 1)) = 2
    \end{align*}

    In particular, other than the two middle vertices, $u_{\frac{\beta}{2} + 1}, v_{\frac{\beta}{2} + 1}$, the remaining vertices have level at least $i + 2 \geq k_0 + 3$, so that we can assume both central vertices have degree at least $s_{k_0 + 2}$.
    Furthermore, at most one of these vertices can be a blocking vertex.
    Suppose for contradiction that both are, so that one, for example $u_{\frac{\beta}{2} + 1}$, has level $k_0 + 1$ and the other level $k_0 + 2$.
    However, in this case $P \subset E_{k_0 + 2}$ so that $u_{\frac{\beta}{2} + 1}$ cannot be a blocking vertex.
    In particular, if $L(P, k)$ does not contain all levels $\set{k_0 + 3, \dotsc, k - 1}$ then we can already obtain a $\beta$ approximation by Lemma \ref{lemma:sparse-apasp-approx}.

    Therefore, let $L_B(P)$ denote the blocking vertices of $P$.
    Since $k_0 + 3 \in L_B(P)$, we have either $u_{\beta/2}, v_{\beta/2} \in B(P)$.
    Without loss of generality, suppose $u_{\beta/2} \in B(P)$.
    
    In Line \ref{line:short-approx:depth-2-bfs}, since $\deg(u_{\beta/2}) < s_{k_0 + 2}$ and $P \subset E_{k_0 + 1}$, we compute $\hat{\delta}(u_{\beta/2}, v_{\beta/2 + 1}) = 2$.
    Then, if $u_{\beta/2}^* = r(u_{\beta/2}, D_{k_0 + 3})$, in Line \ref{line:short-approx:truncated-bfs-hitting-set}, we have $\hat{\delta}(u_{\beta/2}^*, v_{\beta/2 + 1}) \leq 3$.
    We claim that,
    \begin{equation*}
        \hat{\delta}_{k_0 + 3}(u_{\beta/2}^*, x) \leq \delta(u_{\beta/2}, x) + 1
    \end{equation*}
    for all $x \in P$.
    We have shown this already for $x \in \set{u_{\beta/2 + 1}, v_{\beta/2 + 1}}$.
    For $x \in \set{u, u_2, \dotsc, u_{\beta/2}}$, this follows simply from the assumption that all these vertices have level at least $k_0 + 3$ so that $P_{w, x} \subset E_{k_0 + 3}$.
    For $x \in \set{v_{\beta/2}, \dotsc, v_2, v}$, when executing $\dijkstra$ we take the edge $(u_{\beta/2}^*, v_{\beta/2 + 1})$ and the remaining edges of $P$ in $E_{k_0 + 3}$.

    This allows us to make the following claim stronger than Lemma \ref{lemma:block-approx-error}.
    For any $w, x \in P$ where $\level(w) \geq k_0 + 3$,
    \begin{equation}
        \hat{\delta}_{\level(w)}(w^*, x) \leq \delta(w, x) + 2 |L_B(P_{w, x})| - 1
        \label{eq:short-approx-tight-bound}
    \end{equation}
    where we have shown the base case for $\level(u_{\beta/2}) = k_0 + 3$.
    The inductive step will follow similarly.
    Let $j \geq k_0 + 3$ and $\level(w) = j$.
    Let $y = b(w, P_{w, x})$ be the blocking vertex so that $\level(y) < \level(w)$.
    Then, 
    \begin{align*}
        \hat{\delta}_{\level(w)}(w^*, x) &\leq \hat{\delta}_{\level(y)}(w^*, y^*) + 1 + \delta(y, x) \\
        &\leq \delta(w, y) + 2 |L_B(P_{y, w})| + 1 + \delta(y, x) \\
        &\leq \delta(w, x) + 2 |L_B(P_{w, x})| - 1
    \end{align*}
    where we have used $|L_B(P_{w, x})| > |L_B(P_{y, w})|$ as $\level(y) \in L_B(P_{w, x}) \setminus L_B(P_{y, w})$.
    Then, in the $k$-th iteration, let $w$ be the closest vertex to $v$ such that $\deg(w) \geq s_{k - 1}$ so that,
    \begin{equation*}
        \hat{\delta}(u, v) \leq \hat{\delta}_{k - 1}(u, w^*) + 1 + \delta(w, v) \leq \delta(u, v) + 2 |L_B(P_{w, u})| \leq \delta(u, v) + \beta
    \end{equation*}
    as $L_B(P_{w, u}) = \set{k_0 + 3, \dotsc, k - 2} \cup \set{\min(\level(u_{\beta/2 + 1}), \level(v_{\beta/2 + 1})}$ has size $\frac{\beta}{2}$.

    To conclude, if $P$ has length less than $\beta$, then $P \subset E_{k_0 + 3}$, so we obtain a $\beta$ approximation by Lemma \ref{lemma:sparse-apasp-approx}.
\end{proof}

Finally, we claim that all the estimates produced by Algorithm \ref{alg:short-additive-apasp} are feasible.

\begin{lemma}
    Suppose error event $E$ does not occur.
    Then, Algorithm \ref{alg:short-additive-apasp} obtains an approximation such that $\delta(u, v) \leq \hat{\delta}(u, v)$ for all $u, v \in V$.
    \label{lemma:short-approx-feasible}
\end{lemma}

\begin{proof}
    We analyze each instance where the distance estimates are updated.
    In Phase 0, we initialize the distance estimates to the adjacency matrix.
    In Phase 1, we compute length 2 paths in a sub-graph $G_{k_0 + 1} \subset G$.
    Next, the distance estimates are updated in Phase 2 in computing $\dijkstra(G_{j, w}, w, \hat{\delta})$ for all $w \in D_j$.
    Note that this is identical to Algorithm $\denseAPASP$, and produces feasible distance estimates.
    
    Line \ref{line:2-approx:truncated-bfs-hitting-set} does not violate feasibility as if $u \in q(w, D_j)$ then $(u, w) \in E$ so that,
    \begin{equation*}
        \hat{\delta}(w, v) \geq 1 + \hat{\delta}(u, v) \geq 1 + \delta(u, v) \geq \delta(w, v)
    \end{equation*}
    by the triangle inequality, where we have also used that the previous estimates were feasible.
    
    Finally, since $p(u, D_i) \in N(u, r_i)$ and $C(u, i) \subset N(u, r_i)$, the final distance estimates are feasible as for any $z \in N(u, r_i)$, the updated distance estimate,
    \begin{equation*}
        \hat{\delta}(u, v) = r_i + \hat{\delta}(z, v) \geq r_i + \delta(z, v) \geq \delta(u, v)
    \end{equation*}
    where in the first inequality we used that the previous estimates are feasible and in the second we used the triangle inequality.
\end{proof}

\subsubsection*{Proof of Theorem \ref{thm:short-additive-approx}}

To prove the correctness of Algorithm \ref{alg:short-additive-apasp}, it suffices to prove Lemma \ref{lemma:short-approx-correctness}.
Indeed, whenever $E$ occurs, we compute an exact $\APSP$ solution, which is obviously correct.
If $E$ does not occur, Lemma \ref{lemma:short-approx-feasible} guarantees that each estimate is feasible, and is produced by some path in $G$.

\begin{proof}
    Let $u, v$ be a pair of vertices with shortest path $P$ length at most $\beta + 1$.
    We proceed by case analysis.

    \paragraph{Case 1: $P \subset E_{k_0 + 1}$}
    
    Suppose the assumptions of Lemma \ref{lemma:k-1-approx-distance} are not met, otherwise we obtain an accurate approximation.
    That is, for some $i$, $\level(u_{r_i}) \leq i + 1$.
    In particular, $u^* = r(u, D_{i + 1})$ must be in $N(u, r_i) \cap D_{i + 1}$.
    Then, if $|N(u, r_i) \cap D_{i + 1}| \leq 12 s_{k - 1} \log n$, we obtain a correct estimate by Lemma \ref{lemma:r_i-neighborhood-small}.
    Otherwise, since $E(u, i)$ does not occur, $|N(u, i)| \geq s_i$, and we obtain a correct estimate by Lemma \ref{lemma:r_i-neighborhood-large}.

    \paragraph{Case 2: $P \not\subset E_{k_0 + 1}$}

    We begin with the special case $\beta = 4$, generalizing to $\beta \geq 6$ later.

    \paragraph{Warm up: $\beta = 4$}

    Note $k = 7, k_0 = 2$.
    Let $z$ be the vertex of maximum degree.
    Let $w = b(u, P)$ be the vertex closest to $v$ such that $\deg(w) \geq s_{k_0 + 4} = s_{6}$.
    Since $P = P_{u, w} \cdot P_{w, v}$, the sub-path $P_{u, w}$ must contain the vertex $z$.
    Let $z^* = r(z, D_{\level(z)}), w^* = r(w, D_{k_0 + 4})$.
    Since $P \subset E_{\level(z)}$, we have $\hat{\delta}_{\level(z)}(z^*, u) \leq \delta(z, u) + 1$ and $\hat{\delta}_{\level(z)}(z^*, w^*) \leq \delta(z, w) + 2$.
    Furthermore, since $\level(z) + (k_0 + 4) + (k_0 + 5) \leq 3 k_0 + 9 = 2k + 1$, in the $k_0 + 5$-th iteration,
    \begin{equation*}
        \hat{\delta}(u, v) \leq \hat{\delta}_{\level(z)}(u, z^*) + \hat{\delta}_{\level(z)}(z^*, w^*) + 1 + \delta(w, v) \leq \delta(u, v) + 4
    \end{equation*}

    as desired.

    \paragraph{Generalization to $\beta \geq 6$}

    Suppose $\level(v) \leq \level(u)$.
    Recall from Lemma \ref{lemma:block-approx-error} the blocking vertices $B(P) = \set{x_0, x_1, \dotsc, x_t}$ and levels $L_B(P)$ of path $P$.
    
    Let $a = \min_{\level(x_i) \leq k_0 + 5} i$ be the minimum index of an element in the blocking set such that $\level(x_i) \leq k_0 + 5$.
    Since $k \geq \level(x_1) > \level(x_2) > \dotsc > \level(x_t) \geq 1$, we can upper bound,
    \begin{equation*}
        a \leq k - (k_0 + 5) = \frac{\beta}{2} - 2
    \end{equation*}
    Since $P$ is not contained in $E_{k_0 + 1}$, we can assume $a$ exists and $a \geq 1$.
    
    Let $x_a \in B(P)$ be the corresponding vertex in $B(P)$.
    Since $P$ has an edge not in $E_{k_0 + 1}$, the last blocking vertex of minimum level must have $\level(x_t) \leq k_0$.
    Let $v^* = r\left(v, D_{\level(v)}\right)$ for any vertex $v$.
    Since $P \subset E_{\level(x_t)}$, we again have $\hat{\delta}_{\level(x_t)}(x_t^*, x_a^*) \leq \delta(x_t, x_a) + 2$ and $\hat{\delta}_{\level(x_t)}(x_t^*, x_{a+1}^*) \leq \delta(x_t, x_{a + 1}) + 2$.
    
    Consider the $\level(x_a)$-th iteration.
    The edges $D_{\level(x_t)} \times D_{\level(x_{a + 1})}$ are in $G_{\level(x_a), x_a}$ as,
    \begin{equation*}
        \level(z) + \level(x_{a + 1}) + \level(x_{a}) \leq k_0 + (k_0 + 4) + (k_0 + 5) = 3 k_0 + 9 = 3 \left( k - \frac{\beta}{2} - 3 \right) + 9 = 2 k + 1
    \end{equation*}
    Then, since $x_t \in P_{x_a, x_{a + 1}}$,
    \begin{align*}
        \hat{\delta}_{\level(x_a)}(x_a^*, x_{a - 1}) &\leq \hat{\delta}_{\level(x_t)}(x_a^*, x_t^*) + \hat{\delta}_{\level(x_t)}(x_t^*, x_{a + 1}^*) + 1 + \delta(x_{a + 1}, x_{a - 1}) \\ 
        &\leq \delta(x_a, x_t) + \delta(x_t, x_{a + 1}) + \delta(x_{a + 1}, x_{a - 1}) + 5 \\ 
        &\leq \delta(x_a, x_{a - 1}) + 5
    \end{align*}

    Then, following a similar argument to the inductive step of Lemma \ref{lemma:block-approx-error}, we claim the following for all $1 \leq j \leq a$.
    \begin{equation*}
        \hat{\delta}_{\level(x_j)}(x_j^*, x_{j - 1}) \leq \delta(x_j, x_{j - 1}) + 2(2 + (a - j)) + 1
    \end{equation*}
    where we have established the base case $j = a$ above.
    We now proceed by induction for $j < a$.
    Consider an execution of $\dijkstra$ from $x_j^*$ in $G_{\level(x_j), x_j^*}$.
    Let $x_{j + 1}$ be the blocking vertex from the previous iteration.
    We take the edges $(x_j^*, x_{j + 1}^*)$, $(x_{j + 1}^*, x_{j + 1}) \in E^*$, and the remaining edges in $E_{\level(x_j)}$.
    By induction, we have,
    \begin{align*}
        \hat{\delta}_{\level(x_j)} (x_j^*, x_{j - 1}) &\leq \hat{\delta}_{\level(x_{j + 1})} (x_j^*, x_{j + 1}^*) + 1 + \delta(x_{j + 1}, x_{j - 1}) \\
        &\leq \delta(x_j, x_{j + 1}) +  2 (2 + (a - (j + 1))) + 3 + \delta(x_{j + 1}, x_{j - 1}) \\
        &\leq \delta(x_{j}, x_{j - 1}) + 2 (2 + (a - j)) + 1
    \end{align*}

    Thus, we have,
    \begin{equation*}
        \hat{\delta}_{\level(v)}(v^*, u) = \hat{\delta}_{\level(x_1)} (x_1^*, x_0) \leq \delta(v, u) + 2 (a + 1)  + 1
    \end{equation*}
    
    From $u$, we take the edge $(u, v^*)$, followed by $(v^*, v) \in E^*$ so that,
    \begin{align*}
        \hat{\delta}(u, v) &\leq \hat{\delta}_{\level(v)}(u, v^*) + 1 \\
        &\leq \delta(u, v) + 2 (a + 1) + 2 \\
        &\leq \delta(u, v) + \beta
    \end{align*}
\end{proof}

Thus, to prove Theorem \ref{thm:short-additive-approx}, it only remains to analyze the performance of Algorithm \ref{alg:short-additive-apasp}.

\begin{proof}
    (Time Complexity)

    We separately analyze the complexity of each phase.
    By Lemma \ref{lemma:short-approx-error-case}, $E$ occurs with probability $\bigO{\frac{1}{n^2}}$.
    Since in this case Algorithm \ref{alg:short-additive-apasp} requires $\tO{n^3}$ time, this contributes at most $\tO{n}$ to the expected run-time.
    Thus, in the following assume that $E$ does not occur.

    \paragraph{Phase 0}

    $\rDecompose$ requires $O(m + n) = O(n^2)$ time to sample each vertex set randomly and construct the edge sets $E_i$.
    Verifying that $D_i$ is of the appropriate size and dominates $|N(v)| \geq s_i$ requires time $\tO{n^2}$.

    \paragraph{Phase 1}

    By \Cref{lemma:x-c-computation}, computing the sets $X(u, i), C(u, i)$ require time at most $\bigtO{n^{2 + \frac{1}{k}}}$.
    By Lemma \ref{lemma:hitting-set}, augmenting $D_i$ in Line \ref{line:short-approx:di-augment} requires time $\tO{n^2}$.

    Line \ref{line:short-approx:depth-2-bfs} requires time,
    \begin{equation*}
        \bigtO{n s_{k_0 + 2} s_{k_0}} = \bigtO{n^{3 - \frac{2 k_0 + 2}{k}}} = \bigtO{n^{2 + \frac{1}{k}}}
    \end{equation*}
    whenever $\beta \geq 4$ as shown in \Cref{lemma:x-c-computation}.
    
    \paragraph{Phase 2}

    First, we bound the complexity of each call to $\dijkstra$ as in the proof of Lemma \ref{lemma:dhz-apasp} for $\denseAPASP$.
    Fix some $j$.
    Then, $|E_j| \leq n s_{j - 1} = n^{2 - \frac{j - 1}{k}}$ and $|D_{j_1} \times D_{j_2}| = \bigtO{\frac{n}{s_{j_1}} \times \frac{n}{s_{j_2}}} = \bigtO{n^{\frac{j_1 + j_2}{k}}} = \bigtO{n^{2 - \frac{j - 1}{k}}}$ as $j_1 + j_2 \leq 2k + 1 - j$.
    Note for any fixed $w$ that $(w \times V) \cup E^*$ has size $\tO{n}$.
    Summing over $|D_j| = \bigtO{\frac{n}{s_{j}}} = \bigtO{n^{\frac{j}{k}}}$, all invocations of $\dijkstra$ require time $\bigtO{n^{2 + \frac{1}{k}}}$.

    Line \ref{line:short-approx:truncated-bfs-hitting-set} requires time $\tO{n^2}$ as every vertex $u$ is in at most one $q(w, D_j)$ for any fixed $j$.

    Finally, Line \ref{line:short-approx:x(u)-approx} requires time $\tO{n^2}$ over all $u, v, i$.
    By the bound on $|C(u, i)| = \tO{n^{1/k}}$, Line \ref{line:short-approx:c(u)-approx} requires time $\bigtO{n^{2 + \frac{1}{k}}}$.
    In particular, this is the overall time bound when $E$ does not occur, giving the desired result by substituting $k = \frac{3 \beta}{2} + 1$.
\end{proof}

We have the following Corollary.

Recall the following result of Roditty \cite{roditty2023newapasp}.

\begin{lemma}[Theorem 5.1 of Roditty \cite{roditty2023newapasp}]
Let $\multapproxlimit' \geq 6$ be an even integer.
There is an algorithm that computes an additive $\multapproxlimit' - 2$ approximation in $\tO{n^{2 - \frac{2}{\multapproxlimit' + 2}} m^{\frac{2}{\multapproxlimit' + 2}}}$ expected time, for every $u, v \in V$ for which $\delta(u, v) \leq \multapproxlimit'$.
\label{lemma:roditty-short-sparse}
\end{lemma}

\begin{restatable}{corollary}{twoapproxgeqk}
    \label{cor:2-approx-d-geq-k}
    Let $\multapproxlimit \geq 4$ be an even integer.
    Then, we can compute a $(2, 0)$-approximation for distances $\delta(u, v) \geq \multapproxlimit$ combinatorially in expected time
    \begin{equation*}
        \bigtO{\min \left(n^{2 - \frac{2}{\multapproxlimit + 4}} m^{\frac{2}{\multapproxlimit + 4}}, n^{2 + \frac{1}{2(\multapproxlimit - 1)}}, n^{2 + \frac{2}{3\multapproxlimit + 2}}\right)}
    \end{equation*}
    In particular, we output $\hat{\delta}$ such that $\delta(u, v) \leq \hat{\delta}(u, v)$ for all $u, v$ and $\hat{\delta}(u, v) \leq 2 \delta(u, v)$ whenever $\delta(u, v) \geq \multapproxlimit$.
\end{restatable}

\begin{proof}
    We combine \Cref{thm:mult-approx-bk}, \Cref{thm:short-additive-approx}, \Cref{lemma:dhz-apasp} and \Cref{lemma:roditty-short-sparse}.
    If $n^{2 - \frac{2}{\multapproxlimit + 4}} m^{\frac{2}{\multapproxlimit + 4}}$ is the minimum term, we invoke the algorithm of \Cref{lemma:roditty-short-sparse} with $\multapproxlimit' = \multapproxlimit + 2$ and the algorithm $\sparseAPASP$ of $\Cref{lemma:dhz-apasp}$ with $\beta = \multapproxlimit + 2$.
    For paths of length $\multapproxlimit \leq \delta(u, v) \leq \multapproxlimit + 2$, we obtain a $+\multapproxlimit$ approximation from \Cref{lemma:roditty-short-sparse}.
    For paths of length $\delta(u, v) \geq \multapproxlimit + 2$, we obtain a $+(\multapproxlimit+2)$ approximation from $\sparseAPASP$.
    In either case, we obtain a $(2, 0)$ approximation for all $\delta(u, v) \geq \multapproxlimit$.
    The overall running time is,
    \begin{equation*}
        \tO{n^{2 - \frac{2}{\multapproxlimit + 4}} m^{\frac{2}{\multapproxlimit + 4}}}
    \end{equation*}

    On the other hand, if $n^{2 + \frac{1}{2(k - 1)}}$ is the minimum term, we invoke \Cref{alg:mult-approx-bk} and obtain a $(2, 0)$ approximation for paths of length at least $k$ following \Cref{thm:mult-approx-bk}.

    Finally, if $n^{2 + \frac{2}{3\multapproxlimit + 2}}$ is the minimum term, we call \Cref{thm:short-additive-approx} with $\multapproxlimit' = \multapproxlimit + 1$ and $\denseAPASP$ with $\beta = \multapproxlimit + 2$.
    For any path of length at least $\multapproxlimit + 2$, the estimate from $\denseAPASP$ is a $(2, 0)$-approximation.
    For paths of length $\multapproxlimit \leq \delta(u, v) \leq \multapproxlimit + 1$, the estimate from $\shortAdditiveAPASP$ is a $(2, 0)$-approximation.
    The overall computation time can be bounded by,
    \begin{equation*}
        \bigtO{n^{2 + \frac{2}{3 \multapproxlimit + 2}} + n^{2 + \frac{2}{3 (\multapproxlimit + 2) - 2}}} = \bigtO{n^{2 + \frac{2}{3 \multapproxlimit + 2}}}
    \end{equation*}
\end{proof}

Note that for $k = 6$, $n^{2 + \frac{1}{2(k - 1)}} = n^{21/10} = n^{2 + \frac{2}{3k - 2}}$.
For larger $k$, the former term is smaller and we apply \Cref{thm:mult-approx-bk}. However, for $k=4,5$, 
\Cref{alg:short-additive-apasp} gives the best bound (see Table \ref{tbl:2-approx-geq-beta}).

\end{document}